\documentclass[11pt,letterpaper]{article}
\usepackage{amsmath,amsfonts,amsthm,amssymb}
\usepackage{setspace}
\usepackage{fancyhdr}
\usepackage{lastpage}
\usepackage{extramarks}
\usepackage{xspace}
\usepackage{chngpage}
\usepackage[ruled,linesnumbered,vlined]{algorithm2e}
\usepackage{soul,color}
\usepackage{graphicx,float,wrapfig}
\usepackage[font=small,labelfont=bf]{caption}
\usepackage[margin=1in]{geometry}
\usepackage{enumitem}
\usepackage{thmtools,thm-restate}

\setlist{nosep}

\allowdisplaybreaks

\usepackage{mdframed}
\newtheorem*{mdresult}{Result}

\newcommand{\opt}{\mathsf{Opt}}
\newcommand{\optafter}[1]{\opt_{#1}}
\newcommand{\alg}{\mathsf{Alg}}
\newcommand{\algafter}[1]{\alg_{#1}}
\newcommand{\algafterprime}[1]{\alg'_{#1}}
\newcommand{\newthreshold}{\ell'}
\newcommand{\lcb}{\ensuremath{\ell}}
\newcommand{\ucb}{\ensuremath{u}}
\newcommand{\D}{\mathcal{D}}

\newcommand{\clever}{$\mathsf{MoveBound}$\xspace}
\newcommand{\smart}{$\mathsf{SwapTest}$\xspace}
\newcommand{\E}{\mathbb{E}}
\newcommand{\OTild}{\widetilde{O}}
\newcommand{\one}{\mathbf{1}}%
\newcommand{\Ex}[2][]{\mbox{\rm\bf E}_{#1}\left[#2\right]}%
\renewcommand{\Pr}[2][]{\mbox{\rm\bf Pr}_{#1}\left[#2\right]}%
\newcommand{\btau}{\pmb{\tau}}

\usepackage[usenames,dvipsnames]{xcolor}
\usepackage[linktocpage=true,breaklinks,colorlinks,citecolor=blue,linkcolor=BrickRed]{hyperref}
\usepackage{cleveref}
\newcommand{\IGNORE}[1]{}
\newcommand{\poly}{\ensuremath{\mathsf{poly}}}
\newcounter{note}[section]

\newcommand{\newadded}[1]{{#1}}

\title{\textmd{\bf Bandit Algorithms for Prophet Inequality and Pandora's Box}}
\date{\today}

\author{ 
    Khashayar Gatmiry\thanks{(gatmiry@mit.edu)
    Electrical Engineering and Computer Science,
Massachusetts Institute of Technology}
    \and 
    Thomas Kesselheim\thanks{(thomas.kesselheim@uni-bonn.de)
     Institute of Computer Science and Lamarr Institute for Machine Learning and Artificial Intelligence, University of Bonn}
	\and Sahil Singla\thanks{
        (ssingla@gatech.edu)
        School of Computer Science,
        Georgia Tech. \newadded{Supported in part by NSF award CCF-2327010.}
        }
    \and
    Yifan Wang\thanks{
        (ywang3782@gatech.edu)
        School of Computer Science,
        Georgia Tech. \newadded{Supported in part by NSF award CCF-2327010.}
        }
}

\newtheorem{Theorem}{Theorem}[section]
\newtheorem{Lemma}[Theorem]{Lemma}
\newtheorem{Remark}[Theorem]{Remark}

\newtheorem{Definition}[Theorem]{Definition}

\newtheorem{Claim}[Theorem]{Claim}
\newtheorem{Corollary}[Theorem]{Corollary}

\begin{document}



\maketitle 
\thispagestyle{empty}

\begin{abstract}
\medskip

The Prophet Inequality and Pandora's Box problems are fundamental stochastic problem with applications in Mechanism Design, Online Algorithms, Stochastic Optimization, Optimal Stopping, and Operations Research. A usual assumption in these works  is that  the probability distributions of the $n$  underlying random variables are given as input to the algorithm.  Since in practice these distributions need to be learned under limited feedback, we initiate the study of such stochastic problems in the  Multi-Armed Bandits model. 

\smallskip 

In  the Multi-Armed Bandits model we interact with $n$ unknown distributions over $T$ rounds: in  round $t$ we play a policy $x^{(t)}$ and only receive the value of $x^{(t)}$ as feedback. The goal is to minimize the regret, which is the difference over $T$ rounds in the total value of the optimal algorithm that knows the  distributions vs. the total value of our algorithm that learns the distributions from the limited feedback. Our main results give near-optimal   $\widetilde{O}\big(\poly(n)\sqrt{T}\big)$ total regret algorithms for both Prophet Inequality and Pandora's Box. 

\smallskip 

Our proofs proceed by maintaining confidence intervals on the  unknown  indices of the optimal policy. The exploration-exploitation tradeoff prevents us from directly refining these confidence intervals, so the main technique is to design a regret upper bound function that is learnable while playing low-regret Bandit policies.

\end{abstract}
\bigskip 

 
\newpage

\thispagestyle{empty}
\setcounter{tocdepth}{2}
{\small
\begin{spacing}{0}
   \tableofcontents
\end{spacing}
}

\clearpage

\setcounter{page}{1}



\section{Introduction}

The field of Stochastic Optimization deals with optimization problems under uncertain inputs, and has had tremendous success since \cite{Bellman-Book57}. A standard model is that the inputs are random variables  that are drawn from {known} probability distributions. The goal is to design a policy (an adaptive algorithm) to optimize the expected objective function. 
Examples of such problems 
 include Prophet Inequality~\cite{HKS-AAAI07,CHMS-STOC10,KW-STOC12,Rubinstein-STOC16}, Pandora’s Box~\cite{KWW-EC16,Singla-SODA18,GKS-SODA19}, 
and  Auction Design~\cite{Hartline-Book22, RoughgardenTwenty-Book16}. 
Most prior works assume that the underlying distributions are known to the algorithm and the challenge 
is in computing an (approximately) optimal policy.
However, in practical applications, the distributions are typically \emph{unknown} and must be learned concurrently with decision-making. 

A foundational framework that examines stochastic problems with  unknown distributions is the stochastic online learning model; see books~\cite{CL-Book06,BubeckC-Book12,Hazan-Book16}. Here, the learner interacts with the environment for $T$ days. On each day $t \in [T]$, the learner plays a certain policy $a^{(t)} \in A$, where $A$ represents the set of all policies (actions/algorithms). The environment draws a sample $X^{(t)} \sim \D$, where $\D$ indicates the environment's \emph{unknown} underlying distribution, and then the learner  receives a reward $a^{(t)}(X^{(t)})$ along with some ``feedback''. \newadded{For a maximization problem, }the  goal of the online learning model is to approach the optimal policy with  reward $\opt:= \max_{a \in A} \E_{X \sim \D}[a(X)]$ 
while minimizing in expectation the total regret:
\[ \textstyle T \cdot \opt  - \sum_{t \in [T]} a^{(t)}(X^{(t)}) .\]
The best regret bound that can be achieved for an online learning problem highly depends on the feedback given to the algorithm.
In the full-feedback model, the learner observes the complete sample $X^{(t)}$ as daily feedback. Since accessing the entire sample $X^{(t)}$ is often not feasible in many real-world applications, several partial feedback models have  been considered. The most limiting of them is the bandit feedback model 
where the only feedback available is the reward $a^{(t)}(X^{(t)})$; see books \cite{Slivkins-Book19,LS-Book20}. 


Interestingly, in many online learning scenarios, limiting feedback does not excessively impair the regret bound. For instance, consider the classic Learning from Experts problem where the goal is to identify the optimal action. In this case, for a small action set, both full feedback and bandit feedback result in an optimal regret bound of $\Theta(\sqrt{T})$.
This  motivates us to address the following question for general online stochastic optimization problems:
\begin{quote}
\emph{What is the minimum amount of feedback  necessary to learn a stochastic optimization problem while maintaining a near-optimal regret bound in $T$ as the full feedback model?}
\end{quote}

In addition to being an intellectually intriguing question, there are several other motivations for designing low regret algorithms that operate with limited feedback.


\begin{itemize}[leftmargin=*]
\item In numerous real-world scenarios, accessing the complete sample $X^{(t)}$ as feedback is infeasible. Furthermore, in order to safeguard data privacy to the greatest extent possible, it is advantageous to utilize minimal information in real-world online learning tasks. 
\item An online learning algorithm that operates with less feedback is concurrently applicable to all partial feedback models that incorporate the required feedback. We can therefore obtain near-optimal online learning algorithms that function uniformly across different feedback models.
\end{itemize}
Specifically, in this paper, we address the above question in the context of the fundamental Prophet Inequality and Pandora's Box problems, which have wide-ranging applications in areas such as Mechanism Design, Online Algorithms, Microeconomics, Operations Research, and Optimal Stopping. Our main results imply near-optimal $\widetilde{O}\big(\poly(n)\sqrt{T}\big)$ regret algorithms for both these problems under most limited bandit feedback, where $\widetilde{O}(\cdot)$ hides logarithmic factors.

\subsection{Prophet Inequality  under Bandit Feedback}

In the classical Optimal Stopping  problem of Prophet Inequality~\cite{Krengel-Journal77,Krengel-Journal78,Samuel-Annals84}, we are given distributions $\D_1, \ldots, \D_n$ of $n$ independent random variables. The outcomes $X_i \sim \D_i$  for $i \in [n]$ are revealed one-by-one and we have to immediately select/discard $X_i$ with the goal of maximizing the selected random variable in expectation.  They have become popular in Algorithmic Game Theory in the last 15 years since they imply posted pricing mechanisms that are ``simple'' (and hence more practical) and approximately optimal; see related work in \Cref{sec:related}.

The optimal policy for Prophet Inequality is given by a simple (reverse) dynamic program:  always select  $X_n$ on reaching it and select $X_i$ for $i<n$ if its value is more than the expected value of this optimal policy on $X_{i+1}, \ldots, X_n$. Thus, the optimal policy with expected value $\opt$ can be thought of as a \emph{fixed-threshold} policy where we select $X_i$ iff $X_i > \tau_i$ for $\tau_i$ being the expected value of this policy after $i$. How to design this optimal policy for unknown distributions? (See \Cref{remark:hindsightOPT} on the  ``hindsight optimum'' benchmark.)




As a motivating example, consider a  scenario where you want to sell a perishable item (e.g., cheese) in the market each day for the entire year. For simplicity, assume that there are 8 buyers, one arriving in each hour between 9 am to 5 pm. Your goal is to set price thresholds for each hour to maximize the total value. 
If the buyer value distributions are known, this  can be modeled as a Prophet Inequality problem with $n=8$ distributions. However, for unknown  value distributions  this  becomes a repeated game with a fixed arrival order where on each day you play some price thresholds and obtain a value along with feedback.  Next, we formally describe  this repeated game.

\medskip
\noindent \textbf{Online Learning Prophet Inequality.} 
In this problem the distributions $\D_1, \ldots, \D_n$ of Prophet Inequality are unknown to the algorithm in the beginning. We make the standard  normalization assumption  that each $\D_i$ is supported on $[0,1]$. \newadded{Without this normalization, a non-trivial additive regret is not achievable.} Now we play a $T$ rounds repeated game\footnote{We will always assume $T \geq n$  since otherwise getting an $O(\poly(n))$ regret algorithm is trivial.}: in round $t\in [T]$ we play a policy, which is a set of $n$ thresholds $(\tau^{(t)}_1, \ldots, \tau^{(t)}_n)$, and receives as reward its value on freshly drawn independent random variables $X^{(t)}_1 \sim \D_1, \ldots, X^{(t)}_n\sim \D_n$, i.e., the reward is $X^{(t)}_{\alg(t)}$ where $\alg(t) \in [n]$ is the smallest index $i$ with $X^{(t)}_{i} > \tau^{(t)}_{i}$.  The goal is to minimize the  \emph{total regret}: 
\[ \textstyle T \cdot \opt - \Ex{\sum_{t=1}^T X^{(t)}_{\alg(t)}} .\]
Since per-round reward is bounded by $1$, the goal is to get $o(T)$ regret. Moreover, standard examples show that every algorithm incurs $\Omega(\sqrt{T})$ regret;  see \Cref{sec:lowerBounds}.

An important question is what amount of feedback the algorithm receives after a round. One might consider a \emph{full-feedback} setting, where after each round $t$ the algorithm gets to know the entire sample $X^{(t)}_1, \ldots, X^{(t)}_n$ as feedback, which could be used to update beliefs regarding the distributions $\D_1, \ldots, \D_n$.
Here it is easy to design an $\widetilde{O}\big(\poly(n)\sqrt{T}\big)$ regret algorithm. This is because after discretization, we may assume that the there are only $T$  candidate thresholds for each $X_i$, so there are only $T^{n}$ candidate policies. Now the classical multiplicative weights algorithm~\cite{AHK-TOC12} implies that the regret is  $O\big(\sqrt{T \log (\# \text{policies})}\big) = \widetilde{O}\big(\poly(n)\sqrt{T}\big)$. Although this na\"ive algorithm is not polytime, a recent work of \cite{GHTZ-COLT21}  \newadded{on $O(n/\epsilon^2)$ sample complexity for prophet inequality} can be interpreted as giving a polytime  $\widetilde{O}(\poly(n)\sqrt{T})$ regret algorithm under {full-feedback}\footnote{Their results are in the PAC model for   ``strongly monotone'' stochastic problems.  They immediately imply $\widetilde{O}(\sqrt{nT})$ regret  under full-feedback  using the standard  doubling-trick.}. These results, however, do not extend to \emph{bandit feedback}, where the algorithm does not see the entire sample. 

\medskip
\noindent \textbf{Bandit Feedback.} {In many applications, it is unreasonable to assume that the algorithm gets the entire sample $X^{(t)}_1, \ldots, X^{(t)}_n$. 
For instance, in  the above scenario of selling a perishable item, we may only see the winning bid 
(e.g., if you don't run the shop and delegate someone else to sell the item  at the given price thresholds).  There are several reasonable partial feedback models,  namely:} 
\begin{enumerate}[label=(\alph*)]  
\item {We see $X^{(t)}_1, \ldots, X^{(t)}_{\alg(t)}$ but not $X^{(t)}_{\alg(t)+1}, \ldots, X^{(t)}_n$, meaning that we do not observe the sequence after it has been stopped.} \label{enum:modela}
\item We see the index $\alg(t)$ and the value $X_{\alg(t)}$ that we select but no other $X_i$. \label{enum:modelb}
\item We only see the value of $X_{\alg(t)}$ that we select and not even the index $\alg(t)$. \label{enum:modelc}
\end{enumerate}
What is the least amount of feedback needed to obtain $\widetilde{O}(\poly(n)\sqrt{T})$  regret? 

Our first main result is that even with the most restrictive feedback~\ref{enum:modelc}, it is possible to obtain $\widetilde{O}(\poly(n)\sqrt{T})$ regret. Thus, the same bounds also hold under \ref{enum:modela} and \ref{enum:modelb}. 
Note that these bounds are almost optimal because standard examples show that even with full feedback every algorithm incurs $\Omega(\sqrt{T})$ regret (see \Cref{sec:lowerBounds}).

\begin{restatable}{Theorem}{mainProphet} \label{thm:mainProphet}
There is a polytime algorithm  with  $O(n^3\sqrt{T}\log T)$  regret for the Bandit Prophet Inequality problem where we only receive the selected value as the feedback.
\end{restatable}


\noindent({We remark that it is possible to improve the $n^{3}$ factor in this result but we do not optimize it to keep the presentation cleaner.})

\newadded{\Cref{thm:mainProphet} may come as a surprise since there are several stochastic problems that admit $O(\poly(n)/\epsilon^2)$ sample complexity but do not admit $\widetilde{O}\big(\poly(n)\sqrt{T}\big)$ regret bandit algorithms. Indeed, a close variant of prophet inequality is sequential posted pricing. Here, the reward is defined as the revenue, i.e., it is the threshold itself if a random variable crosses it rather than the value of the random variable (welfare) as in prophet inequality. It is easy to show that sequential posted pricing has $O(1/\epsilon^2)$ sample complexity \cite{GHTZ-COLT21}, but even for $n=1$ every bandit algorithm incurs $\Omega(T^{2/3})$ regret \cite{LSTW-SODA23}.}

One might wonder whether $\widetilde{O}\big(\poly(n)\sqrt{T}\big)$ regret in \Cref{thm:mainProphet}  holds even for adversarial online learning, i.e., where $X^{(t)}_1 , \ldots, X^{(t)}_n$ are chosen by an adversary in each round $t$ and we compete against the optimal fixed-threshold policy in hindsight. In \Cref{sec:lowerBounds} we prove that this is impossible since every online learning algorithm incurs $\Omega(T)$ regret for adversarial inputs, even under full-feedback.

\begin{Remark}[Hindsight Optimum] \label{remark:hindsightOPT}
There is a lot of work on Prophet Inequality (with Samples) where the benchmark is the expected hindsight optimum $\Ex{\max X_i}$; see \Cref{sec:related}. However, we will be interested in the more realistic benchmark of the optimal policy, or in other words the optimal solution to the underlying MDP, which is standard  in stochastic optimization.
Firstly, comparing to the hindsight optimum does not  make sense for most stochastic problems, including Pandora's Box, since it cannot be achieved even approximately. Secondly,  optimal policy gives us a much more fine-grained picture than comparing to the offline optimum. For instance, it is known that a single sample suffices to get the optimal 2-competitive guarantee compared to the offline optimum~\cite{RWW-ITCS20}. This  might give the impression that there is nothing to be learned about the distributions for Prophet Inequality and sublinear regrets are impossible. However, this   is incorrect as \Cref{thm:mainProphet} obtains sublinear regret bounds w.r.t. the optimal policy.  
\end{Remark}

\subsection{Pandora's Box under Bandit Feedback}
The Pandora's Box problem was introduced by Weitzman,  motivated by Economic search applications \cite{Weitzman-Econ79}.  For example, how should a large organization decide between  competing research technologies to produce some commodity. 
In the classical setting, we are given  distributions $\D_1, \ldots, \D_n$ of $n$ independent random variables. The outcome $X_i \sim \D_i$  for $i \in [n]$ can be obtained by the algorithm by paying a known \emph{inspection cost} $c_i$. The goal is to find a policy to adaptively inspect a subset $S \subseteq [n]$ of the random variables to maximize \emph{utility}:  $\Ex{\max_{i \in S} X_i - \sum_{i\in S} c_i}$. Note that unlike the Prophet Inequality, we may now inspect the random variables in any order by paying a cost and we don't have to immediately accept/reject $X_i$.  

Even though Pandora's Box has an exponential state space,  \cite{Weitzman-Econ79} showed a simple optimal policy where we  inspect in a fixed order (using ``indices'') along with a stopping rule. We  study this problem  in the Online Learning model where the distributions $\D_i$  supported on $[0,1]$ are {unknown-but-fixed}.  Without loss of generality,  we will assume that the deterministic costs $c_i \in [0,1]$ are known to the algorithm\footnote{If the costs $c_i$   are unknown but fixed then the problem trivially reduces to the case of known costs. This is because we could simply open each box once without keeping the prize inside and receive as feedback the  cost $c_i$.}.

Formally, in Online Learning for Pandora's Box  we play a $T$ rounds repeated game where in round $t\in [T]$ we play a policy $a^{(t)}$, which is an order of inspection along with a stopping rule. As reward, we receive our utility  (value minus total inspection cost) on freshly drawn independent random variables $X^{(t)}_1 \sim \D_1, \ldots, X^{(t)}_n\sim \D_n$.  The goal is to minimize the total \emph{regret}, which is the difference  over $T$ rounds in the expected utility  of the optimal algorithm that knows the underlying distributions and the total utility of our algorithm. 

In the full-feedback setting  the algorithm receives the entire sample $X^{(t)}_1, \ldots, X^{(t)}_n$ as feedback in each round. Here,  it is again easy to design an $\widetilde{O}\big(\poly(n)\sqrt{T}\big)$ regret polytime algorithm relying on the results in \cite{GHTZ-COLT21,FL-NeurIPS20}. But these results do not extend to partial feedback.

There are again multiple ways of defining partial feedback.  E.g., we could see the values of all $X_i$  for $i \in S$, meaning that we get to see the values of the inspected random variables. Indeed, our results again apply to the most restrictive form of partial feedback: We  only see the total utility of a policy and not even the indices of inspected random variables or any of their values.

\begin{restatable}{Theorem}{mainThmPB} \label{thm:mainPandora}
There is a polytime algorithm  with  $O(n^{5.5}\sqrt{T}\log T)$  regret for the Bandit Pandora's Box problem where we only receive utility (selected value minus total cost) as  feedback.
\end{restatable}

Again, standard examples show that every algorithm incurs $\Omega(\sqrt{nT})$ regret even will full feedback; see \Cref{sec:lowerBounds}. Furthermore, we will prove in \Cref{sec:lowerBounds} that \Cref{thm:mainPandora} cannot hold for adversarial online learning where $X^{(t)}_1 , \ldots, X^{(t)}_n$ are chosen by an adversary: every online learning algorithm incurs $\Omega(T)$ regret for adversarial inputs, even under full-feedback.

 \subsection{High-Level Techniques}

Let's consider the general Prophet Inequality problem or the subproblem of Pandora's Box where the optimal order is given. In both cases, a policy is described by $n$ thresholds $\tau_1, \ldots, \tau_n \in [0,1]$, defining when to stop inspecting. It would be tempting to apply standard multi-armed bandit algorithms to maximize the expected reward over $[0,1]^n$. However, such approaches are bound to fail because the expected reward is not even continuous\footnote{For example, consider the Prophet Inequality instance in which $X_1$ is a distribution that returns $\frac{1}{4}$ w.p. $\frac{1}{2}$ and $\frac{3}{4}$ otherwise, while $X_2$ is a distribution that always returns $\frac{1}{2}$. The reward of this example is a piece-wise constant function: When $\tau < \frac{1}{4}$ or $\tau \geq \frac{3}{4}$, the expected reward is $\frac{1}{2}$. When $\frac{1}{4} \leq \tau < \frac{3}{4}$, the expected reward is $\frac{5}{8}$.}, let alone convex or Lipschitz. Discretizing the action space and applying a bandit algorithm only leads to  $\Omega(T^{2/3})$ regret. Another reasonable approach is to try to learn the distributions $\D_i$. However, recall that we only get feedback regarding the overall reward of a policy and do not  see which $X_i$ is selected. It is possible to obtain samples from each $X_i$ by considering policies that ignore all other boxes; however, such algorithms that use \emph{separate} exploration and exploitation also have  $\Omega(T^{2/3})$ regret.

Our algorithms combine exploration and exploitation.  We maintain confidence intervals $[\lcb_i, \ucb_i]$ for $i \in [n]$ satisfying w.h.p.\,that the optimal thresholds  $\tau^*_i \in [\lcb_i, \ucb_i]$. 
The crucial difference  from UCB-style algorithms~\cite{ACF-ML02} is that we don't get unbiased samples with low regret, so we cannot maintain or play upper confidences. Instead, we need a ``refinement'' procedure to shrink the intervals while ensuring that the regret during the refinement is bounded.

\IGNORE{
The crucial difference  from UCB algorithms~\cite{ACF-ML02} is that our 
 confidence intervals  apply to the price thresholds, meaning that they apply to the space of all policies, whereas in a UCB-style algorithm one has confidence intervals on the distributions. Interestingly, our algorithms do not  try to learn the underlying distributions,  apart from a short initialization phase where we obtain an  estimate $\hat{F}_i$ of the CDF of $X_i$. 
Afterwards, the confidence intervals are \emph{refined}. This is the main challenge since we need to keep the regret during refinement bounded. Unlike UCB algorithms,  we do not obtain unbiased samples from the distributions, so cannot  play the ``upper confidences'' during refinement.
}

More precisely, our algorithm works in $O(\log T)$ phases. In each phase, we start with confidence intervals  $[\lcb_i, \ucb_i]$ that  satisfy: (i)~$\tau^*_i \in [\lcb_i, \ucb_i]$ and (ii)~playing any thresholds  within the confidence intervals incur at most some $\epsilon$ regret. During the phase, we refine the confidence interval to $[\lcb'_i, \ucb'_i]$ while only playing thresholds  within our original confidence intervals, so that we don't incur much regret. We will show that the new confidence intervals  satisfy that  $\tau^*_i \in [\lcb'_i, \ucb'_i]$ and   that playing any thresholds  within  $[\lcb'_1, \ucb'_1], \ldots, [\lcb'_n, \ucb'_n]$ incur at most  $\frac{\epsilon}{2}$ regret. Thus,   the regret bound goes down by a constant factor in each phase.  

\medskip
\noindent \textbf{Bounding Function to Refine for $n=2$.}
To illustrate the idea behind a refinement phase, let's discuss the case of $n=2$; see \Cref{PIPB2} for more technical details. In this case, there is only one confidence interval $[\lcb, \ucb]$ that we have to refine. Our idea is to define a ``bounding function'' $\delta(\cdot)$ such that the expected regret in a single round when using threshold $\tau \in [\lcb, \ucb]$ is bounded by $\lvert \delta(\tau) \rvert$. Ideally, we would like to choose the optimal threshold $\tau^*$ for which $\delta(\tau^*) = 0$. However, this  requires the knowledge of $\delta$, which we don't have since the distributions are unknown. Instead, we compute an estimate $\hat{\delta}$ of $\delta$ and construct the new confidence interval $[\lcb', \ucb']$ to include all $\tau$ for which $\lvert \hat{\delta}(\tau) \rvert$ is small. The main technical difficulty is to obtain $\hat{\delta}$ while only playing low-regret policies. We achieve this by choosing $\delta$ such that $\hat{\delta}$ can be obtained by using only the estimates $\hat{F}_i$  of the CDF and the empirical average rewards when choosing the boundaries of the confidence interval as thresholds. Note that we do not make any statements about the width of the confidence interval; we only ensure that the regret is bounded when choosing any threshold inside the confidence interval.

\medskip
\noindent \textbf{Prophet Inequality for General $n$.} 
In the case of general $n$, each refinement phase updates the confidence intervals from the last random variable $X_n$ to the first one $X_1$. To refine  confidence interval $[\lcb_i, \ucb_i]$, we use our algorithm for the $n=2$ case as a subroutine, i.e.,  we play $\lcb_i$ and $\ucb_i$ sufficiently many times keeping the other thresholds fixed. However, there are several challenges in this approach. The first important one is that the probability of reaching $X_i$ will change depending on which thresholds are applied before it. We deal with this issue  by always using thresholds from our confidence intervals that maximize the probability of reaching $X_i$. Another important challenge while refining $[\lcb_i, \ucb_i]$ is that   the current choice of thresholds for $X_{i+1},\ldots, X_n$ is not optimal, so we maybe learning a threshold different from $\tau^*_i$. We handle this issue by choosing the other thresholds  in a way that they only improve from phase to phase. We then leave some space in the confidence intervals to accommodate for the improvements in later phases. 

\medskip
\noindent \textbf{Pandora's Box for General $n$.} We still maintain confidence and refine them using ideas similar to Prophet Inequality for general $n$.
The main additional challenge arising in Pandora's box is that the inspection ordering is not fixed. The optimal order is given by ordering the random variables by decreasing thresholds. However, there might be multiple orders consistent with our confidence intervals. Therefore, we keep a set  $S$ of constraints corresponding to a directed acyclic graph on the variables, where an edge from $X_i$ to $X_j$ means that $X_i$ comes before $X_j$ in the optimal order. We update this set by 
consider pairwise swaps. Then, during refinement of confidence interval $[\lcb_i,\ucb_i]$, we choose an inspection order satisfying these constraints while (approximately) maximizing a difference of products objective.

\subsection{Further Related Work} \label{sec:related}
There is a long line of work on both Prophet Inequality (PI) and Pandora's Box (PB), so we only discuss the most relevant papers. For more references, see~\cite{Lucier-Survey17,Singla-Thesis18}. Both PI and PB are classical single-item selection problems, but were popularized in TCS in \cite{HKS-AAAI07} and \cite{KWW-EC16}, respectively, due to their applications in mechanism design. Extensions of these problems to combinatorial settings have been studied in \cite{CHMS-STOC10,KW-STOC12,FGL-SODA15,FSZ-SODA16,Rubinstein-STOC16,RS-SODA17,EFGT-EC20} and in \cite{KWW-EC16,Singla-SODA18,GKS-SODA19,GJSS-IPCO19,FTWWZ-ICALP21}, respectively.
Although the optimal policy for  PI with known distributions   is   a simple dynamic program, designing optimal policies for free-order or in combinatorial PI settings is challenging. Some recent works designing approximately-optimal policies are~\cite{ANSS-EC19,PPSW-EC21,SS-EC21,LLPSS-EC21,BDL-EC22}.

Starting with  Azar, Kleinberg, and Weinberg~\cite{AKW-SODA14}, there is a lot of work on PI-with-Samples where the distributions are unknown but the algorithm has sample access to it~\cite{CDFS-EC19,RWW-ITCS20,GHTZ-COLT21,CDFFLLP-SODA22}. These works, however, compete against the benchmark of expected hindsight optimum, so lose at least a multiplicative factor of $1/2$ due the classical single-item PI and do not admit sublinear regret algorithms. 

The field of Online Learning under both full- and bandit-feedback is  well-established; see books~\cite{CL-Book06,BubeckC-Book12,Hazan-Book16,Slivkins-Book19,LS-Book20}. Most of the initial works focused on obtaining sublinear regret for single-stage problems (e.g., choosing the reward maximizing arm). The last decade has seen progress on learning  multi-stage policies for tabular MDPs under bandit feedback; see \cite[Chapter 38]{LS-Book20}.  However, these  algorithms have a regret that is polynomial in the state space, so they do not apply to PI and PB that  have large MDPs.

Finally, there is  some recent work at the intersection of Online Learning and Prophet Inequality/Pandora's Box~\cite{EHLM-AAAI19,atsidakou2022contextual,GT-ICML22}. These models are significantly different from ours, so do not apply to our problems. The closest one is \cite{GT-ICML22}, where the authors consider Pandora's Box under partial feedback (akin to model \ref{enum:modela}), but for \emph{adversarial} inputs (i.e.,  no underlying distributions). They obtain $O(1)$-competitive algorithms and leave open whether sublinear regrets are possible~\cite{Gergatsouli-22}.  Our lower bounds in \Cref{sec:lowerBoundAdvers} resolve this question by showing that sublinear regrets are impossible for adversarial inputs (even under full feedback), and one has to lose a multiplicative factor in the approximation.


\section{Prophet Inequality and Pandora's Box for $n = 2$}\label{PIPB2}

In this section, we give $O(\sqrt{T}\log T)$ regret algorithms for both Bandit Prophet Inequality  and Bandit Pandora's Box problems with $n=2$ distributions. We discuss this special case of \Cref{thm:mainProphet} before  since it's already non-trivial and showcases one of our main ideas of designing a regret bounding function  that is learnable while playing low-regret Bandit policies.

Our algorithms run in  $O(\log T)$ phases, where the number of rounds doubles each phase. Starting with an initial confidence interval containing the optimal threshold $\tau^* $, the goal of each phase is to refine this interval such that the one-round regret drops by a constant factor for the next phase. In \Cref{ISASec} we discuss each phase's  algorithm for  Prophet Inequality with $n=2$. In \Cref{sec:doublingFramework} we give a generic doubling framework  that combines all phases to prove total regret bounds. Finally, in \Cref{sec:PBTwoBoxes} we extend these ideas to Pandora's Box with $n=2$.


\subsection{Prophet Inequality via an Interval-Shrinking Algorithm}\label{ISASec}

  We first introduce the setting of the Bandit Prophet Inequality Problem with two distributions. Let $\D_1, \D_2$  denote the  two unknown distributions over $[0,1]$ with cdfs $F_1,F_2$ and densities $f_1,f_2$. Consider a $T$ rounds game where in each round $t$ we play a threshold $\tau^{(t)} \in [0, 1]$ and receive as feedback the following reward:
  \begin{itemize}
      \item Independently draw  $X_1^{(t)}$ from  $\D_1$.   If $X_1^{(t)} \geq \tau^{(t)}$, return $X_1^{(t)}$ as the reward.
      \item Otherwise, independently draw $X_2^{(t)}$ from $\D_2$ and return it as the reward.
  \end{itemize}
The only  feedback we receive is the reward, and not even which random variable gets selected.
  
If the distributions are known then the optimal policy is to play  $\tau^* := \Ex{X_2}$ in each round. For $\tau \in [0,1]$,  let $R(\tau)$ be the expected reward of playing one round with threshold $\tau$, i.e.,
  \begin{align} \label{eq:RtauPITwp}
      R(\tau) ~:=~ F_1(\tau)\cdot \Ex{X_2} + \int_\tau^1 x\cdot f_1(x)dx ~=~ 1+F_1(\tau)(\Ex{X_2} - \tau) - \int_\tau^1 F_1(x)dx ,
  \end{align}
  where the second equality uses integration by parts. The \emph{total regret}  is 
$\textstyle T \cdot R(\tau^*) - \sum_{t=1}^T R(\tau^{(t)}) .$

\medskip
\noindent \textbf{Initialization.}  For the initialization, we get $\Theta(\sqrt{T}\log T)$ samples from both $\D_1$ and $\D_2$ by playing $\tau = 0$ and $\tau=1$, respectively. This  incurs $\Theta(\sqrt{T}\log T)$  regret since each round incurs at most $1$ regret. 
The following simple lemma uses the  samples to obtain  initial distribution estimates.

\begin{Lemma}
\label{PI2Init}
After getting $C \cdot sqrt{T}\log T$ samples from $\D_1$ and  $\D_2$, with probability $1 - T^{-10}$ we can:
\begin{itemize}
    \item Calculate $\hat{F}_1(x)$ such that
    $|\hat{F}_1(x) - F_1(x)| \leq T^{-\frac{1}{4}}$ for all $x \in [0, 1]$ simultaneously.
    \item Calculate $\lcb$ and $\ucb$ such that $\ucb - \lcb \leq T^{-\frac{1}{4}}$ and  $\Ex{X_2} \in [\lcb, \ucb]$.
\end{itemize}
\end{Lemma}

\begin{proof}

The first statement follows the DKW inequality (\Cref{DKW}). After taking $N = C\cdot \sqrt{T}\log T$ samples, the probability that $\exists x$ s.t. $|\hat{F}_1(x) - F_1(x)| > \varepsilon = T^{-\frac{1}{4}}$ is at most $2\exp(-2N\varepsilon^2) = 2T^{-2C} < T^{-2C+1}$. So, the first statement holds with probability at least $1 - T^{-11}$ when $C > 10$.

The second statement follows the Hoeffding's Inequality (\Cref{Hoeffding}). After taking $N = C\cdot \sqrt{T}\log T$ samples,   let $\mu$ be the average reward. Let  $\lcb = \mu - \varepsilon$ and $\ucb = \mu + \varepsilon$ for $\varepsilon = \frac{1}{2} T^{-\frac{1}{4}}$. Then, $\ucb - \lcb \leq T^{-\frac{1}{4}}$ by definition. Since the reward of each sample in inside $[0, 1]$, by Hoeffding's Inequality the probability that $|\mu - \Ex{X_2}| > \varepsilon$ is bounded by $2\exp(-2N\varepsilon^2) = 2T^{-2C} < T^{-2C+1}$. So, the second statement holds with probability at least $1 - T^{-11}$  when $C > 10$. Taking a union bound for two statements gives the desired lemma.
\end{proof}

Next we discuss our core algorithm.

\medskip
\noindent \textbf{Interval-Shrinking Algorithm.}
 Starting with an initial confidence interval containing $\tau^*=\Ex{X_2}$, our Interval-Shrinking algorithm (\Cref{ISA}) runs for $\Theta(\frac{\log T}{\epsilon^2})$ rounds and outputs a refined confidence interval.  In the following lemma, we will show that this  refined interval  still contains $\tau^*$ and that the regret of playing any $\tau$ inside this refined interval is bounded by $O(\epsilon)$.

\begin{algorithm}[tbh]
\caption{Interval-Shrinking Algorithm for Prophet Inequality}
\label{ISA}
\KwIn{Interval $[\lcb, \ucb]$,  approximate cdf $\hat{F}_1(x)$, and accuracy $\epsilon$.}
Run $C \cdot \frac{\log T}{\epsilon^2}$ rounds with $\tau = \lcb$. Let $\hat{R}_\lcb$ be the average reward. \\
Run $C \cdot \frac{\log T}{\epsilon^2}$ rounds with $\tau = \ucb$. Let $\hat{R}_\ucb$ be the average reward. \\
For $\tau \in [\lcb, \ucb]$, define $\hat{\Delta}(\tau) := \hat{F}_1(\ucb)(\tau - \ucb) - \hat{F}_1(\lcb)(\tau - \lcb) + \int_{\ell}^{\ucb} \hat{F}_1(x) dx$. \\
For $\tau \in [\lcb, \ucb]$, define $\hat{\delta}(\tau) := \hat{\Delta}(\tau) - (\hat{R}_\ucb - \hat{R}_\lcb)$. \\
Let $\lcb' := \min \{\tau \in [\lcb, \ucb] \text{ s.t. } \hat{\delta}(\tau) \geq -5\epsilon \}$  and let $\ucb' := \max \{\tau \in [\lcb, \ucb] \text{ s.t. } \hat{\delta}(\tau) \leq 5\epsilon \}$. \\
\KwOut{$[\lcb', \ucb']$}
\end{algorithm}

\begin{Lemma}
\label{mainthm}
Suppose we are given:
\begin{itemize}
\item \emph{Initial interval} $[\lcb, \ucb]$ of length  $\ucb - \lcb \leq T^{-\frac{1}{4}}$ and satisfying $ \tau^* \in [\lcb, \ucb]$. 
\item \emph{Distribution estimate} $ \hat{F}_1(x)$ satisfying  $|F_1(x) - \hat{F}_1(x)| \leq T^{-\frac{1}{4}}$ for all $x \in [0, 1]$ simultaneously. 
 \end{itemize}
 Then, for  $\epsilon > T^{-\frac{1}{2}}$
\Cref{ISA} runs thresholds inside  $[\lcb, \ucb]$ for at most $1000 \cdot \frac{\log T}{\epsilon^2}$ rounds,  and outputs a sub-interval $[\lcb', \ucb'] \subseteq [\lcb, \ucb]$ satisfying with probability $1 - T^{-10}$ the following statements:
\begin{enumerate} \topsep=0pt \itemsep=0pt
    \item \label{tSta1} $\tau^* \in [\lcb', \ucb']$.
    \item \label{tSta2} For every $\tau \in [\lcb', \ucb']$ the expected one-round regret of playing $\tau$ is at most $10 \epsilon$.
\end{enumerate}
\end{Lemma}


\medskip
\noindent \textbf{Proof Overview of \Cref{mainthm}.}
The main idea is to define a \emph{bounding function}
\[
\delta(\tau)~:=~(F_1(\ucb) - F_1(\lcb))\cdot (\tau - \tau^*).
\]
As we show in \Cref{LossBound} below, this function satisfies $R(\tau^*) - R(\tau) \leq |\delta(\tau)|$ for all $\tau \in [\lcb, \ucb]$, i.e., $|\delta(\tau)|$ is an upper bound on the one-round regret when choosing $\tau$ instead of $\tau^*$. So, ideally, we would like to choose $\tau$ that minimizes $|\delta(\tau)|$. However, we do not know $\delta(\tau)$. Therefore, we derive an estimate $\hat{\delta}(\tau)$  for all $\tau \in [\lcb, \ucb]$ and discard $\tau$ for which $|\hat{\delta}(\tau)|$ is too large because these cannot be the minimizers.

In order to estimate $\delta(\tau)$, we rewrite it in a different way as sum of terms that can be estimated well.  First, consider the difference in expected rewards when choosing thresholds $\ucb$ and $\lcb$, i.e., 
\[
R(\ucb) - R(\lcb) ~=~ (F_1(\ucb) - F_1(\lcb)) \tau^* - \int_\lcb^\ucb xf_1(x)dx ~=~ F_1(\ucb)\cdot (\tau^* - \ucb) - F_1(\lcb)\cdot (\tau^* - \lcb) + \int_\lcb^\ucb F_1(x)dx,
\]
where we used integration by parts. Adding this with $\delta(\tau)$ gives $\delta(\tau) + (R(\ucb) - R(\lcb))$ equals
\begin{align}\label{CDDfn}
 F_1(\ucb)(\tau - \ucb) - F_1(\lcb)(\tau - \lcb) + \int_\lcb^\ucb F_1(x)dx  ~=:~ \Delta(\tau),
\end{align}
which gives an alternate way of expressing 
$ \delta(\tau) = \Delta(\tau) - (R(\ucb) - R(\lcb)).$ 
(Another way of understanding the definition of $\Delta(\tau)$   is that it represents the difference of playing thresholds $\ucb$ and $\lcb$, assuming that $\Ex{X_2}=\tau$.) 
So, we define the estimate 
\[ \hat{\delta}(\tau) := \hat{\Delta}(\tau) - (\hat{R}_\ucb - \hat{R}_\lcb) ,\] 
where $\hat{\Delta}$ uses the estimate $\hat{F}_1$ instead of $F_1$ in \eqref{CDDfn}  and  to estimate $\hat{R}_\ucb$ and $\hat{R}_\lcb$ we use  empirical averages obtained in the current phase. The advantage is that besides the coarse knowledge of $\hat{F}_1$ we assumed to be given, we only need to choose thresholds from within our current confidence interval to obtain $\hat{\delta}$.  \Cref{AccBound} will show that $\hat{\delta}(\tau)$ estimates $\delta(\tau)$ within an additive error of $O(\epsilon)$.

\medskip
\noindent \textbf{Completing the Proof of \Cref{mainthm}.}
Now we complete the missing details. We first prove that $|\delta(\tau)|$ gives an upper bound on one-round regret with threshold $\tau$.

\begin{Claim}
\label{LossBound}
If $\tau, \tau^* \in [\lcb, \ucb]$, then $R(\tau^*) - R(\tau) \leq |\delta(\tau)|$.
\end{Claim}

\begin{proof}

Consider $R(\tau^*) - R(\tau)$. The two settings are different only when $X_1$ is between $\tau^*$ and $\tau$, and the difference of the reward is bounded by $|\tau^* - \tau|$. Therefore,
$R(\tau^*) - R(\tau) ~\leq~ |\tau^* - \tau|\cdot |F_1(\tau^*) - F_1(\tau)| ~\leq~ |\tau^* - \tau|\cdot |F_1(\ucb) - F_1(\lcb)| = |\delta(\tau)|$, where the second inequality uses $\tau^*, \tau \in [\lcb, \ucb]$ implies $|F_1(\tau) - F_1(\tau^*)| \leq |F_1(\ucb) - F_1(\lcb)|$.
%
\end{proof}

Next, we prove that $\hat{\delta}(\tau)$ is a good estimate of $\delta(\tau)$.

\begin{Claim}
\label{AccBound}
In \Cref{ISA}, if the conditions in \Cref{mainthm} hold then with probability $1 - T^{-10}$ we have $|\hat{\delta}(\tau) - \delta(\tau)| \leq 5 \cdot \epsilon$  for all  $\tau \in [\lcb, \ucb]$ simultaneously.  
\end{Claim}

\begin{proof}
Recall that $\delta(\tau) = \Delta(\tau) -  \big( R(\ucb) - R(\lcb) \big)$. 
We first bound the error  $|\hat{\Delta}(\tau) - \Delta(\tau)|$. Notice, 
\[
    |\hat{\Delta}(\tau) - \Delta(\tau)| ~\leq~ |F_1(\ucb) - \hat{F}_1(\ucb)|\cdot |\tau -\ucb| + |F_1(\lcb) - \hat{F}_1(\lcb)|\cdot |\tau - \lcb| + \int_\lcb^\ucb |F_1(x) - \hat{F}_1(x)|dx.
\]
 The main observation is that all three terms on the right-hand-side can be bounded by $T^{-\frac{1}{2}}$ since  $|F_1(x) - \hat{F}_1(x)| \leq T^{-\frac{1}{4}}$ and $\ucb - \lcb \leq T^{-\frac{1}{4}}$.  Hence, $|\hat{\Delta}(\tau) - \Delta(\tau)| \leq 3T^{-\frac{1}{2}} \leq 3\epsilon$.

Next, we bound the errors for $|\hat{R}_\lcb - R(\lcb)|$ and for $|\hat{R}_\ucb - R(\ucb)|$. For $|\hat{R}_\lcb - R(\lcb)|$, notice that $\hat{R}_\lcb$ is an estimate of $R(\lcb)$ with $N = C\cdot \frac{\log T}{\epsilon^2}$ samples. Since  the reward of each sample is in $[0, 1]$, by Hoeffding's Inequality (\Cref{Hoeffding}) the probability that $|\hat{R}_\lcb - R(\lcb)| > \epsilon$ is bounded by $2\exp(-2N\epsilon^2) = 2T^{-2C}$. Then, $|\hat{R}_\lcb - R(\lcb)| \leq \epsilon$ holds with probability at least $1 - T^{-11}$ when $C>10$. The error bound for $|\hat{R}_\ucb - R(\ucb)|$ is identical. Taking a union bound for two error for $|\hat{R}_\lcb - R(\lcb)|$ and for $|\hat{R}_\ucb - R(\ucb)|$, and then summing them with the error for $|\hat{\Delta}(\tau) - \Delta(\tau)|$ completes the proof.
\end{proof}

Now, we are ready to prove \Cref{mainthm}.

\begin{proof}[Proof of \Cref{mainthm}]
 We will assume that $|\hat{\delta}(\tau) - \delta(\tau)| \leq 5\epsilon$, which is true with probability $1 - T^{-10}$ by \Cref{AccBound}.

Observe that $\hat{\delta}(\tau)$ is a monotone increasing function because $\hat{\delta}'(\tau) = \hat{\Delta}'(\tau) = \hat{F}_1(\ucb) - \hat{F}_1(\lcb) \geq 0$. Therefore, according to the definition of $\lcb'$ and $\ucb'$, we have $[\lcb', \ucb'] = \{\tau \in [\lcb, \ucb]: |\hat{\delta}(\tau)| \leq 5\epsilon\}$. Now, we can use this property to prove the two statements of this lemma separately.

For Statement~\ref{tSta1}, notice that $\delta(\tau^*) = 0$. \Cref{AccBound} gives $|\hat{\delta}(\tau^*)| \leq 5\epsilon$. Then, since $\tau^* \in [\lcb, \ucb]$ and $|\hat{\delta}(\tau^*)| \leq 5\epsilon$, we must have $\tau^* \in [\lcb', \ucb']$ as $[\lcb', \ucb'] = \{\tau \in [\lcb, \ucb]: |\hat{\delta}(\tau)| \leq 5\epsilon\}$. 

Next, we prove Statement~\ref{tSta2}. By \Cref{LossBound}, it suffices to bound $|\delta(\tau)|$ for all $\tau \in [\lcb', \ucb']$. By \Cref{AccBound}, we have  w.h.p. for all $\tau \in [\lcb', \ucb']$ that $|\delta(\tau)| \leq |\hat{\delta}(\tau)| + 5\epsilon \leq 10\epsilon$, where the last inequality uses the definition of $\lcb'$ and $\ucb'$.
\end{proof}

\subsection{Doubling Framework for Low-Regret Algorithms} \label{sec:doublingFramework}

In this section we show how to run \Cref{ISA} for multiple phases with a doubling trick to get   $O(\sqrt{T}\log T)$ regret. Instead of directly proving the regret bound for Prophet Inequality with $n=2$, we first give a general doubling framework that will later  be useful  for Prophet Inequality and Pandora's Box problems with $n$ random variables:

\begin{Lemma}
\label{DoublingBound}

Consider an online learning problem with size $n$. Assume the one-round regret for every possible action is bounded by 1. Suppose there exists an action set-updating algorithm $\alg$ satisfying: Given accuracy $\epsilon$ and action set $A$, algorithm $\alg$ runs $\Theta(\frac{n^\alpha \log T}{\epsilon^2})$ rounds in $A$ and outputs $A' \subseteq A$ satisfying the following  with probability $1 - T^{-10}$:
\begin{itemize}
    \item The optimal action in $A$ belongs to $A'$.
    \item For $a \in A'$, the one-round regret of playing $a$ is bounded by $\epsilon$.
\end{itemize}
Then, with probability $1 - T^{-9}$ the regret of \Cref{Doubling} is $O(n^{\alpha/2} \sqrt{T} \log T)$.

\end{Lemma}

\begin{algorithm}[tbh]
\caption{General Doubling Algorithm}
\label{Doubling}
\KwIn{Time horizon $T$, problem size $n$, action space $A$, algorithm $\alg$, and parameter $\alpha$.}
Let $i = 1$, $\epsilon_1 = 1$, $A_1 = A$ \\
\While{$\epsilon_i > \frac{n^{\alpha / 2} \log T}{\sqrt{T}}$}
{
   Call $\alg$ with input $\epsilon_i$ and $A_i$, and get output $A_{i+1}$ \\
   $\epsilon_{i+1} \gets \frac{\epsilon_i}{2}$ \\
    $i \gets i+1$ 
}
Run $a \in A_i$ for the remaining rounds.
\end{algorithm}

The proof of the lemma uses simple counting; see \Cref{sec:appendixSec2}.

Based on \Cref{DoublingBound}, we can immediately give the Bandit Prophet Inequality regret bound.

\begin{Theorem}
\label{PI2Regret}
There exists an algorithm that achieves $O(\sqrt{T}\cdot \log T)$ regret with probability $1 - T^{-9}$ for Bandit Prophet Inequality  problem with two distributions. 
\end{Theorem}

\begin{proof}
 The initialization runs $O(\sqrt{T} \log T)$ rounds, so the regret is $O(\sqrt{T} \log T)$. For the following interval shrinking procedure, \Cref{ISA} matches the algorithm $\alg$ described in \Cref{DoublingBound} with $\alpha = 0$. Therefore, applying \Cref{DoublingBound} completes the proof.
\end{proof}

%
\subsection{Extending to Pandora's Box  with a Fixed Order}	\label{sec:PBTwoBoxes}
  
In order to extend the approach to Pandora's Box, in this section we consider a simplified problem with a \emph{fixed box order}. There are two boxes taking values in $[0,1]$ from unknown  distributions $\D_1, \D_2$ with cdfs $F_1,F_2$ and densities $f_1,f_2$. The boxes have known costs $c_1, c_2 \in [0,1]$. We assume that we always  pay $c_1$ to observe $X_1$ (i.e., $\Ex{X_1}>c_1$), and then decide whether to observe $X_2$ by paying $c_2$. Indeed, it might be better to open the second box before the first box or not to open any box. We make these simplifying assumptions in this section to make the presentation cleaner.
Generally, determining an approximately optimal order will be one of the main technical challenges that we will need to handle for general $n$ in \Cref{PBGen}.

Formally, consider a $T$ rounds game where in each round $t$ we play a threshold $\tau^{(t)} \in [0, 1]$ and receive as feedback the following  utility:
\begin{itemize}
      \item Independently draw  $X_1^{(t)}$ from  $\D_1$.   If  $X_1^{(t)} \geq  \tau^{(t)}$, we stop and receive $X_1^{(t)}-c_1$ as the utility.
      \item Otherwise, we pay $c_2$ to see $X_2^{(t)}$ drawn independently from $\D_2$, and receive $\max\{X_1, X_2\}-(c_1+c_2)$ as utility.
  \end{itemize}
The only  feedback we receive is the utility, and not even which random variable gets selected.

To see the optimal policy, define a \emph{gain function} $g(v):= \Ex{ \max\{0, X_2 - v\} - c_2}$ to represent the expected additional utility  from opening $X_2$ assuming we already have $X_1 = v$, i.e.,
\begin{align} \label{eq:gainPBTwo}
g(v)  ~&=~ -c_2 + \int_v^1 (x - v)f_2(x) dx 
~=~ -c_2 + (1-v) - \int_v^1 F_2(x)dx. 
\end{align}
The optimal threshold (Weitzman's reservation value) $\tau^*$ is now the solution to $g(\tau^*)=0$, i.e.,  $\Ex{\max\{X_2 - \tau^*,0\}}=c_2$. Since our algorithm does not know $F_2(x)$ but only  an approximate  distribution  $\hat{F}_2(x)$, we get an estimate $\hat{g}(v)$ of $g(v)$ by replacing $F_2(x)$ with $\hat{F}_2(x)$ in \eqref{eq:gainPBTwo}.

For $\tau \in [0,1]$, let \emph{reward function} $R(\tau)$ denote the expected reward of playing $\tau$. With the definition of gain function $g(v)$  and linearity of expectation, we can write
\begin{align*}
    R(\tau) ~:=~ -c_1 + \Ex{X_1} + \int_0^\tau f_1(x)g(x)dx .
\end{align*}
The \emph{total regret} of our algorithm is now defined as
$\textstyle T \cdot R(\tau^*) - \sum_{t=1}^T R(\tau^{(t)}) .$

\medskip
\noindent \textbf{Interval-Shrinking Algorithm.} Starting with an initial confidence interval $[\lcb,\ucb]$ containing $\tau^*$, we again design an Interval-Shrinking algorithm (\Cref{ISAPan}) that runs for $\Theta(\frac{\log T}{\epsilon^2})$ rounds and outputs a refined confidence interval $[\lcb',\ucb']$. We will show that this  refined interval  still contains $\tau^*$ and that the regret of playing any $\tau$ inside this refined interval is bounded by $O(\epsilon)$.  
Now we give the algorithm and the theorem.

\begin{algorithm}[tbh]
\caption{Interval-Shrinking Algorithm for Pandora's Box}
\label{ISAPan}
\KwIn{Interval $[\lcb, \ucb]$, length $m$, and CDF estimates $\hat F_1(x), \hat F_2(x)$. }
Run $C \cdot \frac{\log T}{\epsilon^2}$ rounds with $\tau = \lcb$. Let $\hat{R}_\lcb$ be the average reward. \\
Run $C \cdot \frac{\log T}{\epsilon^2}$ rounds with $\tau = \ucb$. Let $\hat{R}_\ucb$ be the average reward.\\
For $\tau \in [\lcb, \ucb]$, define $\hat{\Delta}(\tau) := (\hat{g}(\ucb) - \hat{g}(\tau))\hat{F}_1(\ucb) - (\hat{g}(\lcb) - \hat{g}(\tau))\hat{F}_1(\lcb) - \int_\lcb^\ucb \hat{g}'(x)\hat{F}(x) dx$. \\
For $\tau \in [\lcb, \ucb]$, define $\hat{\delta}(\tau) := \hat{\Delta}(\tau) - (\hat{R}_\ucb - \hat{R}_\lcb)$. \\
Let $\lcb' = \min \{\tau \in [\lcb, \ucb] \text{ s.t. } \hat{\delta}(\tau) \geq -4\epsilon \}$  and let $\ucb' = \max \{\tau \in [\lcb, \ucb] \text{ s.t. } \hat{\delta}(\tau) \leq 4\epsilon \}$. \\
\KwOut{$[\lcb', \ucb']$}
\end{algorithm}

\begin{Lemma}
\label{mainthmpan}
Suppose we are given:
\begin{itemize}
\item \emph{Initial interval} $[\lcb, \ucb]$ satisfying $ \tau^* \in [\lcb, \ucb]$, gain function  $|g(\tau)| \leq T^{-\frac{1}{4}}$, and bounding function $|\delta(\tau)| \leq 16\epsilon$ where $\delta$ is defined in \eqref{eq:deltaPBnTwo}.
\item CDF estimate $\hat F_1(x)$ which is constructed via $1000 \cdot \frac{\log T}{\epsilon}$ new i.i.d. samples of $X_1$.
\item CDF estimate $\hat F_2(x)$ which is constructed via $1000 \cdot \frac{\log T}{\epsilon}$ new i.i.d. samples of $X_2$.
 \end{itemize}
 Then, for  $\epsilon > T^{-\frac{1}{2}}$, \Cref{ISAPan} runs thresholds inside  $[\lcb, \ucb]$ for no more than $10000 \cdot \epsilon^{-2}\log T$ rounds and outputs with probability $1 - T^{-10}$ a sub-interval $[\lcb', \ucb'] \subseteq [\lcb, \ucb]$ satisfying:

\begin{enumerate} 
    \item  $\tau^* \in [\lcb', \ucb']$.
    \item Simultaneously for every $\tau \in [\lcb', \ucb']$, we have $|\delta(\tau)| \leq 8 \epsilon$. 
    \item  Simultaneously for every $\tau \in [\lcb', \ucb']$, the expected one-round regret of playing $\tau$ is at most $8\epsilon$.
\end{enumerate}
\end{Lemma}

To understand the main idea of the proof, let's compare the expected reward of choosing the optimal threshold $\tau^*$ and an arbitrary threshold $\tau \in [\lcb, \ucb]$. The difference is given by
\[
R(\tau^*) - R(\tau) ~=~ \int_0^{\tau^*} f_1(x) g(x) dx \,- \int_0^{\tau} f_1(x) g(x) dx ~=~ \int_\tau^{\tau^*} f_1(x) g(x) dx .
\]
Note that $g$ is non-increasing since $g'(x) = F_2(x) - 1 \leq 0$. So, using $\tau^*,\tau \in [\lcb,\ucb]$ imply $|F_1(\tau^*)-F_1(\tau)| \leq |F_1(\lcb)-F_1(\ucb)|$, we get $R(\tau^*) - R(\tau) \leq |(F_1(\lcb)-F_1(\ucb)) \cdot  g(\tau)|$. This motivates defining \emph{bounding function} 
\begin{align} \label{eq:deltaPBnTwo}
\delta(\tau) ~:=~ (F_1(\ucb) - F_1(\lcb)) \cdot \big( g(\tau^*) - g(\tau)) ~=~ -(F_1(\ucb)-F_1(\lcb)) \cdot  g(\tau),
\end{align}
and we get the following upper bound on the one-round regret when choosing $\tau$ instead of $\tau^*$.
\begin{restatable}{Claim}{PBBoundTwoBoxes}
\label{LossBoundPan}
If $\tau, \tau^* \in [\lcb, \ucb]$ then $R(\tau^*) - R(\tau) \leq |\delta(\tau)|$.
\end{restatable}

In order to define an estimate $\hat{\delta}(\tau)$ that can be computed using the available information, again consider the rewards when playing thresholds $\ucb$ and $\lcb$. The difference is given by
\begin{align*}
R(\ucb) - R(\lcb) ~=~ \int_\lcb^{\ucb} f_1(x) g(x) dx ~&=~ F_1(\ucb) g(\ucb) - F_1(\lcb) g(\lcb) - \int_\lcb^{\ucb} F_1(x) g'(x) dx \\
~&=~ F_1(\ucb) g(\ucb) - F_1(\lcb) g(\lcb) - \int_\lcb^{\ucb} F_1(x) \cdot(F_2(x) - 1) dx.
\end{align*}
Adding this equation with the definition of $\delta(\tau)$ gives $ \delta(\tau) + R(\ucb) - R(\lcb)$ equals
\begin{align} \label{eq:DeltaDefPBTwo}
F_1(\ucb) \cdot \big( g(\ucb) - g(\tau)\big) - F_1(\lcb) \cdot  \big( g(\lcb) - g(\tau )\big) - \int_\lcb^{\ucb} F_1(x) \cdot  (F_2(x) - 1) dx ~=:~ \Delta(\tau),
\end{align}
which gives us  an alternate way to express 
$\delta(\tau) = \Delta(\tau) - (R(\ucb) - R(\lcb)).$ 
So, we define the estimate 
\[ \hat{\delta}(\tau) := \hat{\Delta}(\tau) - (\hat{R}_\ucb - \hat{R}_\lcb) ,\] 
where $\hat{\Delta}$ uses the estimates $\hat{F}_1$ and  $\hat{g}$ instead of $F_1$ and $g$ in \eqref{eq:DeltaDefPBTwo},  and  to estimate $(\hat{R}_\ucb-\hat{R}_\lcb)$ we use  empirical averages obtained in the current phase. We have the following claim on the accuracy of $\hat{\delta}$ in \Cref{sec:missingPandoraTwo}, which is similar to \Cref{AccBound}.

\begin{restatable}{Claim}{AccBoundPan}
\label{AccBoundPan}
In \Cref{ISAPan}, if the conditions in \Cref{mainthmpan} hold, then with probability $1 - T^{-10}$  $|\hat{\delta}(\tau) - \delta(\tau)| \leq 4\epsilon$ simultaneously for all $\tau \in [\lcb, \ucb]$.
\end{restatable}

The proof of \Cref{AccBoundPan} is different from \Cref{AccBound}: After the initialization, it's not possible to give an initial confidence interval of length at most $T^{-\frac{1}{4}}$. So, we cannot prove an $O(T^{-\frac{1}{2}})$ accuracy for $\Delta(\tau)$. Instead, we use the fact that $\textsf{Var} \Delta(\tau) \leq O(\epsilon)$ to   give an $O(\epsilon)$ accuracy bound using Bernstein inequality (\Cref{Bernstein}) for a single $\tau$. To extend the bound to the whole interval, we discretize and apply a union bound. To avoid the dependency from the previous phases when discretizing, in each phase we use new samples to construct $\hat F_1$ and $\hat F_2$. This is the reason that we introduce sample sets in \Cref{ISAPan}.



Now the proof  of \Cref{mainthmpan} is similar to the proof of \Cref{mainthm} via Claims~\ref{LossBoundPan} and \ref{AccBoundPan}.

Finally, we state the main theorem for Pandora's Box problem with two boxes in a fixed order.

\begin{Theorem}
\label{PB2Regret}
For Bandit Pandora's Box learning problem with two boxes in a fixed order, there exists an algorithm that achieves $O(\sqrt{T} \log T)$ total regret.
\end{Theorem}

The proof  of \Cref{PB2Regret} is similar to \Cref{PI2Regret}: We first show that  $\Theta(\sqrt{T}\log T)$ initial samples are sufficient to meet the conditions in \Cref{mainthmpan}. Combining this with \Cref{DoublingBound} proves the theorem. See \Cref{sec:missingPandoraTwo} for    details.

%
%

\section{Prophet Inequality  for General $n$}
In the Bandit Prophet Inequality problem,  there are $n$ unknown independent distributions $\D_1, \ldots , \allowbreak\D_n$ taking values  in $[0,1]$ with cdfs $F_1,  \ldots, F_n$ and densities $f_1,\ldots, f_n$. Consider a $T$ rounds game where round $t$ we play thresholds $\btau^{(t)}= (\tau_1^{(t)}, \tau_2^{(t)},\ldots,\allowbreak \tau_{n-1}^{(t)}, \tau_n^{(t)} = 0)$ and receive  the following reward: \newadded{For $i \in [n]$, independently draw  $X_i^{(t)}$ from  $\D_i$. Let $j = \min \{i \in [n]: X_i^{(t)} \geq \tau_i^{(t)}\}$. $X_j^{(t)}$ is returned as the reward.}
The only  feedback  is the reward, and we do not see the index $j$ of the selected random variable. Since we have $\tau_n^{(t)} = 0$,  the algorithm will always select a value.
In the following, we omit $\tau_n^{(t)}$ and only use  $ \btau^{(t)}:=(\tau_1^{(t)}, \tau_2^{(t)},\ldots,\allowbreak \tau_{n-1}^{(t)})$ to represent a threshold setting. 

Let $\optafter{i}$ represent the optimal expected reward if only running on distributions $\D_i, \D_{i+1},\ldots, \D_n$. Then, the optimal $i$-th threshold setting  is exactly $\optafter{i + 1}$. We can calculate $\{\optafter{i + 1}\}$ as follows:
\begin{itemize}
    \item Let $\optafter{n} =\Ex{X_n}$
    \item For $i = n - 1 \to 1$: Let $\optafter{i} = R(1,1,\ldots,1, \optafter{i+1}, \optafter{i+2},\ldots, \optafter{n})$, where the function $R(\btau)$ represents the expected one-round reward under thresholds $\btau=(\tau_1,\ldots, \tau_{n-1})$.  
\end{itemize}
 The \emph{total regret} is defined  
 \[ \textstyle T \cdot \optafter{1} - \sum_{t=1}^T R(\btau^{(t)}) .\]

\paragraph{\textbf{High-Level Approach. }}

Following the doubling framework from \Cref{Doubling}, we only need to design an initialization algorithm and a constraint-updating algorithm. For the initialization, we get $O(\poly(n)\sqrt{T}\log T)$ i.i.d. samples for each $X_i$ by playing thresholds $(1,1,\ldots, \tau_{i-1}=1, \tau_i=0, 0, \ldots, 0)$. Besides, we run $O(\poly(n)\sqrt{T} \log T)$ samples to get the initial confidence intervals with small length. 
For the constraint-updating algorithm, we reuse the idea from the $n = 2$ case where we shrink  confidence intervals by testing $X_i$ with thresholds $\lcb_i$ or $\ucb_i$. However, there are two major new challenges while testing $X_i$. 

The first challenge while testing $X_i$ is that we may stop early, and not get sufficiently many samples for $X_i$. Although the probability of reaching $X_i$ could be very small, this also means that we will not reach $X_i$ frequently. To avoid this problem, for  $j < i$, we use the upper confidence bounds as thresholds since they maximize the probability of reaching $X_i$. In particular, it is at least as high as in the optimal policy. Therefore, we will be able to show that the probability term cancels in calculation, so the total loss from $X_i$ can still be bounded.

The second challenge is that when we are testing $X_i$, we need to also set thresholds $\tau_j$ for $j > i$. The problem is that the optimal choice for $\tau_i$ depends on $\tau_j$ for $j>i$. 
To cope this this problem, in our algorithm we use the \textbf{lower confidence bounds} as  thresholds for $j > i$. 
Formally, let $\algafter{i}$ denote the expected reward if only running on distributions $\D_i,\ldots, \D_n$ with lower confidence bounds as the thresholds, i.e., 
\[ 
    \algafter{i} ~:=~ R(1,\ldots, 1, \tau_i = \lcb_i, \tau_{i+1} = \lcb_{i+1},\ldots , \tau_{n-1} = \lcb_{n-1}).
\]
%
Now, under our threshold setting, we can only hope to learn $\algafter{i+1}$, while the optimal threshold is $\optafter{i+1}$. So, our key idea is to first get a new confidence interval for $\algafter{i+1}$. \newadded{Then, since we have $\algafter{i+1} \leq \optafter{i+1}$, the lower bound for $\algafter{i+1}$ is also a lower bound for $\optafter{i+1}$. For the upper bound, we first bound the difference between $\optafter{i+1}$ and $\algafter{i+1}$, and adding this difference to the upper bound for $\algafter{i+1}$ gives the upper bound for  $\optafter{i+1}$.}

\subsection{Interval-Shrinking Algorithm for General $n$}

In this section, we give the interval shrinking algorithm, and provide the regret analysis to show that we can get a new group of confidence intervals that achieves $O(\epsilon)$ regret after $\OTild(\frac{\poly(n)}{\epsilon^2})$ rounds. We first give the algorithm and the corresponding lemma.


\begin{algorithm}[tbh]
\caption{Interval shrinking Algorithm for general $n$}
\label{ISAgen}
\KwIn{Intervals $[\lcb_1, \ucb_1], \ldots, [\lcb_{n-1}, \ucb_{n-1}]$, CDF estimates $\hat{F}_1(x),  \ldots , \hat{F}_n(x)$, and $\epsilon$. }

For $i \in [n-1]$, define $\hat{P}_{i} := \prod_{j \in [i-1]} \hat{F}_{j}(\ucb_{j})$ \\
\For{$i = n - 1 \to 1$}
{
Run $C \cdot \frac{\log T}{\epsilon^2}$ rounds with thresholds  $(\ucb_1, \ldots, \ucb_{i-1}, \lcb_i, \newthreshold_{i+1},
\ldots \newthreshold_{n-1})$ and $C \cdot \frac{\log T}{\epsilon^2}$ rounds with $(\ucb_1, \ldots, \ucb_{i-1}, \ucb_i, \newthreshold_{i+1},\ldots \newthreshold_{n-1})$. Let $\hat{D}_i$ be the difference of the average rewards. \\
For $\tau \in [\lcb_i, \ucb_i]$, define $\hat{\Delta}_i(\tau) := \hat{P}_i(\hat{F}_i(\ucb_i)(\tau - \ucb_i) - \hat{F}_i(\lcb_i)(\tau - \lcb_i) + \int_{\lcb_i}^{\ucb_i} \hat{F}_{i}(x)dx)$. \\
For $\tau \in [\lcb_i, \ucb_i]$, define $\hat{\delta}_i(\tau) := \hat{\Delta}_i(\tau) - \hat{D}_i$. \label{alg:lineDefnDelta} \\
Let $\lcb'_i =\min \big\{\tau \in [\lcb_i, \ucb_i] \text{ s.t. } \hat{\delta}_i(\tau) \geq -\epsilon \big\}$. \\
Let $\ucb'_i = \max \big\{\tau \in [\lcb_i, \ucb_i] \text{ s.t. } \hat{\delta}_i(\tau) \leq  (2n-2i-1)\epsilon \big\}$.
}
\KwOut{$[\lcb'_2, \ucb'_2],\ldots, [\lcb'_{n}, \ucb'_{n}]$}
\end{algorithm}

\begin{Lemma}
\label{maingen}
Suppose we are given:
\begin{itemize}
\item \emph{Distribution estimates} $ \hat{F}_i(x)$ for $i \in [n-1]$ satisfying  $|\prod_{i \in S} \hat{F}_i(x) - \prod_{i \in S} F_i(x)| \leq T^{-{1}/{4}}$ for  all $x \in [0,1]$ and $S \subseteq[n]$.
\item \emph{Initial intervals} $[\lcb_i, \ucb_i]$ for $i \in [n-1]$  of length  $\ucb_i - \lcb_i \leq T^{-{1}/{4}}$ that satisfy  $\optafter{i+1} \in [\lcb_i, \ucb_i]$ and $\algafter{i+1} \in [\lcb_i, \ucb_i]$.
 \end{itemize}
 Then, for  $\epsilon > 12T^{-\frac{1}{2}}$ \Cref{ISAgen} runs no more than 
 $1000 \cdot \frac{n\log T}{\epsilon^2}$ rounds such that in each round the threshold $\btau$ satisfies $\tau_i \in [\lcb_i, \ucb_i]$ for all $i  \in [n-1]$. Moreover,  with probability $1 - T^{-10}$ the following statements hold:
\begin{enumerate}[label=(\roman*)]
     \item \label{Sta1} $\optafter{i+1} \in [\lcb'_i, \ucb'_i]$ for all $i \in [n-1]$.
    \item \label{Sta2} Let $\algafterprime{i} := R(1, \ldots1, \lcb'_{i},\ldots, \lcb'_{n-1})$ for $i\in[n-1]$. Then $\algafterprime{i+1} \in [\lcb'_i, \ucb'_i]$.
    \item \label{Sta3} For every threshold setting $ \btau= (\tau_1, \ldots, \tau_{n-1})$ where $\tau_i \in [\lcb'_i, \ucb'_i]$, the expected one-round regret of playing $\btau$ is at most $2n^2\epsilon$.
\end{enumerate}
\end{Lemma}

We first introduce some notation to prove \Cref{maingen}.  First, we define a single-dimensional function $R_i(\tau)$ to generalize reward function $R(\tau)$  from the $n = 2$ case in \Cref{ISASec}. Ideally, $R_i(\tau)$ should represent the reward of playing $\tau_i = \tau$, but  thresholds $\tau_j$ for $j>i$  also affect its  expected reward. So, to match the setting in \Cref{ISAgen}, we set  thresholds $\tau_{i+1},\ldots, \tau_{n-1}$ to be the updated lower bounds, i.e.,  define
\[
    R_i(\tau) ~:=~ R(1,\ldots,1, \tau_i =\tau, \newthreshold_{i+1},\ldots, \newthreshold_{n-1}). 
\]
Next, we introduce $P_i$, representing the maximum probability of observing $X_i$ when we have confidence intervals $\{[\lcb_i, \ucb_i]\}$, i.e.,
\[
  { \textstyle P_{i} ~:=~ \prod_{j=1}^{i-1} F_j(\ucb_{j}). } 
\]
Replacing $F_j$ with $\hat{F}_j$ in this equation defines estimate $\hat{P}_i$.

Notice that $P_{i}$ also equals the probability of reaching $X_i$ when we play thresholds $\tau_j = \ucb_j$ for all $j < i$ in \Cref{ISAgen}.  So, the loss of playing a sub-optimal threshold $\tau_i$ will be $P_i\cdot (R_i(\algafterprime{i+1}) - R_i(\tau))$ because $P_i$ is the probability of reaching $X_i$ and $\algafterprime{i+1}$ is the optimal threshold when $\tau_j = \lcb'_j$ for all $j > i$. 
We define the generalized \textit{bounding function}:
\[
    \delta_i(\tau) ~:=~  P_{i}\cdot \big(F_{i}(\ucb_i) - F_{i}(\lcb_i)\big)\cdot (\tau - \algafterprime{i+1}). 
\]
We will show in \Cref{LossBoundGen} below that $|\delta_i(\tau)|$ upper bounds $P_i \cdot (R_i(\algafterprime{i+1}) - R_i(\tau))$ for all $\tau \in [\lcb_i, \ucb_i]$. Since we don't know $ \delta_i(\tau) $, we will estimate it by writing in a different way.

Consider the difference in expected rewards between $\tau_i = \ucb_i$ and $\tau_i = \lcb_i$ when the other thresholds are set to $\tau_j = \ucb_j$ for  $j < i$ and $\tau_j = \lcb'_j $ for  $j > i$. The difference between these two settings only comes from $\tau_i$, so the expected difference is
\begin{align*}
    P_i  \cdot  (R_i(\ucb_i) - R_i(\lcb_i)) &~=~ P_i \cdot \left((F_i(\ucb_i) - F_i(\lcb_i)) \algafterprime{i+1} - \int_{\lcb_i}^{\ucb_i} xf_i(x)dx\right) \\
    &~=~ P_i  \cdot  \left(F_i(\ucb_i)(\algafterprime{i+1} - \ucb_i) - F_i(\lcb_i)(\algafterprime{i+1} - \lcb_i) + \int_{\lcb_i}^{\ucb_i} F_i(x)dx\right).
\end{align*}
Adding this with   $\delta_i(\tau)$ implies $\delta_i(\tau) + P_i  \cdot  (R_i(\ucb_i) - R_i(\lcb_i)) $ equals
\begin{align}
 P_i \cdot \left(F_i(\ucb_i)(\tau - \ucb_i) - F_i(\lcb_i)(\tau - \lcb_i) + \int_{\lcb_i}^{\ucb_i} F_i(x)dx\right) ~=:~     \Delta_i(\tau), \label{CDiDfn}
\end{align}
which gives another way of writing $\delta_i(\tau)  = \Delta_i(\tau) -  P_i  \cdot  (R_i(\ucb_i) - R_i(\lcb_i))$. 
Since $\hat{D}_i$ from \Cref{ISAgen} is the difference between average rewards of taking  samples with $\tau_i = \ucb_i$ and $\tau_i = \lcb_i$, it is an unbiased estimator of $P_i  \cdot  (R_i(\ucb_i) - R_i(\lcb_i))$. So, we define  estimate
\[
\hat{\delta}_i(\tau)  ~:=~  \hat{\Delta}_i(\tau) -  \hat{D}_i,
\]
where  $\hat{\Delta}_i(\tau)$ is obtained by replacing $F_i$ with $\hat{F}_i$ and $P_i$ with $\hat{P}_i$ in \eqref{CDiDfn}.

 Similar to \Cref{LossBound} and \Cref{AccBound}, we introduce the following claims for \Cref{ISAgen}.

\begin{Claim}
\label{LossBoundGen}
For $i \in [n-1]$, if $\algafterprime{i+1} \in [\lcb_i, \ucb_i]$ and $\tau \in [\lcb_i, \ucb_i]$, then $P_i\cdot \big(R_i(\algafterprime{i+1}) - R_i(\tau) \big) \leq  |{\delta_i(\tau)}|$. 
\end{Claim}

\begin{proof}

We only need to prove that $R_i(\algafterprime{i+1}) - R_i(\tau) \leq \frac{|{\delta_i(\tau)}|}{P_i} =\big(F_{i}(\ucb_i) - F_{i}(\lcb_i)\big)\cdot (\tau - \algafterprime{i+1})$. Now the proof is identical to  \Cref{LossBound} by replacing function $R(\cdot)$ with $R_i(\cdot)$. 
\end{proof}

\begin{Claim}
\label{AccBoundGen}
In \Cref{ISAgen}, if the conditions in \Cref{maingen} hold, then with probability $1 - T^{-10}$, we have $|\hat{\delta}_i(\tau) - \delta_i(\tau)| \leq \epsilon$ simultaneously for all $\tau \in [\lcb_i, \ucb_i]$.
\end{Claim}

\begin{proof}

There are two terms in $\delta_i(\tau) = \Delta_i(\tau) - P_{i}\cdot \left(R_i(\ucb) - R_i(\lcb)\right)$.  We prove that the error of each term is bounded by $\frac{\epsilon}{2}$ with high probability, which will complete the proof by a union bound.

We first bound $|\hat{\Delta}_i(\tau) -\Delta_i(\tau)|$. There are three terms in $\frac{\Delta_i(\tau)}{P_i} = F_i(\ucb_i)(\tau - \ucb_i) - F_i(\lcb_i)(\tau - \lcb_i) + \int_{\lcb_i}^{\ucb_i} F_i(x)dx$. Since the conditions in \Cref{maingen} guarantee that $|\hat{F}_i(x) - F_i(x)| \leq T^{-{1}/{4}}$ and $\ucb_i - \lcb_i \leq T^{-{1}/{4}}$,  the error in each term is at most $\frac{1}{\sqrt{T}}$ and the total error  $\left|\frac{\Delta_i(\tau)}{P_i} - \frac{\hat{\Delta}_i(\tau)}{\hat{P}_i}\right| \leq \frac{3}{\sqrt{T}}$. 

For $P_{i}$, the  preconditions in \Cref{maingen} guarantee that $|\hat{P}_{i} - \hat{P}_{i}| \leq T^{-{1}/{4}}$. Moreover, observe that $\ucb_i - \lcb_i \leq T^{-{1}/{4}}$ implies that  $\frac{|\Delta_i(\tau)|}{P_i} \leq 3T^{-{1}/{4}}$. So,
\begin{align*}
    \left|\hat{\Delta}_i(\tau) - \Delta_i(\tau)\right| ~\leq~ \left|\hat{P}_{i}\left(\frac{\hat{\Delta}_i(\tau)}{\hat{P}_i} - \frac{\Delta_i(\tau)}{P_i}\right)\right| + \left|(\hat{P}_{i} - P_{i})\frac{\Delta_i(\tau)}{P_i}\right| ~\leq~ \frac{6}{\sqrt{T}} ~\leq~ \frac{\epsilon}{2},
\end{align*}
where the last inequality uses $\epsilon > 12T^{-\frac{1}{2}}$.

For $P_{i}\cdot (R_i(\ucb_i)-R_i(\lcb_i))$, note that $\hat{D}_i$ is an unbiased estimator of $P_i\cdot (R_i(\ucb_i) - R_i(\lcb_i))$ with $N = C \cdot \frac{\log T}{\epsilon^2}$ samples. So, by  Hoefdding's Inequality,
\[ \Pr{ \big|\hat{D}_i - P_i\cdot (R_i(\ucb_i) - R_i(\lcb_i)) \big| > \frac{\epsilon}{2}} ~\leq~ 2\exp(-8N\epsilon^2) ~=~ 2T^{-8C}.\] 
Thus, $|\hat{D}_i - P_i\cdot (R_i(\ucb_i) - R_i(\lcb_i))| \leq \frac{\epsilon}{2}$ holds with probability $1 - T^{-10}$ when  $C > 10$.
\end{proof}



Besides \Cref{LossBoundGen} and \Cref{AccBoundGen}, we also need some other properties of \Cref{ISAgen} to prove \Cref{maingen}. The next claim shows that the expected reward of playing lower confidence bounds increases phase to phase.

\begin{Claim}
\label{BeforeMain1}
Assume the conditions in \Cref{maingen} and the bound in \Cref{AccBoundGen} hold. Then, for $i \in [n]$, we have $\algafterprime{i} \geq \algafter{i}$.
\end{Claim}

\begin{proof}
We prove  by induction for $i$ going from $n$ to $1$. The base case $i = n$  holds  because $\algafterprime{n} = \algafter{n} = \optafter{n}= \Ex{X_n}$ by definition.
  
For the induction step, assume  that $\algafterprime{i+1} \geq \algafter{i+1}$ by induction hypothesis. Observe that
 %
%
 \begin{align}
     R(1,\ldots, 1, \lcb_i, \lcb_{i+1}, \ldots, \lcb_{n-1}) & = \Ex{X_i \cdot \one_{X_i > \lcb_i}} + \Pr{X_i \leq \lcb_i} R(1,\ldots, 1, 1, \lcb_{i+1}, \ldots, \lcb_{n-1}) \notag \\
     & \leq \Ex{X_i \cdot \one_{X_i > \lcb_i}} + \Pr{X_i \leq \lcb_i} R(1,\ldots, 1, 1, \newthreshold_{i+1}, \ldots, \newthreshold_{n-1}) \notag  \\
     & = R(1,\ldots, \lcb_i, \newthreshold_{i+1},\ldots, \newthreshold_{n-1}), \label{EqBM1.1}
 \end{align}
 where the inequality uses induction hypothesis as $R(1,\ldots, 1, 1, \lcb_{i+1}, \ldots, \lcb_{n-1}) = \algafter{i+1}\leq \algafterprime{i+1}  =  R(1,\ldots, 1, 1, \newthreshold_{i+1}, \ldots, \newthreshold_{n-1}) $.

Next, we  have $R_i(\lcb_i) \leq R_i(\lcb'_i)$, i.e.,
 \begin{align}
      R(1,\ldots, 1, \lcb_{i}, \newthreshold_{i+1},\ldots, \newthreshold_{n-1}) ~\leq~ R(1,\ldots, 1, \newthreshold_{i}, \newthreshold_{i+1},\ldots, \newthreshold_{n-1}). \label{EqBM1.2}
 \end{align}
 To prove this, we first observe that if $\lcb_{i} = \lcb'_{i}$, then the inequality is an equality. Otherwise, there must be  $\hat{\delta}_{i}(\lcb'_{i}) = -\epsilon$. 
Next, combining the definition of $\delta_i(\tau)$ and \Cref{AccBoundGen}, we have $\hat{\delta}_{i}(\algafterprime{i+1}) \geq \delta_{i}(\algafterprime{i+1}) - |\delta_{i}(\algafterprime{i+1}) - \hat{\delta}_{i}(\algafterprime{i+1})|\geq -\epsilon$. Since $\hat{\delta}'_i(\tau) = P_i\cdot (\hat{F}_i(\ucb_i) - \hat{F}_i(\lcb_i)) \geq 0$ means  $\hat{\delta}_i(\tau)$ is increasing, there must be $\algafterprime{i+1} \geq \lcb'_{i}$.

Now consider function $R_i(\tau)$. Recall that $R_i(\tau) = R(1,...1, \tau_i = \tau, \lcb'_{i+1},..., \lcb'_{n-1})$. Therefore, 
\[
R_i(\tau) ~= ~ \Pr{X_i \leq \tau} \cdot \algafterprime{i+1} + \Ex{X_i\cdot \one_{X_i > \tau}} ~=~ F_i(\tau) \cdot \algafterprime{i+1} + \int_\tau^1 f_i(x)\,x\,dx,
\]
which means $R'_i(\tau) = f_i(\tau)(\algafterprime{i+1} - \tau)$, showing that $R_i(\tau)$ is a unimodular function and reaches its maximum when $\tau = \algafterprime{i+1}$. Hence,  \eqref{EqBM1.2} holds because $\lcb_{i+1} \leq \newthreshold_{i+1} \leq \algafterprime{i+1}$.

Combining (\ref{EqBM1.1}) and (\ref{EqBM1.2}) proves the claim.
\end{proof}

Next, we prove that $\algafterprime{i+1} \in [\lcb_i, \ucb_i]$, which is crucial for us to use \Cref{LossBoundGen}. 

\begin{Claim}
\label{BeforeMain2}
Assume that the preconditions in \Cref{maingen} and the bound in \Cref{AccBoundGen} hold, then $\algafterprime{i+1} \in [\lcb_i, \ucb_i]$ for all $i \in [n-1]$.
\end{Claim}

\begin{proof}
  \Cref{BeforeMain1} shows that $\algafter{i+1} \leq \algafterprime{i+1}$. On the other hand, $\algafterprime{i+1} \leq \optafter{i+1}$ holds because $\optafter{i+1}$ is the maximum achievable reward. Then, \Cref{BeforeMain2} holds because $\optafter{i+1}, \algafter{i+1} \in [\lcb_i, \ucb_i]$ by the  preconditions in \Cref{maingen}.
\end{proof}

Finally, we show that $\algafterprime{i}$ cannot be much smaller than $\optafter{i}$.

\begin{Claim}
\label{BeforeMain3}
Assume that the preconditions in \Cref{maingen} and the bound in \Cref{AccBoundGen} hold, then  $\optafter{i} - \algafterprime{i} \leq \frac{2(n-i)\epsilon}{P_i}$ for all $i \in [n-1]$.
\end{Claim}

\begin{proof}
We prove by induction for $i$ going from $n$ to $1$. The base case $i = n$ holds because $\optafter{n} = \algafterprime{n}=\Ex{X_n}$.

For the induction step, we assume that $\optafter{i+1} - \algafterprime{i+1} \leq \frac{2(n-i-1)\epsilon}{P_{i+1}}$ and  would like to show that $\optafter{i} - \algafterprime{i} \leq \frac{2(n-i)\epsilon}{P_i}$.
We first have
\begin{align}
    R(1,\ldots,1, \optafter{i+1}, \optafter{i+2},\ldots, \optafter{n}) &= \Ex{X_i \cdot \one_{X_i > \optafter{i+1}}} + \Pr{X_i \leq \optafter{i+1}} \optafter{i+1} \notag \\
    &\leq \textstyle \Ex{X_i \cdot \one_{X_i > \optafter{i+1}}} + \Pr{X_i \leq \optafter{i+1}}(\algafterprime{i+1} + \frac{2(n-i-1)\epsilon}{P_{i+1}}) \notag \\
    &\leq \textstyle \Ex{X_i \cdot \one_{X_i > \optafter{i+1}}} + \Pr{X_i \leq \optafter{i+1}}\algafterprime{i+1} + \frac{2(n-i-1)\epsilon}{P_{i}} \notag \\
    &= \textstyle R(1,\ldots,1, \optafter{i+1}, \lcb'_{i+1},\ldots, \lcb'_{n-1}) + \frac{2(n-i-1)\epsilon}{P_{i}}, \label{EqBM3.3}
\end{align}
where we use the induction hypothesis in the second line, and the fact that $\Pr{X_i \leq \optafter{i+1}} \leq \Pr{X_i \leq \ucb_i} = \frac{P_{i+1}}{P_i}$ in the third line.
 
Next, since $\algafterprime{i+1}$ is the optimal threshold, we  have
  \begin{align}
     R(1,\ldots,1, \optafter{i+1}, \newthreshold_{i+2}, \ldots, \newthreshold_{n-1}) ~\leq~ R(1,\ldots,1, \algafterprime{i+1}, \newthreshold_{i+2}, \ldots, \newthreshold_{n-1}). \label{EqBM3.4}
 \end{align}
Finally,
\[
     |\delta(\newthreshold_{i})| ~\leq~ |\hat{\delta}(\newthreshold_{i})| + |\hat{\delta}(\newthreshold_{i}) - \delta(\newthreshold_{i})| ~\leq~ \epsilon + \epsilon ~=~2\epsilon, 
\]
 where the bound of $|\hat{\delta}(\newthreshold_{i}) - \delta(\newthreshold_{i})|$ is from \Cref{AccBoundGen}, and the bound of $|\hat \delta(\newthreshold_{i})|$ is from \Cref{ISAgen}. Combining this with \Cref{LossBoundGen}, we have
     $R_{i}(\algafterprime{i+1}) - R_{i}(\newthreshold_{i}) ~\leq~  \frac{|\delta_{i}(\newthreshold_{i})|}{P_{i}}
      ~\leq~ \frac{2\epsilon}{P_i} $, which is exactly
 \begin{align*}
     R(1,\ldots,1, \algafterprime{i+1}, \newthreshold_{i+1}, \ldots, \newthreshold_{n-1}) ~\leq~ R(1,\ldots,1, \newthreshold_{i}, \newthreshold_{i+1}, \ldots, \newthreshold_{n-1})  + \frac{2\epsilon}{P_i}.  
 \end{align*}
 Summing this with (\ref{EqBM3.3}) and (\ref{EqBM3.4})  completes the induction step.
\end{proof}

Finally, we can prove \Cref{maingen}.

\begin{proof}[Proof of \Cref{maingen}]

In this proof, we assume \Cref{AccBoundGen} always holds. Then the whole proof should success with probability $1 - T^{-10}$.

We prove the three statements separately:

\noindent \textit{Statement~\ref{Sta1}.}  For the upper bound, \Cref{BeforeMain3} shows that $\optafter{i} - \algafterprime{i} \leq \frac{2(n - i)\epsilon}{P_{i}}$. Therefore, $\delta_i(\optafter{i+1}) \leq P_{i} \cdot F_{i}(\ucb_i) \cdot \frac{2(n - i-1)\epsilon}{P_{i+1}} = 2(n-i-1)\epsilon$. Combining this with \Cref{AccBoundGen}, we have $\hat{\delta}_i(\optafter{i+1}) \leq 2(n-i-1)\epsilon+ \epsilon = (2n-2i-1)\epsilon$. Then $\optafter{i+1} \leq \ucb'_i$, because $\optafter{i+1} \in [\lcb_i, \ucb_i]$, $\ucb'_i = \max\{\tau: \tau \in [\lcb_i, \ucb_i] \land \hat{\delta}_i(\tau) \leq (2n-2i-1)\epsilon\}$ and the monotonicity of $\hat{\delta}_i$ function.

For the lower bound, at least we have $\optafter{i+1} \geq \algafterprime{i+1}$. Therefore, $\delta_i(\optafter{i+1}) \geq 0$, so $\hat{\delta}_i(\optafter{i+1}) \geq -\epsilon$. Then $\optafter{i+1} \geq \lcb'_i$, because $\optafter{i+1} \in [\lcb_i, \ucb_i]$, $\lcb'_i = \max\{\tau: \tau \in [\lcb_i, \ucb_i] \land \hat{\delta}_i(\tau) \geq -\epsilon\}$ and the monotonicity of $\hat{\delta}_i$ function. Combining the two bounds proves Statement~\ref{Sta1}.

\noindent \textit{Statement~\ref{Sta2}.} The proof idea is the same as Statement~\ref{Sta1}. Notice that $\delta_i(\algafterprime{i+1}) = 0$. Then, according to \Cref{AccBoundGen}, $|\hat{\delta}_i(\algafterprime{i+1})| \leq \epsilon$. So Statement~\ref{Sta2} hold because $\algafterprime{i+1} \in [\lcb_i, \ucb_i]$, which is from \Cref{BeforeMain2}, and $[\lcb'_i, \ucb'_i] \supseteq \{\tau\in [\lcb_i, \ucb_i]: |\hat{\delta}_i(\tau)| \leq\epsilon\}$.

\noindent \textit{Statement~\ref{Sta3}.} We prove the following stronger statement by induction on $i$: If $\tau_j \in [\lcb'_j, \ucb'_j]$ for all $j \in \{ i,\ldots, n\}$, then 
\[ 
    \algafterprime{i}~-~ R(1,\ldots1, \tau_{i},\ldots, \tau_{n-1}) ~\leq~ \textstyle \frac{(n-i+1)^2\epsilon}{P_{i}}.
\]
When the statement above holds, taking $i = 1$ gives $R(\tau_1,\ldots, \tau_{n-1}) \geq \algafterprime{1} - n^2\epsilon$. Furthermore, \Cref{BeforeMain3} shows that $\algafterprime{1} \geq \optafter{1} - 2(n-1)\epsilon$. Combining these two inequalities proves Statement~\ref{Sta3}.

It remains to prove the induction statement. The base case  $i = n$  holds trivially.


For the induction step, we will assume that the statement holds for $i+1$ and we have to show it also holds for $i$.
By induction hypothesis, 
\begin{align*}
 R(1,\ldots,1, \tau_{i}, \ldots, \tau_{n-1}) ~& =~  \Ex{X_{i} \cdot \one_{X_{i} \geq \tau_{i}}} + \Pr{X_{i} < \tau_{i}} R(1,\ldots,1, \tau_{i+1}, \ldots, \tau_{n-1}) \\
&\textstyle \geq \Ex{X_{i} \cdot \one_{X_{i} \geq \tau_{i}}} + \Pr{X_{i} < \tau_{i}} \left( R(1,\ldots,1, \lcb'_{i+1}, \ldots, \lcb'_{n-1}) - \frac{(n-i)^2\epsilon}{P_{i+1}} \right) \\
& \textstyle \geq \Ex{X_{i} \cdot \one_{X_{i} \geq \tau_{i}}} + \Pr{X_{i} < \tau_{i}} R(1,\ldots,1, \lcb'_{i+1}, \ldots, \lcb'_{n-1}) - \frac{(n-i)^2\epsilon}{P_{i}} \\
& \textstyle = R(1,\ldots,1, \tau_{i}, \lcb'_{i+1}, \ldots, \lcb'_{n-1}) - \frac{(n-i)^2\epsilon}{P_{i}} .
\end{align*}

Furthermore,  $|\delta_i(\tau_{i})| \leq |\hat{\delta}_{i}(\tau_i)| + \epsilon$ by \Cref{AccBoundGen} and $|\hat{\delta}_{i}(\tau)| \leq (2n-2i-1)\epsilon$ by the definitions of $\lcb'_{i}$ and $\ucb'_{i}$, which means $|\delta_i(\tau)| \leq 2(n-i)\epsilon$. So, \Cref{LossBoundGen} implies
%
\[
     R(1,\ldots,1, \algafterprime{i+1}, \newthreshold_{i+1}, \ldots, \newthreshold_{n-1}) - R(1,\ldots,1, \tau_{i}, \newthreshold_{i+1}, \ldots, \newthreshold_{n-1}) ~\leq~  \frac{2(n-i)\epsilon}{P_{i}} .
\]

Finally, using  $R(1,\ldots,1, \algafterprime{i+1}, \newthreshold_{i+1}, \ldots, \newthreshold_{n-1}) \geq \algafterprime{i}$, we get 
\[
R(1,\cdots,1, \tau_{i}, \ldots, \tau_{n-1}) ~\geq~ \algafterprime{i} - \frac{2(n-i)\epsilon}{P_{i}} - \frac{(n-i)^2\epsilon}{P_{i}} ~\geq~ \algafterprime{i} - \frac{(n-i+1)^2\epsilon}{P_{i}}. \qedhere
\]
%
%
\end{proof}

\subsection{Initialization and Putting Everything Together}

Now, we can give the initialization algorithm. The main goal of the initialization is to satisfy the conditions listed in \Cref{maingen}. Starting from the second call of \Cref{ISAgen}, the confidence interval length constraint and the distribution estimates constraints hold from the initialization, and the constraints $\optafter{i+1}, \algafter{i+1} \in [\lcb_i, \ucb_i]$ are guaranteed by Statements~\ref{Sta1} and \ref{Sta2} in \Cref{maingen}. Then, we can apply \Cref{DoublingBound} to bound the total regret.

We first give the initialization algorithm:

\begin{algorithm}[tbh]
\caption{Initialization}
\label{alg:pi-initialization}
\KwIn{Time horizon $T$, problem size $n$.}
\For{$i = 1 \to n $}
{
Run $1000n^2\sqrt{T} \log T$ free samples for $X_i$ to estimate $\hat{F}_i(x)$.
}
\For{$i = n-1 \to 1$}
{
Run $1000n^2\sqrt{T}\log T$ samples under the threshold setting $(1,\ldots,1,\tau_{i+1} = \newthreshold_{i+1},\ldots, \tau_{n-1} = \newthreshold_{n-1})$. Let $\mu_i$ be the average reward. \\
Let $\lcb_i = \mu_i - \frac{T^{-{1}/{4}}}{10n}$, $\ucb_i = \mu_i + (2n - 2i-1) \cdot \frac{T^{-{1}/{4}}}{10n}$
}
\KwOut{$[\lcb_1,\ucb_1],\ldots, [\lcb_{n-1}, \ucb_{n-1}]$. }
\end{algorithm}

\begin{Lemma}
\label{PIgenInit}
\Cref{alg:pi-initialization} runs $O(n^3\sqrt{T}\log T)$ rounds. The output satisfies with probability $1 - T^{-10}$ all constraints listed in \Cref{maingen}.
\end{Lemma}

\begin{proof}

For the accuracy bound of $\hat{F}_i(x)$, we first show that $|\hat{F}_i(x) - F_i(x)| \leq \frac{T^{-1/4}}{2n}$ with probability $1 - T^{-11}$ after running $N = C\cdot n^2\sqrt{T}\log T$ samples with $C = 1000$. With DKW inequality (\Cref{DKW}), we have \[\Pr{|\hat{F}_i(x) - F_i(x)| > \varepsilon ~=~ \frac{T^{-\frac{1}{4}}}{2n}} ~\leq~ 2\exp(-2N\varepsilon^2) ~=~ 2T^{-C/4}.\] 
So the bound holds with probability $1 - T^{-12}$ when $C = 1000$. By the union bound, with probability $1 - T^{-11}$, we have $|\hat{F}_i(x) - F_i(x)| \leq \frac{T^{-\frac{1}{4}}}{2n}$ holds for every $i \in [n]$. Then, for the accuracy of $\prod_{i \in S} F_i(x)$, we have $ \big((1 - \frac{T^{-\frac{1}{4}}}{2n})^n - 1\big) \leq \prod_{i \in S} \hat{F}_i(x) -\prod_{i \in S} F_i(x) \leq \big((1 + \frac{T^{-\frac{1}{4}}}{2n})^n - 1\big)$. For the lower bound, we have $(1 - \frac{T^{-\frac{1}{4}}}{2n})^n - 1 \geq 1 - \frac{T^{-\frac{1}{4}}}{2} - 1 > -T^{-\frac{1}{4}}$. For the upper bound, we have $(1 + \frac{T^{-\frac{1}{4}}}{2n})^n - 1 \leq \exp(\frac{T^{-\frac{1}{4}}}{2n} \cdot n) - 1 \leq 1 + 2\cdot \frac{T^{-\frac{1}{4}}}{2}  - 1 = T^{-\frac{1}{4}}$. Combining two bounds finishes the proof.

For the confidence interval, the constraints $\ucb_i - \lcb_i \leq T^{-\frac{1}{4}}$ hold by definition. Then, it only remains to show $\optafter{i+1} \in [\lcb_i, \ucb_i]$ and $\algafter{i+1} \in [\lcb_i, \ucb_i]$. 

We start from proving $\algafter{i+1} \in [\lcb_i, \ucb_i]$. Notice that $\mu_i$ is an estimate of $\algafter{i+1}$ with $N=C\cdot n^2\sqrt{T}\log T$ samples with $C = 1000$. With Hoeffding's Inequality (\Cref{Hoeffding}), we have 
\[ \Pr{|\mu_i - \algafter{i+1}| > \varepsilon = \frac{T^{-{1}/{4}}}{10n} } ~<~ 2\exp(-2N\varepsilon^2) ~=~ 2T^{-C/50}.\] 
Notice that $\lcb_i = \mu_i - \frac{T^{-{1}/{4}}}{10n}$ and $\ucb_i = \mu_i + \frac{T^{-{1}/{4}}}{10n}$. Then, by the union bound for all $i \in [n]$ ,  we have $\algafter{i+1} \in [\lcb_i, \ucb_i]$ holds for all $i$ with probability $1 - T^{-11}$ when $C \geq 1000$.

For $\optafter{i+1}$, we prove the statement by doing induction with the assumption that $|\algafter{i+1} - \mu_i| \leq \frac{T^{-{1}/{4}}}{10n}$ for all $i$. The base case is $i = n$, the statement simply holds because $\algafter{n} = \optafter{n}$. Next, we consider $i$, with the condition that $\optafter{j+1} \in [\lcb_{j}, \ucb_{j}]$ for all $j > i$. For the lower bound, since we know that $\algafter{i+1} \geq \lcb_i$, there must be $\optafter{i+1} \geq \lcb_i$, because $\optafter{i+1} \geq \algafter{i+1}$. For the upper bound, we first bound the difference between $\algafter{i+1}$ and $\optafter{i+1}$. Consider the setting $(1,\ldots1, \tau_{i+1} = \lcb_{i+1},\ldots, \tau_{n-1} = \lcb_{n-1})$ and $(1,\ldots, 1, \tau_{i+1} =\optafter{i+2},\ldots, \tau_{n-1} = \optafter{n})$. The first setting incurs an extra loss only when its behavior is different from the second setting. Assume the two settings behave differently when meeting a threshold $\tau_j$. Notice that this extra loss is bounded by $|\lcb_j - \optafter{j+1}|$. Since $\optafter{j+1} \in [\lcb_j, \ucb_j]$ for all $j > i$, this difference is upper bounded by $\max_{j > i} \ucb_j - \lcb_j = \ucb_{i+1} - \lcb_{i+1} = (2n-2i-2) \cdot \frac{T^{-{1}/{4}}}{10n}$. Therefore, 
\begin{align*}
    \optafter{i+1} ~\leq~ \algafter{i+1} + (2n-2i-2) \cdot \frac{T^{-{1}/{4}}}{10n} ~\leq~ \mu_i + (2n-2i-1) \cdot \frac{T^{-{1}/{4}}}{10n} ~=~ \ucb_i.
\end{align*}
Combining the lower bound and the upper bound proves $\optafter{i+1} \in [\lcb_i, \ucb_i]$. Finally, taking union bounds for all events that hold with probability $1 - T^{-11}$ finishes the proof.
\end{proof}

Now we are ready to prove the main theorem.

\mainProphet*

\begin{proof}

For the initialization, \Cref{alg:pi-initialization} runs $O(n^3\sqrt{T} \log T)$ rounds, so the total regret from the initialization is $O(n^3\sqrt{T} \log T)$.

For the main algorithm, we run \Cref{Doubling} with \Cref{ISAgen} being the required sub-routine $\alg$. This is feasible because the requirements in \Cref{maingen} are guaranteed by the initialization and \Cref{maingen} itself. Besides, \Cref{maingen} implies that \Cref{ISAgen} upper-bound the one-round regret by $\epsilon$ after $O(\frac{n^5\log T}{\epsilon^2})$ samples. Applying \Cref{DoublingBound} with $\alpha = 5$, we have the $O(n^{2.5}\sqrt{T} \log T)$ regret bound. Combining two parts finishes the proof.
\end{proof}


\section{Pandora's Box for General $n$}\label{PBGen}

In the Bandit Pandora's Box problem,  there are $n$ unknown independent distributions $\D_1, \ldots , \allowbreak\D_n$ representing the values of the $n$ boxes. The distributions have cdfs $F_1,  \ldots, F_n$ and densities $f_1,\ldots, f_n$.  Moreover, each box/distribution $\D_i$ has a known inspection cost $c_i$. Although in the original problem in introduction we assumed that  the values and costs  have support $[0,1]$, in this section we will scale down the costs and  values by a factor of $2n$, so that they have support $[0, \frac{1}{2n}]$. 
This scaling helps to bound the utility in each round between $[-0.5, 0.5]$. To obtain bounds for the original unscaled problem, we will multiply our bounds with this  factor $2n$  in the final analysis.

Consider a $T$ rounds game where in each round  we play some permutation $\pi$ representing the order of  inspection and $n$ thresholds $(\tau_{\pi(1)}, \ldots, \tau_{\pi(n)})$.  Our algorithm receives the following utility as feedback: \newadded{For $i \in [n]$, draw $X_\pi(i) \sim \D_\pi(i)$. Let $j$ be the minimum index that satisfies $\max\{X_\pi(1),\ldots, X_\pi({j-1})\} \geq \tau_{\pi(j)}$. If such $j$ does not exist, $j$ is set to be $n + 1$ (all boxes opened). The utility we receive in this round is $\max\{X_\pi(1),\ldots, X_\pi({j-1})\} - \sum_{k<j} c_{\pi(k)}$. }

Note that the  only  feedback  is the utility, and we do not see any value or even the index $j$ where we stop. 

  \IGNORE{
\begin{itemize}
    \item Step 1: Start from $i = 1$ and $v = 0$.
    \item Step 2: Check whether we have $v > \tau_{\pi(i)}$. If so, stop and return the reward. 
    \item Step 3: Otherwise, pay $c_{\pi(i)}$ to draw $X_{\pi(i)}$ from $\D_{\pi(i)}$ and update $v = \max\{v, X_{\pi(i)}\}$, then go back to Step 2 with $i \gets i + 1$.
\end{itemize}

Finally, the algorithm can only receive the reward in each round, i.e., the maximum value $v$ minus the total cost of those opened distributions.
}

In the case of  known distributions,  the optimal one-round policy for this problem was designed by Weitzman~\cite{Weitzman-Econ79}: For every distribution $\D_i$, solve the equation $\Ex{\max\{X_i - \sigma_i, 0\}} = c_i$; now play permutation $\pi$ by sorting in decreasing order of $\sigma_i$ and set threshold $\tau^*_i = \sigma_i$.  
Let $OPT$ be the optimal expected reward according to this optimal policy. Let $ALG_t$ be the expected reward of our policy in the $t$-th round. Then, we want to design an algorithm with \emph{total regret} $T\cdot OPT - \sum_{t \in T} ALG_t $ at most $\widetilde{O}(\poly(n) \sqrt{T})$.

Before introducing the algorithm, we define the \textit{gain function} for this general case:
\begin{align} \label{eq:gainPBGen}
g_i(v)  ~&:=~ -c_i + \int_v^1 (x - v)f_i(x) dx 
~=~ -c_i + (1-v) - \int_v^1 F_i(x)dx. 
\end{align}
Similar to the $n = 2$ case, this \textit{gain function} is the expected additional utility we get on opening $X_i$ when we already have value $v$ in hand. Note that the optimal threshold $\tau^*_i$ satisfies $g_i(\tau^*_i) = 0$.

\subsection{High-Level Approach via Valid Policies.}

We first briefly introduce the initialization algorithm. The following lemma shows what we achieve in the initialization  (proved in \Cref{sec:appendixSec4Lemma1}).

\begin{Lemma}
\label{PBInit}
The initialization algorithm runs $1000 \cdot \sqrt{T}\log T$ samples for each distribution to output interval $[\lcb_i,\ucb_i]$, such that with probability $1 - T^{-10}$ the following hold simultaneously for all $i \in [n]$:
\begin{itemize}
    \item $\lcb_i \leq \sigma_i \leq \ucb_i$.
    \item $|g_i(x)| \leq T^{-\frac{1}{4}}$ simultaneously for all $x \in [\lcb_i, \ucb_i]$.
\end{itemize}
\end{Lemma}

After initialization, the main part  is the action set-updating algorithm. Similar to the algorithm for $n = 2$, we hope to use estimates $\hat{F}_i(x)$ to gradually shrink the intervals $[\lcb_i, \ucb_i]$. However, one major challenge is that we don't have a fixed order. If $n$ is a constant, we can just simply try all possible permutations and use a multi-armed bandit style algorithm to find the optimal permutation. But the number of permutations is exponential in $n$, so this approach is impossible when $n$ is a general parameter. To get a polynomial regret algorithm, we can only test $\poly(n)$ number of different orders.

Another challenge is that the idea for $n = 2$ can bound the regret when we play a sub-optimal threshold, but it tells nothing about playing a sub-optimal order. We don't have a direct way to bound the regret when playing an incorrect order.

Both difficulties imply that only keeping the confidence intervals as the constraint for the actions is not enough. Therefore, we also introduce a set of order constraints:

\begin{Definition}[Valid Constraint Group]
Given a set of confidence intervals $I = \{[\lcb_1,\ucb_1], [\lcb_2, \ucb_2],...,$ $[\lcb_n, \ucb_n]\}$ and a set $S$ of order constraints, satisfying:
\begin{itemize}
    \item $\ucb_i - \lcb_i \leq T^{-\frac{1}{4}}$.
    \item $\sigma_i \in [\lcb_i, \ucb_i]$
    \item Every constraint in $S$ can be defined as $(i, j)$ that means $\sigma_i > \sigma_j$.
    \item The constraints in $S$ are closed, i.e., if $(i, j), (j, k) \in S$, there must be $(i, k) \in S$.
    \item If $(i, j) \in S$, we must have $\ucb_i \geq \ucb_j$ and $\lcb_i \geq \lcb_j$.
\end{itemize}
For $(I, S)$ satisfying the conditions above, we call it a \emph{valid constraint group}.
\end{Definition}

The intuition of the extra order constraints is: When we are shrinking the intervals, if it is evident that $\sigma_i > \sigma_j$, we will require $\D_i$ to be in front of $\D_j$ in the following rounds. Correspondingly, we give the following definition for a ``valid" policy. During the algorithm, we will only run valid policies, according to the current constraint group we have.

\begin{Definition}[Valid Policy]
\label{validThSetting}
Let $(\tau_{\pi(1)}, \tau_{\pi(2)}, ..., \tau_{\pi(n)})$ be a policy to play in one round, where $\pi$ is the distribution permutation for this policy, and the threshold in front of box $\pi(i)$ is $\tau_{\pi(i)}$. For simplicity, we use $\pi$ to represent a policy.

For a policy $\pi$, we say it is \emph{valid} for a constraint group $(I, S)$ if the following conditions hold:

\begin{itemize}
    \item For $i \in [n]$, $\tau_{\pi(i)} \in [\lcb_{\pi(i)}, \ucb_{\pi(i)}]$.
    \item If $(i, j) \in S$, then $\D_i$ must be in front of $\D_j$, i.e., $\pi^{-1}(i) < \pi^{-1}(j)$.
    \item For $i < j$, $\tau_{\pi(i)} \geq \tau_{\pi(j)}$.
\end{itemize}
\end{Definition}

Notice that for a valid constraint group, we have $\sigma_i \in [\lcb_i, \ucb_i]$ for all $i \in [n]$, and $\sigma_i > \sigma_j$ for all $(i, j) \in S$. Then, the optimal policy is valid. Therefore, we can always find a valid policy from the constraint group.

Now, we are ready to give the main idea of the constraint-updating algorithm. In each phase, we first update the confidence intervals and then update the order constraints as follows: 
\begin{itemize}
    \item Step 1: For each $i \in [n]$, we run $\OTild(\frac{\poly(n)}{\epsilon^2})$ samples to update the confidence interval to $[\lcb'_i, \ucb'_i]$, such that for every threshold pair $\tau_i, \tau'_i \in [\lcb'_i, \ucb'_i]$, the \textit{moving difference} is small, i.e., if we move $\tau_i$ to $\tau'_i$ and keep the validity, the difference of the expected reward is bounded by $O(\poly(n)\cdot \epsilon)$.
    \item Step 2: For each distribution pair $(i, j)$ without a constraint, we run $\OTild(\frac{\poly(n)}{\epsilon^2})$ samples to test the order between them, such that we can either clarify which one is bigger between $\sigma_i$ and $\sigma_j$, or we can claim that the \textit{swapping difference} (the difference before and after swapping $\D_i$ and $\D_j$) is bounded by $O(\poly(n)\cdot \epsilon)$.
\end{itemize}

Finally, we argue that for every valid policy, we can convert it into the optimal policy by using $\poly(n)$ number of moves and swaps. This is sufficient for us to give  $O(\poly(n)\cdot \epsilon)$ regret bound.

In the following analysis, we use separate sub-sections to introduce each part. \Cref{step1} provides the Interval-Shrinking algorithm to bound the moving difference. \Cref{step2} introduces the way to add a new order constraint to bound the swapping difference. \Cref{step3} shows how to convert a valid policy to the optimal policy using a $\poly(n)$ number of moves and swaps. Finally, \Cref{combinepan} combines the results of three sub-sections to complete the analysis.

\subsection{Step 1: Interval-Shrinking to Bound Moving Difference}\label{step1}

The goal of this sub-section is: Given $i\in [n]$ and an original constraint group $(I, S)$, we want to update the confidence interval $[\lcb_i, \ucb_i]$, to make sure that moving $\tau_i$ inside the new confidence interval incurs a small difference. The key idea of the Interval-Shrinking algorithm is similar to the case when $n = 2$: For each $i \in [n]$, we want to play two different values for $\tau_i$, and see the difference of the expected reward. However, playing $\tau_i = \lcb_i$ and $\tau_i = \ucb_i$ might be impossible. The reason is: We hope to keep a decreasing threshold setting. There may not be a policy that allow $\tau_i$ to be set to $\ucb_i$ and $\lcb_i$ without changing other thresholds. If we need different permutations to test $\tau_i = \ucb_i$ and $\tau_i = \lcb_i$, this makes the analysis involved. Therefore, we should find a policy that fixes the order and other thresholds, then test $\tau_i$ under this fixed policy while keeping a decreasing thresholds.

When we set $\tau_i$ to be different values, the two policies will be different only when the maximum reward before $\tau_i$ falls between the two thresholds. Therefore, to see the largest difference, we hope the probability of this event is maximized. This intuition allows us to give the following definition:

\begin{Definition}[\clever  Policy]
\label{clever}

Given $(I, S)$ and $i \in [n]$, a \clever policy is a valid \textit{partial} policy $\pi$ parameterized by $\lcb$ and $\ucb$\footnote{Here, we say $\pi$ is a partial policy because it's not completely fixed. We fix the permutation of the distributions and the value of all other thresholds, but the value of $\tau_i$ is flexible.}, such that $F_{\pi, i}(\ucb) - F_{\pi, i}(\lcb)$ is maximized.

In the definition, $F_{\pi, i}(x)$ is the probability that the algorithm reaches distribution $X_i$ with maximum value $v < x$ in hand, i.e.,
\[
  F_{\pi, i}(x) ~:=~ \prod_{j < \pi^{-1}(i)} F_{\pi(j)}(x).
\]

Furthermore, $\ucb$ and $\lcb$ represents two possible value of $\tau_i$ to keep a valid $\pi$, i.e.,
$\pi$ is valid when both $\tau_i = \ucb$ and $\tau_i = \lcb$.
\end{Definition}

A key fact of \clever policy is that for every different distribution, we might find a different \clever policy. This is different from the Prophet Inequality problem: In the Pandora's Problem, we don't keep a fixed order. Every order that satisfies the constraints $(I,S)$ is possible to be tested.

Now, the key idea of the Interval-Shrinking algorithm is clear: For each $i$, find the \clever policy and run samples with $\tau_i = \ucb$ and $\tau_i = \lcb$. Then, use a method similar to \Cref{ISAPan} to calculate the new interval. The following algorithm describes the details of this idea:

\begin{algorithm}[tbh]
\caption{Interval-Shrinking Algorithm}
\label{PBISA}
\KwIn{$(I, S)$, $\epsilon$, $i$, $\hat{F}_1(x),...,\hat{F}_n(x)$}
Get an \textit{approximate} \clever policy $\hat \pi$ and $\lcb, \ucb$ using \Cref{Lemma:Clever}.  \\
Calculates $\hat F_{\hat \pi, i}(x)$. \\ 
For $\tau \in [\lcb_i, \ucb_i]$, let $\hat{\Delta}_i(\tau) := \hat{F}_{\hat \pi, i}(\ucb) \int_\tau^\ucb (\hat{F}_i(x)-1)dx + \hat{F}_{\hat \pi, i}(\lcb) \int_\lcb^\tau (\hat{F}_i(x)-1) dx - \int_\lcb^\ucb \hat{F}_{\hat \pi, i}(x)(\hat{F}_i(x)-1)dx$. \\
Run $C \cdot \epsilon^{-2}\log T$ samples with $\tau_i = \ucb$. Let the average reward be $\hat{R}_\ucb$. \\
Run $C \cdot \epsilon^{-2}\log T$ samples with $\tau_i = \lcb$. Let the average reward be $\hat{R}_\lcb$. \\
Define $\hat{\delta}_i(\tau) := \hat{\Delta}_i(\tau) - (\hat{R}_\ucb - \hat{R}_\lcb)$. \\
Let $\ucb'_i = \max_{\tau \in [\lcb_i, \ucb_i]} |\hat{\delta}_i(\tau)| < \epsilon$  and let $\lcb'_i = \min_{\tau \in [\lcb_i, \ucb_i]} |\hat{\delta}_i(\tau)| < \epsilon$. \\
\KwOut{$[\lcb'_i, \ucb'_i]$}
\end{algorithm}

Then, the following lemma shows the bound when modifying a threshold:

\begin{Lemma}[Moving Difference Bound]
\label{MoveOP}

Suppose we are given $(I, S)$, $\epsilon > 16 T^{-\frac{1}{2}}$, CDF estimates $\hat F_1(x), \cdots, \hat F_n(x)$, and $i \in [n]$, satisfying the following conditions for all $j \in [n]$:
\begin{itemize}
\item $|g_j(\tau)| \leq T^{-\frac{1}{4}}$ for all $\tau \in [\lcb_j, \ucb_j]$.
\item $(I, S)$ is valid.
\item For any valid partial policy $\pi'$ of $(I, S)$, we fix the order and the other thresholds except $\tau_j$. Assume $\pi'$ is valid when both $\tau_j = \lcb'$ and $\tau_j = \lcb'$. Define $\delta_{\pi', \ucb', \lcb', j}(\tau) = (F_{\pi', j}(\lcb') - F_{\pi', j}(\ucb'))g_i(\tau)$. Then $|\delta_{\pi', \ucb', \lcb', j}(\tau)| \leq 6\epsilon$.
\item CDF estimate $\hat F_j(x)$ is constructed via $10^5 \cdot \frac{n^2 \log T}{\epsilon}$ fresh i.i.d. samples of $X_j$.
 \end{itemize}
Then, \Cref{PBISA} runs $O(\frac{\log T}{\epsilon^2})$ samples and calculates a new interval $[\lcb'_i, \ucb'_i]$, such that the following properties hold with probability $1 - T^{-11}$: 
\begin{enumerate}[label=(\roman*)]
    \item \label{MoveSta1} $\sigma_i \in [\lcb'_i, \ucb'_i]$
    \item \label{MoveSta3} Let $I'_i = (I \setminus \{[\lcb_i, \ucb_i]\}) \cup \{[\lcb'_i, \ucb'_i]\} $. For any valid partial policy $\pi'$ of $(I'_i, S)$, we fix the order and the other thresholds. Assume $\pi'$ is valid when both $\tau_i = \ucb'$ and $\tau_i = \lcb'$. Define $\delta_{\pi', \ucb', \lcb', i}(\tau) = (F_{\pi', i}(\lcb') - F_{\pi', i}(\ucb'))g_i(\tau)$. Then $|\delta_{\pi', \ucb', \lcb', i}(\tau)| \leq 3\epsilon$.
    \item \label{MoveSta2}  For any valid policy of $(I'_i, S)$, if we fix the order and the other thresholds, but modify $\tau_i$ to $\tau'_i$, satisfying that the new policy is still valid, the difference of the expected reward between these two policies is less than $3\epsilon$.
\end{enumerate}
\end{Lemma}

Before starting the proof, we first give an accuracy bound of the distribution estimates, which is proved in \Cref{sec:appendixSec4Lemma2}.

\begin{Claim}
\label{PBGenDisAcc}
Assume the preconditions in \Cref{MoveOP} hold. Then with probability $1 - T^{-12}$, we have
  $|\prod_{i \in S} \hat{F}_i(x) - \prod_{i \in S} F_i(x)| \leq \sqrt{\epsilon}$ simultaneously hold for all $x \in [0,1]$ and $S \subseteq[n]$.
\end{Claim}

\begin{proof}[Proof of \Cref{MoveOP}]
  Fix the \clever policy $\pi$. Assume we want to move $\tau_i$ from $\tau_i = \ucb$ to $\tau_i = \lcb$, such that the policies are both valid when $\tau_i = \lcb$ and $\tau_i = \ucb$. Since we only care about the absolute value of the difference between two expected rewards, we may assume $\ucb > \lcb$.

If moving $\tau_i$ from $\ucb$ to $\lcb$, the performance of the two policies will only be different if the previous maximum reward falls between $\lcb$ and $\ucb$: It will reject the previous maximum if $\tau_i = \ucb$, but accept it when $\tau_i = \lcb$. Besides, since $\lcb$ is greater than the next threshold in $\pi$, when the previous maximum is inside $[\lcb, \ucb]$, the algorithm must stop before the next threshold, which means the difference only comes from $\tau_i$ and $X_i$.

Recall that $F_{\hat \pi,i}(x) = \prod_{j < \hat \pi^{-1}(i)} F_j(x)$, i.e., $F_{\hat \pi(i)}(x)$ is the probability that \Cref{ISAgen} reaches $\tau_i$ with $v \leq x$ in hand. Let $f_{\hat \pi, i}(x) = F'_{\hat \pi, i}(x)$. Then, the difference of the expected reward is $\int_\lcb^\ucb f_{\hat \pi, i}(x)g_i(x) dx = F_{\hat \pi, i}(\ucb)g_i(\ucb) - F_{\hat \pi, i}(\lcb)g_i(\lcb) - \int_\lcb^\ucb F_{\hat \pi, i}(x)g'_i(x)dx$. To upper-bound this difference,  define \emph{generalized bounding function}
\begin{align}
    \delta_i(\tau) ~:=~  -(F_{\hat \pi, i}(\ucb) - F_{\hat \pi, i}(\lcb)) \cdot g_i(\tau). \label{EqPan1}
\end{align}
Then, to learn $\delta_i(\tau)$, we define
\begin{align*}
    \Delta_i(\tau) &~:=~ F_{\hat \pi, i}(\ucb)(g_i(\ucb) - g_i(\tau)) - F_{\hat \pi, i}(\lcb)(g_i(\lcb)-g_i(\tau)) - \int_\lcb^\ucb F_{\hat \pi, i}(x)g'_i(x)dx \notag \\
    &~=~ F_{\hat \pi, i}(\ucb) \int_\tau^\ucb (F_i(x) - 1)dx + F_{\hat \pi, i}(\lcb) \int_\lcb^\tau (F_i(x) - 1) dx - \int_\lcb^\ucb F_{\hat \pi, i}(x)(F_i(x) - 1)dx.
\end{align*}
Observe that $\delta_i(\tau) = \Delta_i(\tau) - (R_\ucb - R_\lcb)$, where $R_\ucb$ and $R_\lcb$ correspond to the expected reward in $\hat \pi$ with $\tau_i = \ucb$ and $\tau_i = \lcb$ respectively. Then, by replacing $F_i(x)$ with $\hat{F}_i(x)$, we can get $\hat{\Delta}_i(\tau)$, which is an estimate of $\Delta_i(\tau)$. For $R_\ucb$ and $R_\lcb$, we can learn the estimates $\hat{R}_\ucb$ and $\hat{R}_\lcb$ via running samples. Combining these estimates results in $\hat{\delta}_i(\tau)$. Then, the following claim shows that $\hat{\delta}_i(\tau)$ estimates $\delta_i(\tau)$ accurately (proved in \Cref{sec:appendixSec4Lemma3}).

\begin{Claim}
\label{AccBoundPBGen}
In \Cref{PBISA}, if the conditions in \Cref{MoveOP} holds, then with probability $1 - T^{-12}$ we have $|\hat{\delta}_i(\tau) - \delta_i(\tau)| \leq \epsilon$ simultaneously for all $\tau \in [\lcb_i, \ucb_i]$.
\end{Claim}

Now we prove the statements in \Cref{MoveOP}. In the following proofs, we assume $|\delta_i(\tau) - \hat{\delta}_i(\tau)| \leq \epsilon$  holds simultaneously for all $\tau \in [\lcb_i, \ucb_i]$. 

\noindent\textit{Statement \ref{MoveSta1}.} Look at \Cref{PBISA}: It finds $\hat \pi, \lcb, \ucb$, gets $\hat{\delta}_i(\tau)$, then calculates $[\lcb'_i, \ucb'_i] = \{\tau \in [\lcb_i, \ucb_i]: |\hat{\delta}_i(\tau) < \epsilon|\}$. Since $\delta_i(\tau^*_i) = 0$, there must be $|\hat{\delta}_i(\tau^*_i)| \leq \epsilon$. Therefore, $\tau^*_i \in [\lcb'_i, \ucb'_i]$.

\noindent\textit{Statement \ref{MoveSta3}.} Notice that $[\lcb'_i, \ucb'_i] = \{\tau \in [\lcb_i, \ucb_i]: |\hat{\delta}_i(\tau) \leq  \epsilon|\}$. Therefore, for all $\tau \in [\lcb'_i, \ucb'_i]$, $|\delta_i(\tau)| \leq |\hat{\delta}_i(\tau)| + |\delta_i(\tau) - \hat{\delta}_i(\tau)| \leq 2\epsilon$. We first assume that $\hat \pi$ is an \textit{accurate} \clever policy. Then,  from the definition, we have $F_{\hat \pi, i}(\ucb) - F_{\hat \pi, i}(\lcb) \geq F_{\pi', i}(\ucb') - F_{\pi', i}(\lcb')$ for all valid partial policy $\pi'$ parameterized by $\ucb', \lcb'$. Therefore, $|\delta_{\pi', \ucb', \lcb',i}(\tau)| \leq |\delta_i(\tau)| \leq 2\epsilon$.

\noindent\textit{Statement \ref{MoveSta2}.} We again assume that $\hat \pi$ is an accurate \clever policy. Recall that we just proved $|\delta_i(\tau)| \leq 2\epsilon$. Combining this with (\ref{EqPan1}), we have $|g_i(\tau)| \leq \frac{2\epsilon}{(F_{\hat \pi, i}(\ucb) - F_{\hat \pi, i}(\lcb))}$. 

Now, consider the policy $\pi'$. Assume we first have $\tau_i = \ucb'$ and we want to move it to $\tau_i = \lcb'$, satisfying $\lcb', \ucb' \in [\lcb'_i, \ucb'_i]$ and $\pi'$ is valid when both $\tau_i = \lcb'$ and $\tau_i = \ucb'$. Then, the difference of the expected reward is $\left|\int_{\lcb'}^{\ucb'} f_{\pi', i}(x)g_i(x)dx\right|$, and we have the following bound: 
\begin{align}
    \left|\int_{\lcb'}^{\ucb'} f_{\pi', i}(x)g_i(x)dx\right| &~\leq~ |F_{\pi', i}(\ucb') - F_{\pi', i}(\lcb')|\max_{v \in [\lcb', \ucb']}|g_i(v)| ~\leq~ 2\epsilon, \label{EqPan2}
\end{align}
where in the last inequality we use the fact that $F_{\hat \pi, i}(\ucb) - F_{\hat \pi, i}(\lcb) \geq |F_{\pi', i}(\ucb') - F_{\pi', i}(\lcb')|$ when $\pi$ is a \clever policy, and $|g_i(v)| \leq \frac{2\epsilon}{(F_{\pi, i}(\ucb) - F_{\pi, i}(\lcb))}$ for all $v \in [\lcb'_i, \ucb'_i]$.
This gives an upper bound on the difference of the expected reward when we want to move $\tau_i$.

The remaining part is to show how to get a \clever policy. However, since we only have CDF estimates $\hat F_i(x)$ instead of an accurate $F_i(x)$, there is no hope to get an accurate \clever policy. The following Lemma then shows that we can calculate an approximate \clever policy:

\begin{Lemma}
\label{Lemma:Clever}
There exists an algorithm with time complexity $O(n\cdot 2^n)$ that calculates a \clever policy with an extra $4\sqrt{\epsilon}$ additive error.
\end{Lemma}
{ We leave the details of the algorithm and the proof to \Cref{sec:appendixSec4Lemma4}.}

Finally, we show that this $4\sqrt{\epsilon}$ error doesn't hurt too much for both Statement \ref{MoveSta3} and \ref{MoveSta2}. Define
\[
q_i ~:=~ \max_{\pi} F_{\pi, i}(\ucb) - F_{\pi, i}(\lcb), \quad \text{and} \quad
\hat{q}_i ~:=~ F_{\hat \pi, i}(\ucb) - {F}_{\hat \pi, i}(\lcb),
\]
where $\hat \pi$ is the approximate \clever policy we get via \Cref{Lemma:Clever}. Then, we have $q_i \leq \hat{q}_i + 4\sqrt{\epsilon}$. For Statement \ref{MoveSta3}, we have
\begin{align*}
    |\delta_{\pi', \ucb', \lcb', i}(\tau)| ~\leq~ q_i \max_{v \in [\lcb'_i, \ucb'_i]} |g_i(v)| &~\leq~  (\hat{q}_i + 4\sqrt{\epsilon}) \cdot \max_{v \in [\lcb'_i, \ucb'_i]} |g_i(v)| \\
    &~\leq~ \hat q_i \cdot \frac{\max_{v \in [\lcb'_i, \ucb'_i]} |\delta_i(v)|}{\hat q_i} + 4\sqrt{\epsilon} \cdot T^{-1/4} \\
    & ~\leq~ 2\epsilon + 4\sqrt{\epsilon} \cdot T^{-\frac{1}{4}} ~<~ 3\epsilon.
\end{align*}
Here, the second line follows the definition of $\delta_i(v)$ and the precondition in \Cref{MoveOP}. The third line holds because the condition $|\delta_i(v)| \leq 2\epsilon$ does not require $\hat \pi$ to be accurate, and the last inequality
 holds when $\epsilon > 16T^{-\frac{1}{2}}$. 
 
 For Statement \ref{MoveSta2}, following (\ref{EqPan2}), we can bound the moving difference to
\[
    \left|\int_{\lcb'}^{\ucb'} f_{\pi', i}(x)g_i(x)\right| ~\leq~ q_i \max_{v \in [\lcb'_i, \ucb'_i]} |g_i(v)|.
\]
Therefore, the same $3\epsilon$ bound holds.
\end{proof}

\subsection{Step 2: Updating Order Constraints to Bound Swapping Difference}\label{step2}

In this section, our goal is to verify $\sigma_i$ and $\sigma_j$ which one is larger, or claiming that reversing the order of $X_i$ and $X_j$ doesn't hurt too much. We first provide the following lemma, which shows the difference of the expected reward when we swap two distributions with a same threshold:

\begin{Lemma}
\label{SwapLemma1}
For a policy $\pi$, such that $X_i$ and $X_j$ are consecutive with $\tau_i = \tau_j = \tau$, let $\Delta_{\pi, i,j}(\tau)$ be the change of the expected reward after swapping $X_i$ and $X_j$, then
\begin{align*}
    \Delta_{\pi, i, j}(\tau) ~=~ F_{\pi, i}(\tau)(g_i(\tau)(1 - F_j(\tau)) - g_j(\tau)(1-F_i(\tau))).
\end{align*}
\end{Lemma}

\begin{proof}
Assume we have value $v$ in hand before arriving $X_i$ and $X_j$. To pass the threshold, there must be $v \leq \tau$. If $X_i$ is in the front, the expected gain of opening $X_i$ is $g_i(v)$. After that, if $X_i < \tau$, we can play $X_j$ as well. The expected gain is $F_i(v)g_j(v) + \int_v^\tau f_i(x)g_j(x)dx = F_i(\tau)g_j(\tau) - \int_u^\tau F_i(x)g'_j(x)dx$. Therefore, the total expected gain from $X_j$ and $X_i$ is $g_i(v) + F_i(\tau)g_j(\tau) - \int_u^\tau F_i(x)g'_j(x)dx$. Similarly, if $X_j$ is in the front, the total expected gain from $X_i$ and $X_j$ is $g_j(v) + F_j(\tau)g_i(\tau) - \int_u^\tau F_j(x)g'_i(x)dx$.

Notice that the order of $X_i$ and $X_j$ doesn't affect the expected gain from the distributions behind $X_i$ and $X_j$. Therefore, the difference of the gain from $X_i$ and $X_j$ is exactly the difference of the expected reward: 
\begin{align*}
     & \left(g_i(v) + F_i(\tau)g_j(\tau) - \int_u^\tau F_i(x)g'_j(x)dx \right) - \left(g_j(v) + F_j(\tau)g_i(\tau) - \int_u^\tau F_j(x)g'_i(x)dx \right) \\
     &~=~ g_i(v) + F_i(\tau)g_j(\tau) -g_j(v) - F_j(\tau)g_i(\tau) + \int_u^\tau \left(F_j(x)(F_i(x) - 1) - F_i(x)(F_j(x) - 1)\right)dx \\
     &~=~ \left(g_i(v) + u - \tau + \int_u^\tau F_i(x)dx\right) - \left(g_j(v) + u - \tau + \int_u^\tau F_j(x)dx\right) + F_i(\tau)g_j(\tau) - F_j(\tau)g_i(\tau) \\
     &~=~ g_i(\tau)(1 - F_j(\tau)) - g_j(\tau)(1-F_i(\tau)).
\end{align*}
Since the probability that $v$ arrives with $v < \tau$ is exactly $F_{\pi, i}(\tau)$, the expected difference is $\Delta_{\pi, i,j}(\tau) = F_{\pi, i}(\tau)(g_i(\tau)(1 - F_j(\tau)) - g_j(\tau)(1-F_i(\tau)))$.
\end{proof}

\Cref{SwapLemma1} shows the following properties:

\begin{enumerate}
    \item \label{SwapLemmaP1} Assume $\sigma_i > \sigma_j$. When $\tau \in [\sigma_j, \sigma_i]$,  $\Delta_{\pi, i,j}(\tau) < 0$, i.e., letting $X_i$ be in the front is better. This implies: If we know the sign of $\Delta_{\pi, i, j}(\tau)$, and we are sure that $\tau$ is between $\sigma_i$ and $\sigma_j$, then we can determine that $\sigma_i$ and $\sigma_j$ which one is greater.
    \item \label{SwapLemmaP2} Fix $i,j, \tau$, $|\Delta_{\pi, i, j}(\tau)|$ is maximized when $F_{\pi, i}(\tau)$ is maximized.
\end{enumerate}
    
According to Property \ref{SwapLemmaP2}, we hope to test $X_i$ and $X_j$ with a policy $\pi$ that maximizes $F_{\pi, i}(\tau)$. If the difference is bounded when $F_{\pi, i}(\tau)$ is maximized, the swapping difference is bounded in all policies. Inspired by this, we give the definition of the \smart policy:

\begin{Definition}[\smart Policy]
\label{smart}
Given $(I,S)$ and $i,j \in [n]$ with $i \neq j$ and $(i, j), (j, i) \notin S$. Assume we have $[\lcb'_i, \ucb'_i], [\lcb'_j, \ucb'_j] \in I$. A \smart policy is a pair of valid policies $(\pi, \pi')$, such that
\begin{itemize}
    \item $\tau_i = \tau_j = \max\{\lcb'_i, \lcb'_j\}$.
    \item $X_i$ and $X_j$ are adjacent in both $\pi$ and $\pi'$, but under different orders, and this is the only difference between $\pi$ and $\pi'$. W.l.o.g, assume $X_i$ is in the front in $\pi$, while $X_j$ is in the front in $\pi'$, i.e., $\pi^{-1}(i) = \pi^{-1}(j) - 1$, and $\pi'^{-1}(j) = \pi'^{-1}(i) - 1$.
    \item The \smart policy maximizes $F_{\pi, i}(\tau)$ when the first two conditions are satisfied.
\end{itemize}
\end{Definition}

Then, the algorithm for testing $X_i$ and $X_j$ is clear: We find the \smart policy for $X_i$ and $X_j$, run some samples for two policies and see the difference. If the difference is too large, we can verify $\sigma_i$ and $\sigma_j$ which one is larger. Otherwise, we can bound the swapping difference. \Cref{PBOCA} gives the details of this idea.

\begin{algorithm}[tbh]
\caption{$\mathsf{SwapTest}$ Algorithm}
\label{PBOCA}
\KwIn{Distribution indices $i$ and $j$}
Run \Cref{CalcSmart} to get \smart policy $(\pi, \pi')$ \\
Run $C \cdot \frac{\log T}{n^2\epsilon^2}$ samples with policy $\pi$. Let $\hat{R}_{i,j}$ be the average reward. \\
Run $C \cdot \frac{\log T}{n^2\epsilon^2}$ samples with policy $\pi'$. Let $\hat{R}_{j,i}$ be the average reward. \\
\If{$|\hat{R}_{i, j} - \hat{R}_{j, i}| > 40n\epsilon$}
{
Add constraint $(i, j)$ into $S'$ if $\hat{R}_{i, j} > \hat{R}_{j, i}$, otherwise add constraint $(j, i)$ into $S'$. \\
Update $S'$ according to the transitivity. Update $I'$ according to the new order constraints, i.e., when adding a constraint $(a, b)$, let $\ucb'_b \gets \min\{\ucb'_a, \ucb'_b\}$ and $\lcb'_a \gets \max\{\lcb'_a, \lcb'_b\}$.
}
\KwOut{Updated constraint group $(I', S')$}
\end{algorithm}

\begin{algorithm}[tbh]
\caption{Finding \smart Policy}
\label{CalcSmart}
\KwIn{\textbf{Input:} $(I', S')$, m, $i, j$}
Let $\tau = \tau_i = \tau_j = \max\{\lcb'_i, \lcb'_j\}$ \\
Let $T = \{k|(k, i) \in S' \lor (k, j) \in S' \lor \lcb'_k > \tau \}$.\\
For $k \in T$, let $\tau_k = \ucb'_k$ \\
For $k \in [n] \setminus (\{i, j\} \cup T)$, let $\tau_k = \lcb'_k$ \\
Let $\pi$ and $\pi'$ be two policies that sort the distributions in a decreasing threshold order, and break ties according to $S'$. The only difference is: $X_i$ is in front of $X_j$ in $\pi$, but $X_j$ is in front of $X_i$ in $\pi'$. \\
\KwOut{$\pi$ and $\pi'$}
\end{algorithm}

Before analysing the algorithm, we point out two facts of \Cref{PBOCA}:

\begin{itemize}
    \item \Cref{PBOCA} relies on \Cref{PBISA}, i.e., we need to first run \Cref{PBISA} to get $n$ new confidence intervals, then run \Cref{PBOCA} to update order constraints. This is critical to the regret analysis. 
    \item In the \smart algorithm, we only test the swapping difference with $\tau_i = \tau_j = \max\{\lcb'_i, \lcb'_j\}$, and give the difference bound only with this threshold. This is sufficient for our regret analysis.
\end{itemize}

\begin{Lemma}[Swapping Difference Bound]
\label{SwapOP}
Given $(I',S)$, $\epsilon$, and $i,j \in [n]$ with $i \neq j$ and $(i, j), (j, i) \notin S$, where $I'$ is generated by \Cref{PBISA}. Assume the preconditions in \Cref{PBInit} hold. \Cref{PBOCA} runs $O(\frac{\log T}{n^2\epsilon^2})$ samples and achieves one of the following:
\begin{itemize}
    \item Clarify $\sigma_i$ and $\sigma_j$ which one is bigger with probability $1 - T^{-12}$, and give a new constraint $(i, j)$ or $(j, i)$. 
    \item Make the following claim with probability $1 - T^{-12}$: For every two valid policies of $(I', S)$, satisfying:
    \begin{itemize}
        \item $\tau_i = \tau_j = \max\{\lcb'_i, \lcb'_j\}$.
        \item $X_i$ and $X_j$ are consecutive in both policies but in a different order. This is the only difference between two policies.
    \end{itemize} 
    The difference of the expected reward between these two policies is no more than $60n\epsilon$.
\end{itemize}
\end{Lemma}

\begin{proof}
We first prove the theorem assuming \Cref{CalcSmart} returns an accurate \smart policy $(\pi, \pi')$. According to the definition of \smart policy, $\pi$ and $\pi'$ maximizes the probability of reaching $X_i$ and $X_j$ when $\tau_i = \tau_j = \max\{\lcb'_i, \lcb'_j\}$. According to Property \ref{SwapLemmaP2}, for any valid policy, such that $X_i$ and $X_j$ are consecutive with $\tau_i = \tau_j = \max\{\lcb'_i, \lcb'_j\}$, the swapping difference is no more than the difference between $\pi$ and $\pi'$. Therefore, if we are evident that the difference between $\pi$ and $\pi'$ is no more than $60n\epsilon$, we can claim that this upper bounds the swapping difference between $X_i$ and $X_j$ for any other policy. The proof idea is the following: We run multiple samples to estimate $R_{i, j}$ and $R_{j, i}$, where $R_{i, j}$ is the expected reward of $\pi$ and $R_{j, i}$ is the expected reward of $\pi'$. Next, we show that $|R_{i, j} - \hat{R}_{i,j}| \leq 10n\epsilon$ and $|R_{j, i} - \hat{R}_{j, i}| \leq 10n\epsilon$ with probability $1 - T^{-12}$. Then, $|R_{i, j} - R_{j, i}| \leq 60n\epsilon$ when $|\hat{R}_{i, j} - \hat{R}_{j, i}| \leq 40n\epsilon$.

Now, we bound $|R_{i, j} - \hat{R}_{i,j}|$ with Hoeffding's Inequality (\Cref{Hoeffding}). $\hat{R}_{i, j}$ is an estimate of $R_{i,j}$ by running $N = C\cdot \frac{\log T}{n^2\epsilon^2}$ samples, and the per-round reward is bounded by $[-0.5, 0.5]$. Then, $\Pr{|R_{i, j} - \hat{R}_{i,j}| > 10n\epsilon} < 2\exp(-2N\cdot 100n^2\epsilon^2/4) = 2T^{-50C}$. Hence, $|R_{i, j} - \hat{R}_{i,j}| \leq  10n\epsilon $ with probability $1 - T^{-13}$ when $C > 10$. Bounding $|R_{j,i} - \hat{R}_{j,i}|$ is identical, and by the union bound, $|R_{i, j} - \hat{R}_{i,j}| \leq 10n\epsilon$ and $|R_{j, i} - \hat{R}_{j, i}| \leq 10n\epsilon$ simultaneously hold with probability $1 - T^{-12}$.

The concentration proof above also shows that when $|\hat{R}_{i, j} - \hat{R}_{j, i}| > 40n\epsilon$, we can claim that w.h.p. $|R_{i, j} - R_{j, i}| > 20n\epsilon$.
Next, we show that this is evident to clarify which of $\sigma_i$ and $\sigma_j$ is greater. We first introduce a special case to give the intuition: Consider the case that all other confidence intervals are disjoint with $[\lcb'_i, \ucb'_i]$ or $[\lcb'_j, \ucb'_j]$. W.l.o.g., assume $\pi$ ($X_i$ in the front) is better than $\pi'$ ($X_j$ in the front). If $\tau = \max\{\lcb'_i, \lcb'_j\}$ is between $\sigma_i$ and $\sigma_j$, we can immediately claim that $\sigma_i > \sigma_j$ according to Property \ref{SwapLemmaP1}. If $\tau$ doesn't fall between $\sigma_i$ and $\sigma_j$, there must be $\tau < \min\{\sigma_i, \sigma_j\}$. Then, we adjust $\pi$ and $\pi'$ by increasing $\tau_i$ and $\tau_j$ to $\min\{\sigma_i, \sigma_j\}$. According to \Cref{MoveOP}, these operations do not change the expected reward too much: Since we move two thresholds in each policy, the expected reward of $\pi$ can decrease by at most $6\epsilon$, and the expected reward of $\pi'$ can increase by at most $6\epsilon$. Therefore, if the original $\pi$ is at least $20\epsilon$ better than $\pi'$, we can still claim that $\sigma_i > \sigma_j$.

However, this moving process can be invalid in the general case: $\min\{\sigma_i, \sigma_j\}$ might be greater than some thresholds in front of $X_i$ and $X_j$. To fix this issue, consider the following process:
\begin{itemize}
    \item Step 1: Increase $\tau_i$ and $\tau_j$ until reaching $\tau_k$, where $X_k$ is the distribution just in front of $X_i$ and $X_j$.
    \item Step 2: Swap $X_i$ and  $X_j$ with $X_k$.
    \item Repeat Step 1 and 2 until $\tau_i = \tau_j = \min\{\sigma_i, \sigma_j\}$.
\end{itemize}
Let $\Delta_{\pi, \pi'}$ be the difference between expected values of $\pi$ and $\pi'$. We monitor the change of $\Delta_{\pi, \pi'}$ during these operations. Step 1 can decrease $\Delta_{\pi, \pi'}$ by at most $12\epsilon <20\epsilon$. Step 2 can increase the absolute value of $\Delta_{\pi, \pi'}$. Since there can be at most $n$ Step 1 and 2,  if initially $\Delta_{\pi, \pi'} > 20n\epsilon$, this is sufficient to guarantee that $\Delta_{\pi, \pi'} > 0$ at the end of the process. Then, we are evident to claim $\sigma_i > \sigma_j$.

It remains to show that \Cref{CalcSmart} returns a \smart policy. Besides, this policy should also guarantee that when we are swapping $X_i$ and $X_j$ with $X_k$, the policy after doing a swap is still valid. Therefore, we introduce the following lemma:

\begin{Lemma}
\label{SwapSmart}
\Cref{CalcSmart} calculates a \smart policy. Besides, it has the following property: Let $\tau = \max\{\lcb'_i, \lcb'_j\}$ and $\tau' = \min\{\sigma_i, \sigma_j\}$. If $\tau' > \tau$, then for all $k \in [n] \setminus \{i, j\}$, if $\tau_k \in [\tau, \tau']$, there must be $(k, i) \notin S$ and $(k, j) \notin S$.
\end{Lemma}

\begin{proof}
The first two conditions in \Cref{smart} directly follows \Cref{CalcSmart}. For the objective condition, observe that no distribution in the set $T$ can be moved behind $X_i$ and $X_j$. Therefore, the policy calculated by \Cref{CalcSmart} minimizes $F_{\pi, i}(\tau)$, which means the third condition holds.

For the additional property, assume there exists $k$ satisfying $\tau_k = \ucb'_k$, $\tau_k < \min\{\sigma_i, \sigma_j\}$. Notice that if $(k, i) \in S'$, there must be $\ucb'_k \geq \ucb'_i \geq \min\{\sigma_i, \sigma_j\}$, which is in contrast to the condition $\tau_k = \ucb'_k < \min\{\sigma_i, \sigma_j\}$. Therefore, $(k, i) \notin S'$. Similarly, $(k, j) \notin S'$. Therefore, the additional property in \Cref{SwapSmart} holds.
\end{proof}

Finally, applying \Cref{SwapSmart} immediately proves \Cref{SwapOP}.
\end{proof}

\subsection{Converting our Policy to the Optimal Policy in Polynomial Steps} \label{step3}

In this section, we show that using $\poly(n)$ number of moves and swaps can convert any valid policy into the optimal policy. Since \Cref{MoveOP} and \Cref{SwapOP} already show that the difference of each move and swap is bounded by $O(\poly(n)\epsilon)$, combining these results, we can argue that the per-round loss of a valid policy is bounded by $O(\poly(n)\epsilon)$. Formally, we give the following lemma:

\begin{Lemma}
\label{RegretBound}
Given a valid constraint group $(I, S)$. For a valid policy of $(I, S)$, we use a ``move'' to represent the action that modifies a single threshold, and guarantees that the  policy after modifying the threshold is still valid. Besides, we use a ``swap'' to represent the action that swaps two consecutive distributions with the same threshold. This threshold should be equal to the maximum of the two lower confidence bounds, and the policy after swapping the distributions should still be valid.

For any valid policy of $(I, S)$, it can be converted into the optimal policy using $2n^2$ moves and $2n^2$ swaps.
\end{Lemma}

\begin{proof}

Let $\pi$ be the policy that $\tau_i = \lcb_i$ for all $i \in [n]$, and the distributions are sorted in a decreasing order of $\tau$. Since for every constraint $(i, j) \in S$, we have $\lcb_i \geq \lcb_j$, $\pi$ must be a valid policy.

We can prove \Cref{RegretBound} by showing the following statement: Starting from the policy $\pi$, we can move it to any valid policy $\pi'$ using $n^2$ moves and $n^2$ swaps:
\begin{itemize}
    \item Step 1: Let  $i = \arg \max_i \tau'_i$, where $\tau'_i$ is the threshold of $X_i$ in policy $\pi'$.
    \item Step 2: If $X_i$ is not the first distribution in $\pi$, move $\tau_i$ to $\tau_{\pi_{\pi^{-1}(i)-1}}$, then swap $X_i$ and $X_{\pi_{\pi^{-1}(i)-1}}$.
    \item Step 3: Do Step 2 until $X_i$ is moved to the first place. Then move $\tau_i$ to $\tau'_i$.
    \item Step 4: Ignore $X_i$ in both $\pi$ and $\pi'$, repeat Step 1, 2 and 3 until every distribution is settled.
\end{itemize}

Each distribution only involves in $n$ swaps and $n$ moves, so the total number of moves and swaps are both bounded by $n^2$. Then, we need to show the validity of every operation. For each move, we increase $\tau_i$ to let it be closer to $\tau'_i$. Since $\tau'_i \in [\lcb_i, \ucb_i]$, every move is valid. For each swap, the threshold in the front must reach its lower confidence bound. Besides, every swap happens only when there is no constraint between two distributions, so every swap is valid.

Finally, notice that every operation is bidirected. It means that starting from any valid policy $\pi'$, we can convert it to the policy $\pi$, and then convert it to the optimal policy using $2n^2$ moves and swaps, which finishes the proof.
\end{proof}

\subsection{Putting Everything Together}\label{combinepan}

In this section, we show how to combine \Cref{PBISA} and \Cref{PBOCA} to generate a new valid constraint group $(I', S')$, then proves that this leads to an $\OTild(\poly(n)\sqrt{T})$ regret algorithm. We first give the one-phase algorithm:

\begin{algorithm}[tbh]
\caption{Constraint Updating Algorithm for Pandora's Box}
\label{PanAlg}
\KwIn{$I = \{[\lcb_1, \ucb_1],...,[\lcb_n, \ucb_n]\}, S = \{(i, j)\}$, $\hat F_1(x),\ldots, \hat F_n(x)$, $m$}
//STEP 1: Calculate new confidence interval for each distribution \\
\For{$i \in [n]$}
{
{For $j \in [n]$, construct $\hat F_j(x)$ using $10^5\cdot \frac{n^2\log T}{\epsilon}$ new i.i.d. samples of $X_j$}\\
Run \Cref{PBISA} with new CDF estimates to get $\lcb'_i$ and $\ucb'_i$.
}
//Adjust the confidence intervals to meet constraints in $S$. \\
\For{$(i, j) \in S$}
{
Let $\lcb'_i  = \max\{\lcb'_i, \lcb'_j\}$ and $\ucb'_j = \min\{\ucb'_j, \ucb'_i\}$.
}
Let $I' = \{[\lcb'_i, \ucb'_i]\}$ and $S' = S$ \\
//Add new constraints for disjoint confidence intervals \\
\For{$(i, j) \notin S'$}
{
\lIf{$\lcb'_i > \ucb'_j$}
{
    Add $(i, j)$ into $S'$
}
}
~\\
//STEP 2: Calculate new constraints for each distribution pair \\
Let $Q = \{(i, j)|(i,j) \notin S' \land (j, i) \notin S'\}$ \\
\While{$Q \neq \emptyset$}
{
Choose $(i, j) \in Q$ and remove $(i, j)$ from $Q$ \\
Run \Cref{PBOCA} with input $(i, j)$ and update $I'$ and $S'$ \\
//New constraints may fail some previous tests. Should add them back \\
For every $k$ such that $\lcb'_k$ changes in \Cref{PBOCA}, if  $\exists k'$ such that $(k, k'), (k', k) \notin S'$, add $(k, k')$ into $Q$.
}
\KwOut{$(I',S')$}
\end{algorithm}

We can directly give the following lemma according to the three lemmas above:

\begin{Lemma}[Main Lemma]
\label{MainPandoraLemma}
Given $(I, S)$ and $\epsilon > 16T^{-\frac{1}{2}}$. Assume the pre-conditions in \Cref{MoveOP} hold, i.e.,
    \begin{itemize}
        \item $|g_j(\tau)| \leq T^{-\frac{1}{4}}$ for all $\tau \in [\lcb_j, \ucb_j]$.
        \item $(I, S)$ is valid.
        \item For any valid partial policy $\pi'$ of $(I, S)$, we fix the order and the other thresholds except $\tau_j$. Assume $\pi'$ is valid when both $\tau_j = \lcb'$ and $\tau_j = \lcb'$. Define $\delta_{\pi', \ucb', \lcb', j}(\tau) = (F_{\pi', j}(\lcb') - F_{\pi', j}(\ucb'))g_i(\tau)$. Then $|\delta_{\pi', \ucb', \lcb', j}(\tau)| \leq 6\epsilon$.
        \item CDF estimate $\hat F_j(x)$ is constructed via $10^5 \cdot \frac{n^2 \log T}{\epsilon}$ fresh i.i.d. samples of $X_j$.
    \end{itemize}

Then, \Cref{PanAlg} runs  $O(\frac{n\log T}{\epsilon^2})$ rounds, such that the policy in each round is valid for $(I, S)$ {(except Line 3)}, and output a new constraint group $(I', S')$, satisfying the following statements with probability $1 - T^{-10}$: 
\begin{itemize}
    \item $(I', S')$ is valid.
    \item For all $j \in [n]$, for any valid partial policy $\pi'$ of $(I', S')$, we fix the order and the other thresholds except $\tau_j$. Assume $\pi'$ is valid when both $\tau_j = \lcb'$ and $\tau_j = \lcb'$. Define $\delta_{\pi', \ucb', \lcb', i}(\tau) = (F_{\pi', j}(\lcb') - F_{\pi', j}(\ucb'))g_i(\tau)$. Then $|\delta_{\pi', \ucb', \lcb', i}(\tau)| \leq 3\epsilon$.
    \item For a valid policy of $(I', S')$, the per-round regret is no more than $126n^3\epsilon$.
\end{itemize}

\end{Lemma}

\begin{proof}

In this proof, we assume \Cref{MoveOP} and \Cref{SwapOP} holds. We use \Cref{MoveOP} for no more than $n$ times and \Cref{SwapOP} for no more than $n^2$ times. By the union bound\footnote{We assume $T > 10 n$, otherwise an $O(n)$ regret algorithm is trivial.}, our proof fails with probability at most $n \cdot T^{-11} + n^2 \cdot T^{-12} \leq T^{-10}$.

For the validity of $(I', S')$, the statement $\sigma_i \in [\lcb'_i, \ucb'_i]$ follows \Cref{MoveOP}, and the statement $\sigma_i > \sigma_j$ for all $(i, j) \in S'$ follows \Cref{SwapOP}. All other statements hold by definition. Therefore, $(I', S')$ is valid.

For the bound of $|\delta_{\pi', \ucb', \lcb', i}(\tau)|$, it's guaranteed directly by \Cref{MoveOP}. Notice that \Cref{MoveOP} even provides a stronger bound for the constraint group $(I'_i, S)$. Since all possible choices of $\pi', \lcb', \ucb'$ must be valid for $(I'_i, S)$ when it's valid for $(I', S')$, this doesn't hurt the statement.

For the per-round regret bound, \Cref{MoveOP} says that the difference of a move is bounded by $3\epsilon$. \Cref{SwapOP} says that the difference of a swap is bounded by $60n\epsilon$. Then, according to \Cref{RegretBound}, we can convert any valid policy to the optimal policy using $2n^2$ moves and swaps. Therefore, the per-round regret is bounded by $126n^3\epsilon$.

Next, we argue that \Cref{PanAlg} runs no more than $O(\frac{n\log T}{\epsilon^2})$ rounds. Note that \Cref{PBISA} is called $n$ times, and \Cref{PBISA} uses $O(\frac{\log T}{\epsilon^2})$ rounds in one call. So the number of rounds is $O(\frac{n\log T}{\epsilon^2})$. For \Cref{PBOCA}, we might test a distribution pair $(X_i, X_j)$ for multiple times. The reason is the following: When using \Cref{RegretBound}, we need to make sure that the value of the final $\max\{\lcb'_i, \lcb'_j\}$ is the one that we test. Therefore, if the value of $\lcb'_i$ changes, we need to re-test some distribution pairs $(i, j)$. We can argue that the total number of tests is bounded: When doing an extra test for $(i, j)$, at least one of $\lcb'_i$ or $\lcb'_j$ must change. This can happen only when a new constraint related to $i$ or $j$ is added into $S'$. There are only $2n$ constraints related to $i$ and $j$, so we can test $(i, j)$ for at most $4n$ times. Therefore, the total number of calls of \Cref{PBOCA} is no more than $4n^3$, and \Cref{PBOCA} uses $O(\frac{\log T}{n^2\epsilon^2})$ samples in one call, so the number of samples is bounded by $O(\frac{n\log T}{\epsilon^2})$. Combining the two results finishes the proof.  
\end{proof}

Now, we are ready to show the total regret bound.

\begin{Theorem}
\label{MainPandoraThm}
There exists an $O(n^{4.5}\sqrt{T}\log T)$ regret algorithm for Pandora's Box problem.
\end{Theorem}

\begin{proof}

We run \Cref{Doubling} and then use \Cref{DoublingBound} to bound the main part of the total regret. To run \Cref{Doubling}, we require the pre-conditions listed in \Cref{MainPandoraLemma} hold. We discuss them separately:

\begin{itemize}
    \item $|g_j(\tau)| \leq T^{-1/4}$: This is guaranteed by \Cref{PBInit}.
    \item $(I, S)$ is valid: For the first phase, the condition $\tau^*_i \in [\lcb_i, \ucb_i]$ is guaranteed by \Cref{PBInit}, and we don't have any initial order constraints between distributions (except those distributions with disjoint confidence intervals). Therefore, $(I, S)$ is valid for the first phase. Starting from the second phase, this is guaranteed by \Cref{MainPandoraLemma}.
    \item $|\delta_{\pi', \ucb', \lcb', i}(\tau)| \leq 6\epsilon$: For the first phase, this is true because $|\delta_{\pi', \ucb', \lcb', i}(\tau)| \leq |g(\tau)| \leq T^{-\frac{1}{4}}$, and initially we have $\epsilon = O(1)$. Starting from the second phase, this is from \Cref{MainPandoraLemma} regarding the previous phase. Notice that parameter $\epsilon$ in the new phase is exactly $\frac{\epsilon}{2}$ in the previous phase. Therefore, there is an extra $2$ factor in the condition.
    \item New CDF estimates: This is guaranteed by \Cref{PanAlg}.
\end{itemize}

\Cref{MainPandoraLemma} implies that after $O(\frac{n^7 \log T}{\epsilon^2})$ rounds, the one-round regret in the new constraint group is bounded by $\epsilon$. Applying \Cref{DoublingBound} with $\alpha = 7$, we have the $O(n^{3.5}\sqrt{T} \log T)$ regret bound.

Besides, there are some extra rounds not covered by \Cref{DoublingBound}, including the initialization and the CDF estimates construction ({{Line 3}} in \Cref{PanAlg}). For the initialization, \Cref{PBInit} runs $O(n \sqrt{T} \log T)$ samples, so the regret is $O(n\sqrt{T} \log T)$. For the CDF estimates construction, let $k$ be the number of phases in the doubling algorithm. Then, the total number of samples is 
\begin{align*}
    \sum_{i = 1}^k n\cdot O(\frac{n^2\log T}{\epsilon_i}) = O(\sqrt{T})
\end{align*}

Combining three parts of regret, the total regret is $O(n^{3.5}\sqrt{T} \log T)$.

Finally, recall that until now we are working on a scaled Pandora's Box problem: We scale down the values and the costs by a factor of $2n$. Therefore, for the original problem, the final regret bound is $O(n^{4.5}\sqrt{T}\log T)$.
\end{proof}

\subsection{Making the Algorithm Efficient}

Currently, the running time of the whole algorithm is exponential in $n$ as just \Cref{Lemma:Clever} introduces an algorithm with $O(n2^n)$ running time. If we want a polynomial time algorithm, we may need an approximation. The following lemma shows a new regret bound with approximation:

\begin{Lemma}
\label{alphaapprox}
Assume for every $i$, we can $\gamma$-approximate $\max_{\pi, \ucb, \lcb} \hat{F}_{\pi, i}(\ucb) - \hat{F}_{\pi, i}(\lcb)$, then there exists an $O(\max\{\gamma n^{4.5}, \gamma^2 n\} \sqrt{T}\log T)$ regret algorithm.
\end{Lemma}

\begin{proof}
In this proof, we first discuss the problem for the scaled Pandora's Box problem, and add the scaled $2n$ factor back at last.

We first see how the $\gamma$ approximation changes \Cref{MoveOP}. Recall that $q_i = \max_{\pi} F_{\pi, i}(\ucb) - F_{\pi, i}(\lcb)$. We further define $\tilde{q}_i = \max_{\pi} \hat F_{\pi, i}(\ucb) - \hat F_{\pi, i}(\lcb)$ and $\bar q_i = F_{\hat \pi, i}(\ucb) - F_{\hat \pi, i}(\lcb)$, where $\hat \pi$ is the chosen policy that $\gamma$-approximates $\max_{\pi, \ucb, \lcb} \hat{F}_{\pi, i}(\ucb) - \hat{F}_{\pi, i}(\lcb)$. According to \Cref{PBGenDisAcc}, we have $\tilde q_i \geq q_i - 2\sqrt{\epsilon}$ and $\bar q_i \geq \frac{\tilde q_i}{\gamma} - 2\sqrt{\epsilon}$. So $q_i \leq \gamma\bar q_i + (2\gamma+2)\sqrt{\epsilon}$. According to (\ref{EqPan2}), Statement \ref{MoveSta3} and \ref{MoveSta2} are both bounded by
\begin{align*}
    q_i \max_{v \in [\lcb'_i, \ucb'_i]} |g_i(v)| ~&\leq~ (\gamma\bar q_i + (2\gamma+2)\sqrt{\epsilon})\max_{v \in [\lcb'_i, \ucb'_i]} |g_i(v)|   \\
    & \leq~ \gamma \bar q_i \cdot \frac{2\epsilon}{\bar q_i} + (2\gamma+2)\sqrt{\epsilon} \cdot T^{-\frac{1}{4}} \leq 3\gamma \epsilon.
\end{align*}

For Statement~\ref{MoveSta3}, this changes the bound of $|\delta_{\hat \pi, \ucb', \lcb', i}(\tau)|$ to $O(\gamma\epsilon)$. In our proof, we use this bound when proving \Cref{AccBoundPBGen}: The bound of $|\delta_{\hat \pi, \ucb', \lcb', i}(\tau)|$ provides a bound for the variance of the $\Delta_i(\tau)$ function, and then we use Bernstein Inequality to show $|\hat \Delta_i(\tau) - \Delta_i(\tau)| \leq O(\epsilon)$. When the bound changes to $O(\gamma\epsilon)$, to get an $O(\epsilon)$ approximation of $\Delta_i(\tau)$, the number of samples for constructing CDF estimates should be multiplied by $\gamma^2$, leading to an $O(\gamma^2\sqrt{T} \log T)$ regret bound.

For Statement~\ref{MoveSta2}, notice that we need to use this moving difference to bound the swapping difference. The main idea of the original proof is: Assume we want to test $X_i$ and $X_j$. After $O(n)$ moves, we can adjust $\tau_i$ and $\tau_j$ to $\min\{\sigma_i, \sigma_j\}$, then bound the swapping difference by $O(n) \cdot O(\epsilon)$. Since there is an extra $\gamma$ factor in the new moving difference bound, the new swapping difference should be $O(\gamma n\epsilon)$. 

Next, \Cref{RegretBound} shows that we need $2n^2$ move operations and swap operations to convert a policy to the optimal one, so the new regret bound after $O(\frac{n\log T}{\epsilon^2})$ samples is $O(\gamma n^3 \epsilon)$. Then, the parameter $\alpha$ in \Cref{DoublingBound} changes to $\gamma^2 n^7$, so the total regret from the doubling algorithm is $O(\gamma n^{3.5} \sqrt{T} \log T)$.

Finally, after combining these two new regret bounds and adding the scaled $2n$ factor back to the regret bound, we get the $O(\max\{\gamma n^{4.5}, \gamma^2 n^2\}\sqrt{T}\log T)$ final regret bound.
\end{proof}

\Cref{alphaapprox} shows that: If we can get a $\poly(n)$ approximation for the \clever policy in polynomial time, we can still get an $O(\poly(n)\sqrt{T})$ regret algorithm. To achieve this goal, we introduce the following sub-routine:

\begin{Definition}[sub-routine]

Let Problem A be the following: Given $n$ and real numbers $a_1,..., a_n,$ $b_1, ..., b_n$, satisfying $0 \leq a_i \leq b_i \leq 1$ for all $i \in [n]$.  The objective of Problem A is to calculate
\begin{align*}
     \max_{B \in [n]} \prod_{i \in B} b_i - \prod_{i \in B} a_i.
\end{align*}
under a set of constraints $\{(i, j)\}$, where a constraint $(i, j)$ means that if we have $i \in B$, there must be $j \in B$.
\end{Definition}

\begin{Lemma}
\label{alphaLemma2}
If there exists an algorithm that calculates an $\gamma$-approximation for Problem A, then there exists an algorithm that $\gamma$-approximates $\max_{\pi} \hat{F}_{\pi, i}(\ucb) - \hat{F}_{\pi, i}(\lcb)$. If the running time of the algorithm for approximating Problem A is polynomial, then the algorithm for approximating $\hat{F}_{\pi, i}(\ucb) - \hat{F}_{\pi, i}(\lcb)$ is also polynomial.
\end{Lemma}

\begin{proof}

Consider calculating a \clever policy for $X_i$. Assume that we know the value of $\lcb$ and $\ucb$. Then, we only need to pick a subset $B \subseteq [n] \setminus \{i\}$ to maximize $\prod_{j \in B} b_j - \prod_{j \in N} a_j$, where $b_j = \hat{F}_j(\ucb)$, and $a_j = \hat{F}_j(\lcb)$.

However, not all subsets $B$ are valid. Firstly, for $j \in B$, there must be $\tau_j \geq \ucb$, which means $\ucb_j \geq \ucb$ is required. Similarly, we should also guarantee that $\lcb_j \leq \lcb$ for all $j \in [n] \setminus B$. Besides, if there is an order constraint $(j, k)$, then $k \in B$ implies $j \in B$, which can be represented as a constraint in Problem A. If all constraints are satisfied, policy $\tau_j = \ucb_j$ for $j \in B$ and $\tau_j = \lcb_j$ for $j \notin B \cup \{i\}$ is a feasible policy. Therefore, finding the optimal policy with fixed $\lcb$ and $\ucb$ is captured by Problem A. So, an $\gamma$-approximation algorithm for Problem A also $\gamma$-approximates $\hat{F}_{\pi, i}(\ucb) - \hat{F}_{\pi, i}(\lcb)$.

Notice that when maximizing $\hat{F}_{\pi, i}(\ucb) - \hat{F}_{\pi, i}(\lcb)$, we want to push the thresholds to the boundaries to give $\tau_i$ enough space. Therefore, the value of $\lcb$ and $\ucb$ must be equal to some $\lcb_j$ or $\ucb_j$, which means that there are only $O(n^2)$ candidates. Therefore, if the algorithm that $\gamma$- approximates Problem A runs in polynomial time, the running time of the algorithm for approximating $\hat{F}_{\pi, i}(\ucb) - \hat{F}_{\pi, i}(\lcb)$ is also polynomial.
\end{proof}

It remains to give a $\poly(n)$-approximation algorithm for Problem A, with $O(\poly(n))$ running time. The following theorem shows that this is possible:

\begin{Lemma}
\label{alphaLemma3}
Given an instance of Problem A. Let $B_j$ be the subset with the smallest size that contains $j$. Let $q_j = \prod_{i \in B_j} b_i - \prod_{i \in B_j} a_i$ Then, $\max_{j} q_j$ is an $n$-approximation of problem A.
\end{Lemma}

\begin{proof}
Construct a graph $G=(V, E)$, such that $V = [n]$, and $E$ is the set of all constraints, i.e., a constraint $(u, v)$ is represented as a directed edge $(u, v) \in E$. Then, $B_j$ is the set of all vertices which is reachable from $j$.

Notice that when $G$ contains a connected component, we can shrink the component into one single vertex, because picking any single vertex in the connected component means picking the whole component. Therefore, we only need to prove the theorem when $G$ is a directed acyclic graph (DAG).

Re-index the vertices in $G$, to make sure that for every edge $(u, v) \in E$, there must be $u > v$. Besides, make sure that $B^* = \{1,..., k\}$ is exactly the optimal set of Problem A. Then,
\begin{align*}
    \prod_{i \in [k]} b_i - \prod_{i \in [k]} a_i &= \sum_{j \in [k]} \left( \prod_{i = 1}^{j-1} b_i \cdot (b_j - a_j) \cdot \prod_{i = j+1}^k a_i\right) \\
    &\leq  \sum_{j \in [k]} \left( \prod_{i \in B_j \setminus \{j\}} b_i \cdot (b_j - a_j) \cdot 1\right) \\
    &\leq \sum_{j \in [k]} \left( \prod_{i \in B_j} b_i - \prod_{i\in T_j} a_i\right)     \qquad \leq \qquad  \sum_{j \in [n]} q_j, 
\end{align*}
where the second-last inequality uses $a_i \leq b_i$.
Therefore, $\max_{j} q_j$ is an $n$-approximation of $\prod_{i \in B^*} b_i - \prod_{i \in B^*} a_i$.
\end{proof}

Finally, combining \Cref{alphaapprox}, \Cref{alphaLemma2}, and \Cref{alphaLemma3} gives the following main theorem:

\mainThmPB*


\section{Lower Bounds} \label{sec:lowerBounds}
In this section we prove lower bounds for Online Learning Prophet Inequality and Online Learning Pandora's Box. Our lower bounds will hold even against full-feedback.

\subsection{$\Omega(\sqrt{T})$  Lower Bound  for Stochastic Input}
We show an $\Omega(\sqrt{T})$ regret lower bound for Bandit Prophet Inequality and an $\Omega(\sqrt{nT})$ lower bound for Pandora's Box problem, which implies that the $\sqrt{T}$ factor in our regret bounds is tight. We first give the lower bound for Prophet Inequality.

\begin{Theorem}
For Bandit Prophet Inequality there exists an instance  with $n=2$ such that all online algorithms  incur $\Omega(\sqrt{T})$ regret.
\end{Theorem}

\begin{proof}
Let $\D_1$ be a distribution that always gives $\frac{1}{2}$. Let $\D_2$ be a Bernoulli distribution. The probability of $X_2 = 1$ might be $\frac{1}{2} + \frac{1}{\sqrt{T}}$ or $\frac{1}{2} - \frac{1}{\sqrt{T}}$. Both settings appear w.p. $\frac{1}{2}$. The online algorithm doesn't know which is the real setting. If it chooses not to open $X_2$, it will lose $\sqrt{T}$ w.p. $\frac{1}{2}$. Otherwise,
because of the variance, the algorithm needs $\Omega(T)$ samples from $X_2$ to learn the real setting, and loses $\frac{1}{2}\cdot \sqrt{T}$ for each round it runs. In both cases, the online algorithm should lose $\Omega(\sqrt{T})$, which finishes the proof.
\end{proof}

For the Pandora's Box problem, \cite{GHTZ-COLT21} already shows a lower bound for the sample complexity of Pandora's Box problem, which directly implies a lower bound for the online learning setting.

\begin{Theorem}[\cite{GHTZ-COLT21}]
\label{PBLB}
For any instance of Pandora’s problem in which the rewards are bounded in $[0, 1]$, running $\Omega(\frac{n}{\epsilon^2})$ samples is necessary to get an $\epsilon$-additive algorithm.
\end{Theorem}

\begin{Corollary}
For Pandora's Box problem, 
\newadded{all online algorithms  incur $\Omega(\sqrt{nT})$ regret.}
\end{Corollary}
\begin{proof}
Assume there exists an online algorithm that achieves $o(\sqrt{nT})$ regret. This implies that after $T$ rounds, we can achieve $o(\frac{n}{\sqrt{T}})$ per-round regret, which is in contradiction with \Cref{PBLB}.
\end{proof}

We remark that   \cite{GHTZ-COLT21}   claims that  $\Omega(\frac{n}{\epsilon^2})$ samples are necessary to get an $\epsilon$-additive algorithm for Prophet Inequality but without giving a proof. However, this claim seems incorrect since in an ongoing work we show  an $\widetilde{O}(\sqrt{T})$ regret algorithm for Prophet Inequality with full-feedback. 


\subsection{$\Omega(T)$ Lower Bound for Adversarial Input} \label{sec:lowerBoundAdvers}

In this paper, we study  Bandit Prophet Inequality and Bandit Pandora's Box problems under the stochastic assumption that input is drawn from  unknown-but-fixed distributions. A natural extension would be: can we obtain $o(T)$ regret  for adversarial inputs where the the input distribution may change in each time step? The following theorems shows that sub-linear regret is impossible even for oblivious adversarial inputs with $n=2$ under full-feedback.

\begin{Theorem}
\label{AdvPI}
For Bandit Prophet Inequality  with oblivious adversarial inputs, there exists an instance with $n=2$ such that the optimal fixed-threshold strategy has total value $\frac{3}{4}T$ but no online algorithm (even under full-feedback) can obtain total value more than $\frac{1}{2}T$.
\end{Theorem}

\begin{proof}

We first introduce a notation used in this proof. Let $s$ be a 01-string. Define $Bin(s)$ to be the binary decimal corresponding to $s$. For example, $Bin(1) = (0.1)_2 = \frac{1}{2}$, $Bin(0011) = (0.0011)_2 = \frac{3}{16}$.

Now, we introduce the main idea of the counter example: At the beginning, the adversary will choose a $T$-bits code $s = s_1s_2...s_T$ uniformly at random (i.e., $s_i$ is set to be $0$ or $1$  w.p. $\frac{1}{2}$ independently). The value of $X_1$ is $\frac{1}{2}$ plus a small bias that contains the information of the code. The value of $X_2$ is either $1$ or $0$, which is decided by the code. Formally, in the $i$-th round:

\begin{itemize}
    \item $X_1 = \frac{1}{2} + \epsilon\cdot v_i$, where $\epsilon$ is an arbitrarily small constant that doesn't effect the reward, and $v_i$ is a value between $Bin(s_1s_2...s_{i-1}+0+1^{T-i})$ and $Bin(s_1s_2...s_{i-1}+1+0^{T-i})$. The notation $0^k$ represents a length-$k$ string with all $0$s, and $1^k$ represents a length-$k$ string with all $1$s.
    \item $X_2 = 1$ if $s_i = 0$, otherwise $X_2 = 1$.
\end{itemize}

For an online algorithm, it only knows that the next $s_i$ can be $0$ or $1$ w.p. $\frac{1}{2}$. Therefore, no matter it switches to the next box or not, it can only get $\frac{1}{2}$ in expectation. So the maximum total reward it can achieve is $\frac{1}{2}T$. 

However, if we know the code, playing $\tau = \frac{1}{2} + \epsilon\cdot Bin(s)$ gets $\frac{3}{4}T$: $X_2 = 1$ when $X_1 < \tau$, while $X_2 = 0$ when $X_1 \geq \tau$. Therefore, playing $\tau$ allows us to pick every $1$, but stays in $X_1 = \frac{1}{2}$ when $X_2=0$. Since we generate the code uniformly at random, $X_2$ is $1$ w.p. $\frac{1}{2}$. Therefore, the expected reward is $T\cdot (\frac{1}{2} \cdot 1 + \frac{1}{2} \cdot \frac{1}{2}) = \frac{3}{4}T$.
\end{proof}

Next, we use a similar proof idea to prove lower bound for Pandora's Box. This resolves an open question of \cite{Gergatsouli-22,GT-ICML22} on whether sublinear regrets are possible for Online Learning of Pandora's Box with adversarial inputs.

\begin{Theorem}
\label{AdvPB}
For Bandit Pandora's Box   with oblivious adversarial inputs, there exists an instance with $n=2$ such that the optimal fixed-threshold strategy has total utility $\frac{1}{4}T$ but no online algorithm  (even under full-feedback)  can obtain total utility more than $0$.
\end{Theorem}

\begin{proof}
At the beginning, the adversary will choose a $T$-bits code $s = s_1s_2...s_T$ uniformly at random ($s_i$ is set to be $0$ or $1$  w.p. $\frac{1}{2}$ independently). The cost $c_1$ is $0$, and the value of $X_1$ is $0$ plus a small bias that contains the information of the code. The cost $c_2$ is $\frac{1}{2}$, and the value of $X_2$ is either $1$ or $0$, which is decided by the code. 
Formally, in the $i$-th round:
\begin{itemize}
    \item $X_1 = 0 + \epsilon\cdot v_i$, where $\epsilon$ is an arbitrarily small constant that doesn't effect the reward, and $v_i$ is a value between $Bin(s_1s_2...s_{i-1}+0+1^{T-i})$ and $Bin(s_1s_2...s_{i-1}+1+0^{T-i})$. The notation $0^k$ represents a length-$k$ string with all $0$s, and $1^k$ represents a length-$k$ string with all $1$s.
    \item $X_2 = 1$ if $s_i = 0$, otherwise $X_2 = 1$.
\end{itemize}

The cost of $X_1$ is $0$, so we can always first open $X_1$. Then, for an online algorithm, it doesn't know whether $X_2$ is $1$ or $0$. No matter it opens $X_2$ or not, the expected reward will only be $0$.

However, when we know the code, playing $\tau = \epsilon\cdot Bin(s)$ gets $\frac{1}{4}T$, because it will open $X_2$ whenever $X_2 = 1$, and skip it when $X_2$ is $0$. Since we generate the code uniformly at random, $X_2$ is $1$ w.p. $\frac{1}{2}$. Therefore, the expected reward is $T\cdot (\frac{1}{2} \cdot (1 -\frac{1}{2})) = \frac{1}{4}T$.
\end{proof}

\bibliographystyle{alpha}
\bibliography{ref.bib,bib.bib}

\newcommand{\etalchar}[1]{$^{#1}$}
\begin{thebibliography}{FTW{\etalchar{+}}21}

\bibitem[ACBF02]{ACF-ML02}
Peter Auer, Nicolo Cesa-Bianchi, and Paul Fischer.
\newblock Finite-time analysis of the multiarmed bandit problem.
\newblock {\em Machine learning}, 47(2):235--256, 2002.

\bibitem[ACG{\etalchar{+}}22]{atsidakou2022contextual}
Alexia Atsidakou, Constantine Caramanis, Evangelia Gergatsouli, Orestis Papadigenopoulos, and Christos Tzamos.
\newblock Contextual pandora's box.
\newblock {\em arXiv preprint arXiv:2205.13114}, 2022.

\bibitem[AHK12]{AHK-TOC12}
Sanjeev Arora, Elad Hazan, and Satyen Kale.
\newblock The multiplicative weights update method: a meta-algorithm and applications.
\newblock {\em Theory of computing}, 8(1):121--164, 2012.

\bibitem[AKW14]{AKW-SODA14}
Pablo~Daniel Azar, Robert Kleinberg, and S.~Matthew Weinberg.
\newblock Prophet inequalities with limited information.
\newblock In {\em Proceedings of the Twenty-Fifth Annual {ACM-SIAM} Symposium on Discrete Algorithms, {SODA} 2014}, pages 1358--1377, 2014.

\bibitem[ANSS19]{ANSS-EC19}
Nima Anari, Rad Niazadeh, Amin Saberi, and Ali Shameli.
\newblock Nearly optimal pricing algorithms for production constrained and laminar bayesian selection.
\newblock In {\em Proceedings of the 2019 {ACM} Conference on Economics and Computation, {EC}}, pages 91--92, 2019.

\bibitem[BC12]{BubeckC-Book12}
S{\'{e}}bastien Bubeck and Nicol{\`{o}} Cesa{-}Bianchi.
\newblock Regret analysis of stochastic and nonstochastic multi-armed bandit problems.
\newblock {\em Foundations and Trends in Machine Learning}, 5(1):1--122, 2012.

\bibitem[BDL22]{BDL-EC22}
Mark Braverman, Mahsa Derakhshan, and Antonio~Molina Lovett.
\newblock Max-weight online stochastic matching: Improved approximations against the online benchmark.
\newblock In {\em {ACM} Conference on Economics and Computation, {EC}}, 2022.

\bibitem[Bel57]{Bellman-Book57}
Richard Bellman.
\newblock {\em Dynamic programming}.
\newblock Princeton University Press, 1957.

\bibitem[CBL06]{CL-Book06}
Nicolo Cesa-Bianchi and G{\'a}bor Lugosi.
\newblock {\em Prediction, learning, and games}.
\newblock Cambridge university press, 2006.

\bibitem[CDF{\etalchar{+}}22]{CDFFLLP-SODA22}
Constantine Caramanis, Paul D{\"{u}}tting, Matthew Faw, Federico Fusco, Philip Lazos, Stefano Leonardi, Orestis Papadigenopoulos, Emmanouil Pountourakis, and Rebecca Reiffenh{\"{a}}user.
\newblock Single-sample prophet inequalities via greedy-ordered selection.
\newblock In {\em Proceedings of the 2022 {ACM-SIAM} Symposium on Discrete Algorithms, {SODA}}, pages 1298--1325, 2022.

\bibitem[CDFS19]{CDFS-EC19}
Jos{\'{e}}~R. Correa, Paul D{\"{u}}tting, Felix~A. Fischer, and Kevin Schewior.
\newblock Prophet inequalities for {I.I.D.} random variables from an unknown distribution.
\newblock In {\em Proceedings of the 2019 {ACM} Conference on Economics and Computation, {EC}}, pages 3--17, 2019.

\bibitem[CHMS10]{CHMS-STOC10}
Shuchi Chawla, Jason~D. Hartline, David~L. Malec, and Balasubramanian Sivan.
\newblock Multi-parameter mechanism design and sequential posted pricing.
\newblock In {\em Proceedings of the 42nd {ACM} Symposium on Theory of Computing, {STOC}}, pages 311--320, 2010.

\bibitem[EFGT20]{EFGT-EC20}
Tomer Ezra, Michal Feldman, Nick Gravin, and Zhihao~Gavin Tang.
\newblock Online stochastic max-weight matching: Prophet inequality for vertex and edge arrival models.
\newblock In {\em The 21st {ACM} Conference on Economics and Computation, {EC}}, pages 769--787, 2020.

\bibitem[EHLM19]{EHLM-AAAI19}
Hossein Esfandiari, Mohammad~Taghi Hajiaghayi, Brendan Lucier, and Michael Mitzenmacher.
\newblock Online pandora's boxes and bandits.
\newblock In {\em The Thirty-Third {AAAI} Conference on Artificial Intelligence, {AAAI}}, pages 1885--1892, 2019.

\bibitem[FGL15]{FGL-SODA15}
Michal Feldman, Nick Gravin, and Brendan Lucier.
\newblock Combinatorial auctions via posted prices.
\newblock In {\em Proceedings of the Twenty-Sixth Annual ACM-SIAM Symposium on Discrete Algorithms, {SODA}}, pages 123--135, 2015.

\bibitem[FL20]{FL-NeurIPS20}
Hu~Fu and Tao Lin.
\newblock Learning utilities and equilibria in non-truthful auctions.
\newblock In {\em Advances in Neural Information Processing Systems 33: Annual Conference on Neural Information Processing Systems 2020, NeurIPS}, 2020.

\bibitem[FSZ16]{FSZ-SODA16}
Moran Feldman, Ola Svensson, and Rico Zenklusen.
\newblock Online contention resolution schemes.
\newblock In {\em Proceedings of the Twenty-Seventh Annual {ACM-SIAM} Symposium on Discrete Algorithms, {SODA} 2016, Arlington, VA, USA, January 10-12, 2016}, pages 1014--1033, 2016.

\bibitem[FTW{\etalchar{+}}21]{FTWWZ-ICALP21}
Hu~Fu, Zhihao~Gavin Tang, Hongxun Wu, Jinzhao Wu, and Qianfan Zhang.
\newblock Random order vertex arrival contention resolution schemes for matching, with applications.
\newblock In {\em 48th International Colloquium on Automata, Languages, and Programming, {ICALP}}, pages 68:1--68:20, 2021.

\bibitem[Ger22]{Gergatsouli-22}
Evangelia Gergatsouli.
\newblock Personal communication.
\newblock 2022.

\bibitem[GHTZ21]{GHTZ-COLT21}
Chenghao Guo, Zhiyi Huang, Zhihao~Gavin Tang, and Xinzhi Zhang.
\newblock Generalizing complex hypotheses on product distributions: Auctions, prophet inequalities, and pandora's problem.
\newblock In {\em Conference on Learning Theory, {COLT}}, pages 2248--2288, 2021.

\bibitem[GJSS19]{GJSS-IPCO19}
Anupam Gupta, Haotian Jiang, Ziv Scully, and Sahil Singla.
\newblock The markovian price of information.
\newblock In {\em Proceedings of Integer Programming and Combinatorial Optimization, {IPCO}}, volume 11480, pages 233--246, 2019.

\bibitem[GKS19]{GKS-SODA19}
Buddhima Gamlath, Sagar Kale, and Ola Svensson.
\newblock Beating greedy for stochastic bipartite matching.
\newblock In {\em Proceedings of the Thirtieth Annual {ACM-SIAM} Symposium on Discrete Algorithms, {SODA}}, pages 2841--2854, 2019.

\bibitem[GT22]{GT-ICML22}
Evangelia Gergatsouli and Christos Tzamos.
\newblock Online learning for min sum set cover and pandora's box.
\newblock In {\em International Conference on Machine Learning, {ICML}}, pages 7382--7403, 2022.

\bibitem[Har22]{Hartline-Book22}
Jason~D Hartline.
\newblock {\em Mechanism design and approximation}.
\newblock Book draft., 2022.

\bibitem[Haz16]{Hazan-Book16}
Elad Hazan.
\newblock Introduction to online convex optimization.
\newblock {\em Foundations and Trends{\textregistered} in Optimization}, 2(3-4):157--325, 2016.

\bibitem[HKS07]{HKS-AAAI07}
Mohammad~Taghi Hajiaghayi, Robert~D. Kleinberg, and Tuomas Sandholm.
\newblock Automated online mechanism design and prophet inequalities.
\newblock In {\em Proceedings of the Twenty-Second {AAAI} Conference on Artificial Intelligence}, pages 58--65, 2007.

\bibitem[KS77]{Krengel-Journal77}
Ulrich Krengel and Louis Sucheston.
\newblock Semiamarts and finite values.
\newblock {\em Bull. Am. Math. Soc}, 1977.

\bibitem[KS78]{Krengel-Journal78}
Ulrich Krengel and Louis Sucheston.
\newblock On semiamarts, amarts, and processes with finite value.
\newblock {\em Advances in Prob}, 4:197--266, 1978.

\bibitem[KW12]{KW-STOC12}
Robert Kleinberg and S.~Matthew Weinberg.
\newblock Matroid prophet inequalities.
\newblock In {\em Proceedings of the 44th Symposium on Theory of Computing Conference, {STOC}}, pages 123--136, 2012.

\bibitem[KWW16]{KWW-EC16}
Robert~D. Kleinberg, Bo~Waggoner, and E.~Glen Weyl.
\newblock Descending price optimally coordinates search.
\newblock In {\em Proceedings of the {ACM} Conference on Economics and Computation, {EC}}, pages 23--24, 2016.

\bibitem[LLP{\etalchar{+}}21]{LLPSS-EC21}
Allen Liu, Renato~Paes Leme, Martin P{\'{a}}l, Jon Schneider, and Balasubramanian Sivan.
\newblock Variable decomposition for prophet inequalities and optimal ordering.
\newblock In {\em The 22nd {ACM} Conference on Economics and Computation, {EC}}, page 692, 2021.

\bibitem[LS20]{LS-Book20}
Tor Lattimore and Csaba Szepesv{\'a}ri.
\newblock {\em Bandit algorithms}.
\newblock Cambridge University Press, 2020.

\bibitem[LSTW23]{LSTW-SODA23}
Renato~Paes Leme, Balasubramanian Sivan, Yifeng Teng, and Pratik Worah.
\newblock Pricing query complexity of revenue maximization.
\newblock In {\em Proceedings of the 2023 {ACM-SIAM} Symposium on Discrete Algorithms, {SODA}}, pages 399--415, 2023.

\bibitem[Luc17]{Lucier-Survey17}
Brendan Lucier.
\newblock An economic view of prophet inequalities.
\newblock {\em SIGecom Exch.}, 16(1):24--47, 2017.

\bibitem[PPSW21]{PPSW-EC21}
Christos~H. Papadimitriou, Tristan Pollner, Amin Saberi, and David Wajc.
\newblock Online stochastic max-weight bipartite matching: Beyond prophet inequalities.
\newblock In {\em The 22nd {ACM} Conference on Economics and Computation, {EC}}, pages 763--764, 2021.

\bibitem[Rou16]{RoughgardenTwenty-Book16}
Tim Roughgarden.
\newblock {\em Twenty lectures on algorithmic game theory}.
\newblock Cambridge University Press, 2016.

\bibitem[RS17]{RS-SODA17}
Aviad Rubinstein and Sahil Singla.
\newblock Combinatorial prophet inequalities.
\newblock In {\em Proceedings of the Twenty-Eighth Annual ACM-SIAM Symposium on Discrete Algorithms, {SODA}}, pages 1671--1687, 2017.

\bibitem[Rub16]{Rubinstein-STOC16}
Aviad Rubinstein.
\newblock Beyond matroids: secretary problem and prophet inequality with general constraints.
\newblock In {\em Proceedings of the 48th Annual {ACM} {SIGACT} Symposium on Theory of Computing, {STOC}}, pages 324--332, 2016.

\bibitem[RWW20]{RWW-ITCS20}
Aviad Rubinstein, Jack~Z. Wang, and S.~Matthew Weinberg.
\newblock Optimal single-choice prophet inequalities from samples.
\newblock In {\em 11th Innovations in Theoretical Computer Science Conference, {ITCS}}, pages 60:1--60:10, 2020.

\bibitem[SC84]{Samuel-Annals84}
Ester Samuel-Cahn.
\newblock Comparison of threshold stop rules and maximum for independent nonnegative random variables.
\newblock {\em the Annals of Probability}, pages 1213--1216, 1984.

\bibitem[Sin18a]{Singla-Thesis18}
Sahil Singla.
\newblock Combinatorial optimization under uncertainty: Probing and stopping-time algorithms.
\newblock {\em Unpublished doctoral dissertation, Carnegie Mellon University}, 2018.

\bibitem[Sin18b]{Singla-SODA18}
Sahil Singla.
\newblock The price of information in combinatorial optimization.
\newblock In {\em Proceedings of the Twenty-Ninth Annual ACM-SIAM Symposium on Discrete Algorithms}. SIAM, 2018.

\bibitem[Sli19]{Slivkins-Book19}
Aleksandrs Slivkins.
\newblock Introduction to multi-armed bandits.
\newblock {\em Foundations and Trends in Machine Learning}, 12(1-2):1--286, 2019.

\bibitem[SS21]{SS-EC21}
Danny Segev and Sahil Singla.
\newblock Efficient approximation schemes for stochastic probing and prophet problems.
\newblock In {\em The 22nd {ACM} Conference on Economics and Computation, {EC}}, pages 793--794, 2021.

\bibitem[Wei79]{Weitzman-Econ79}
Martin~L. Weitzman.
\newblock Optimal search for the best alternative.
\newblock {\em Econometrica: Journal of the Econometric Society}, pages 641--654, 1979.

\end{thebibliography}

\appendix
\section{Basic Probabilistic Inequalities}
 

\begin{Theorem}[Hoeffding's Inequality]
\label{Hoeffding}
Let $X_1,\ldots,X_N$ be independent random variables such that $a_i \leq X_i \leq b_i$. Let $S_N = \sum_{i \in [N]} X_i$. Then for all $t > 0$, we have
$  \Pr{|S_N - \Ex{S_n}| \geq t} \leq 2\exp\big(-\frac{2t^2}{\sum_{i \in [n]} (b_i - a_i)^2}\big).$
This implies, that if $X_i$ are i.i.d. samples of random variable $X$, and $a = a_i, b = b_i$ for all $i \in [N]$, let  $\hat X := \frac{1}{N} \sum_{i \in [N]} X_i$, then for every $\varepsilon > 0$, 
\[
    \Pr{|\hat{X} - \Ex{X}| \geq \varepsilon} \leq 2 \exp \left(-\frac{2N\varepsilon^2}{(b-a)^2}\right).
\]
\end{Theorem}

\begin{Theorem}[Bernstein Inequality]
\label{Bernstein}
Given mean zero random variables $\{X_i\}_{i=1}^N$ with $\mathbb P(|X_i| \leq c) = 1$ and $\mathsf{Var} X_i \leq \sigma_i^2$. If $\bar{X}_N$ denotes their average and $\sigma^2 = \frac{1}{N}\sum_{i=1}^n \sigma_i^2$, then
\begin{align*}
    \mathbb P(|\bar{X}_N| \geq \varepsilon) ~\leq~ 2\exp\big(-\frac{N\varepsilon^2}{2\sigma^2 + 2c\varepsilon/3}\big).
\end{align*}
\end{Theorem}

\begin{Theorem}[DKW Inequality]
\label{DKW}
Given a natural number $N$, let $X_1,\ldots,X_n$ be i.i.d. samples with cumulative distribution function $F(\cdot)$. Let $\hat{F}(\cdot)$ be the associated empirical distribution function $ \hat F(x) ~:=~ \frac{1}{N} \sum_{i \in [N]} \one_{X_i \leq x}.$
Then, for every $\varepsilon > 0$, we have
\begin{align*}
    \Pr{\sup_{x}|\hat F(x) - F(x)|> \varepsilon} \leq 2\exp(-2N\varepsilon^2).
\end{align*}
\end{Theorem}


\section{Missing Proofs from Section 2}
\label{sec:appendixSec2}

\subsection{Proof of \Cref{DoublingBound}}

\begin{proof}

There are two different sources of regret. We bound them separately.

\medskip
\noindent \emph{Loss 1 from the while loop:} The main idea of the proof is to use the regret bound from the previous phase to bound the total regret in the next phase. Specifically, assume $\epsilon_0 = O(1)$ be the maximum possible one-round regret, and assume there are $k$ phases in the while loop. Then the total regret can be bounded by
\begin{align}
    \sum_{i=1}^k O\left(\frac{n^\alpha \log T}{\epsilon_i^2}\right) \cdot \epsilon_{i-1} ~=~ \sum_{i=1}^k O\left(\frac{n^\alpha \log T}{\epsilon_i}\right) ~=~ O(n^{\alpha / 2}\sqrt{T}).
\end{align}
Therefore, the total regret from the while loop is bounded by $O(n^{\alpha / 2}\sqrt{T})$.

\medskip
\noindent \emph{Loss 2 after the while loop:} After the while loop, the one-round regret is bounded by $\epsilon_k =\frac{n^{\alpha / 2} \log T}{\sqrt{T}}$, so the total regret can be bounded by $O(\frac{n^{\alpha / 2} \log T}{\sqrt{T}}) \cdot T = O(n^{\alpha / 2}\sqrt{T} \log T)$.

Finally, combining the two sources of regret proves the theorem.

Besides, we should also verify that \Cref{Doubling}  succeeds with probability $1 - T^{-9}$, and it runs no more than $O(T)$ rounds. For the success probability, \Cref{Doubling} runs $k = O(\log T) < T$ rounds, and the subroutine $\alg$ succeeds with probability $1 - T^{-10}$. By the union bound, \Cref{Doubling} succeeds with probability $1 - T^{-9}$. As for the number of rounds, in the while loop, \Cref{Doubling} runs
\begin{align*}
    \sum_{i=1}^k \frac{n^\alpha \log T}{\epsilon^2_i} ~\leq~ \frac{4n^\alpha \log T}{\epsilon^2_k} ~=~ O(T)
\end{align*}
number of rounds.
Therefore, \Cref{Doubling} is a valid algorithm with respect to time horizon $T$.
\end{proof}

\subsection{Missing Details of Pandora's Box Algorithm for $n = 2$} \label{sec:missingPandoraTwo}

\subsubsection{Proof of \Cref{mainthmpan}} To prove \Cref{mainthmpan}, we need two claims. The first claim says that when we have a good guess $\tau$ with a small $\delta(\tau)$, the loss of playing $\tau$ is bounded:

\PBBoundTwoBoxes*

\begin{proof}
We first upper-bound $R(\tau^*) - R(\tau)$. The two settings are different only when $X_1$ is between $\tau$ and $\tau^*$: Playing $\tau$ will loss an extra $|g(X_1)|$. Since $g(x)$ is monotone, we can bound $|g(X_1)|$ by $|g(\tau)|$. Therefore, the extra loss of playing $\tau$ is no more than $|F_1(\tau^*) - F_1(\tau)||g(\tau)|$.

On the other hand, 
\begin{align*}
    |\delta(\tau)| = (F_1(\ucb) - F_1(\lcb))\left| \int_\tau^{\tau^*} g'(x) dx\right| 
    =  (F_1(\ucb) - F_1(\lcb)) \cdot |g(\tau)|.
\end{align*}
When $\tau, \tau^* \in [\lcb, \ucb]$, $F_1(\ucb) - F_1(\lcb) \geq |F_1(\tau^*) - F_1(\tau)|$. Therefore, $|\delta(\tau)| \geq R(\tau^*) - R(\tau)$.
\end{proof}

The second claim shows that we can get a good estimate for function $\delta(\tau)$:

\AccBoundPan*

\begin{proof}

Recall that $\delta(\tau) = \Delta(\tau) - \left(R(\ucb) - R(\lcb)\right)$. We will give the bound for $|\Delta(\tau^*) - \Delta(\tau)|$, $|R(\ucb) - \hat R(\ucb)|$ and $|R(\lcb) - \hat R(\lcb)|$ separately.

For $|\Delta(\tau^*) - \Delta(\tau)|$, we first bound the magnitude of $\Delta(\tau)$:
\begin{align*}
    \Delta(\tau) = \int_{\tau}^\ucb (F_1(\ucb) - F_1(x))(F_2(x) - 1)dx - 
    \int_{\lcb}^\tau (F_1(x) - F_1(\lcb))(F_2(x) - 1)dx,
\end{align*}
which implies
\begin{align}
    |\Delta(\tau)| & \leq  (F_1(\ucb) - F_1(\tau))\int_{\tau}^{\ucb}(1 - F_2(x))dx + 
    (F_1(\tau) - F_1(\lcb))\int_{\lcb}^\tau (1 - F_2(x))dx\\
    & = (F_1(\ucb) - F_1(\tau))(g(\tau) - g(\ucb)) + (F_1(\tau) - F_1(\lcb))(g(\lcb) - g(\tau))\\
    & \leq (F_1(\ucb) - F_1(\lcb))(g(\lcb) - g(\ucb))\\
    & \leq  |\delta(\ucb)| + |\delta(\lcb)| ~\leq ~ 32 \epsilon,\label{eq:pandoratwoderivation}\end{align}
where the last equality follows from the bound 
 $|\delta(\tau)| \leq 16 \epsilon$ for all $\tau \in [\lcb, \ucb]$ in \Cref{mainthmpan}.
 
 Now notice that the estimate $\hat \Delta(\tau)$ we have based on our initial estimates $\hat F_1$ and $\hat F_2$ is unbiased i.e. 
$\Ex{ \hat \Delta(\tau) }= \Delta(\tau) \leq 32\epsilon$. This simply follows from exchanging interval integration and expectation combined with the independence of $X_1$ and $X_2$:
\begin{align}
     \Ex {\int_{\tau}^u (\hat F_1(u) - \hat F_1(x))(\hat F_2(x) - 1)dx} & = 
    \int_{\tau}^u (\Ex{ \hat F_1(u)} - \Ex{ \hat F_1(x)})(\Ex{ \hat F_2(x)} - 1)dx \notag\\
    &= \int_{\tau}^u ( F_1(u) - F_1(x))(F_2(x) - 1)dx.\label{eq:expzero}
\end{align}
Now let us define $\hat \Delta(\tau)$ per sample $i$ for each initial sample. We run $N = C\cdot \frac{\log T}{\epsilon}$ samples for $C = 1000$. Then for $i \in [N]$, we define
\begin{align*}
    \hat \Delta^{(k)}(\tau)= \int_{\tau}^u (\hat F^{(k)}_1(u) - \hat F^{(k)}_1(x))(\hat F^{(k)}_2(x) - 1)dx - 
    \int_{\ell}^\tau (\hat F^{(k)}_1(x) - \hat F^{(k)}_1(\ell))(\hat F^{(k)}_2(x) - 1)dx,
\end{align*}
where $\hat F^{(k)}_1(.)$ and $\hat F^{(k)}_2(.)$ are simple threshold functions at the $i$th initial sample, which are estimates for the densities $F_1$ and $F_2$ respectively. Note that
\[
    \hat \Delta(\tau) = \frac{1}{N} \sum_{k\in [N]} \hat \Delta^{(k)}(\tau).
\]
Now again similar to \eqref{eq:expzero} we have
\begin{align}
    \Ex{ \Delta^{(k)}(\tau)} ~=~ \Ex{ \Delta(\tau)} ~\leq~ 32 \epsilon.\label{eq:expbound}
\end{align}Moreover, note that the random variable $\hat \Delta^{(i)}(\tau)$ is bounded by one since
\begin{align}
    |\hat \Delta^{(k)}(\tau)| & \leq \int_{\tau}^u \Big|(\hat F^{(k)}_1(u) - \hat F^{(k)}_1(x))(\hat F^{(k)}_2(x) - 1)\Big|dx - 
    \int_{\ell}^\tau \Big|(\hat F^{(k)}_1(x) + \hat F^{(k)}_1(\ell))(\hat F^{(k)}_2(x) - 1)\Big|dx \notag \\
    & \leq \int_{\tau}^u 1dx + 
    \int_{\ell}^\tau 1dx ~=~ \ucb - \lcb ~\leq~ 1.\label{eq:avali}
\end{align}
Combining Equations~\eqref{eq:expbound} and~\eqref{eq:avali}, we have the variance bound:
\begin{align}
    \textsf{Var}[ \hat \Delta(\tau) ] ~\leq~ \Ex{ \hat \Delta(\tau)^2 } ~\leq~ \Ex{ \hat \Delta(\tau) }~\leq~ 32 \epsilon.\label{eq:dovomi}
\end{align}

Now, combining \eqref{eq:avali} and \eqref{eq:dovomi}, we can apply
Bernstein inequality for the random variables $\hat \Delta^{(i)}(\tau)$. We have: 
\begin{align}
\Pr{|\hat \Delta(\tau) - \Delta(\tau)| \geq \epsilon} \leq 2\exp(-\frac{N\epsilon^2}{2\textsf{Var}[ \hat \Delta(\tau) ] + \frac{2}{3} \epsilon}) = 2 T^{-\frac{3C}{194}}.\label{eq:bernsteinbound}
\end{align}
 Therefore, $|\hat \Delta(\tau) - \Delta(\tau)| < \epsilon$ holds with probability $1 - T^{-12}$ when $C = 1000$.

Notice that we only prove the bound for a single $\tau$. To strengthen this concentration bound to hold simultaneously for all $\tau$ and $[\lcb, \ucb]$, we take a union over appropriate cover sets. In particular, consider $\mathcal C$ as a discretization of the interval $[\lcb,\ucb]$ with accuracy $1/T$. To be able to exploit the high probability argument for the elements inside the cover for the ones outside, we need to show that $\Delta$ is Lipschitz with respect to $\tau$, $\ucb$ and $\lcb$.

For $\Delta$ function, we have $|\Delta(\tau) - \Delta(\tau')| \leq 2|\tau - \tau'|$ since
\begin{align*}
    \Big|\Delta(\tau) - \Delta(\tau')\Big| & = \Big|\int_{\tau}^{\tau'} (F_1(u) - F_1(x))(F_2(x) - 1)dx\Big| + \Big|\int_{\tau}^{\tau'} (F_1(u) - F_1(x))(F_2(x) - 1)dx\Big| \\
    & \leq 2|\tau - \tau'|.
\end{align*}
It is easy to see that the same Lipschitz bound also holds for $\hat \Delta$. 

Now for an arbitrary $\tau' \in [\lcb,\ucb]$, if we consider the closest $\tau$ to it in $\mathcal C$, we have $|\tau' - \tau| \leq \frac{1}{T}$. Then, using the Lipschitz constant of $\Delta$ and $\hat \Delta$:
\begin{align}
    \Big| \Delta(\tau) - \Delta(\tau')\Big| ~\leq \frac{2}{T} \qquad \text{and} \qquad \Big| \hat \Delta(\tau) - \hat \Delta(\tau')\Big| \leq \frac{2}{T}.~\label{eq:converguarantees}
\end{align}

Now we apply a union bound over the events $|\hat \Delta(\tau) - \Delta(\tau)| < \epsilon$ for all $\tau \in \mathcal C$. Since running over all possibilities of $|C| \leq T$, after taking a union bound we know that all of these events happen simultaneously with probability at least $1 - T^{-11}$.  We then have for $\tau'$ and its closest element $\tau$ in $\mathcal C$:
\begin{align}
      |\Delta(\tau) - \Delta(\tau')| + |\hat \Delta(\tau) - \hat \Delta(\tau')| ~\leq~ 4|\tau - \tau'| ~\leq~ \frac{4}{T}.\label{eq:coverbounds}
\end{align}
We simply upper-bound $\frac{4}{T}$ by $\epsilon$. This must be true because $\epsilon \geq T^{-\frac{1}{2}} = \omega(\frac{1}{T})$. Then, combining the bound in~\eqref{eq:coverbounds} with~\eqref{eq:bernsteinbound} implies
$|\hat \Delta(\tau') - \Delta(\tau')| \leq 2\epsilon$ holds with probability $1 - T^{-11}$ for all $\tau \in [\lcb, \ucb]$.

Next, we bound $|\hat{R}_\lcb - R(\lcb)|$ and $|\hat{R}_\ucb - R(\ucb)|$. For $|\hat{R}_\lcb - R(\lcb)|$, Notice that $\hat{R}_\lcb$ is an estimate of $R(\lcb)$ with $N = C\cdot \frac{\log T}{\epsilon^2}$ samples, and the reward of each sample falls in $[-1, 1]$. By Hoeffding's Inequality (\Cref{Hoeffding}), the probability that $|\hat{R}_\lcb - R(\lcb)| > \epsilon$ is bounded by $2\exp(-2N\epsilon^2/4) = 2T^{-C/2}$. So, with probability $1 - T^{-11}$ $|\hat{R}_\lcb - R(\lcb)| \leq  \epsilon$ when $C>100$. The bound for $|\hat{R}_\ucb - R(\ucb)|$ is identical. 
Finally, combining three parts with union bound finishes the proof.
\end{proof}

Finally, we have the tools to prove \Cref{mainthmpan}:

\begin{proof}[Proof of \Cref{mainthmpan}]
We will assume that $|\hat{\delta}(\tau) - \delta(\tau)| \leq 4\epsilon$, which is true with probability $1 - T^{-10}$ by \Cref{AccBoundPan}.

Observe that $\hat{\delta}(\tau)$ is a monotone increasing function, because $\hat{\delta}'(\tau) = \hat{\Delta}'(\tau) = (\hat{F}_1(\ucb) - \hat{F}_1(\lcb))(1 - \hat{F}_2(\tau) \geq 0$. Therefore, according to the definition of $\lcb'$ and $\ucb'$, we have $[\lcb', \ucb'] = \{\tau \in [\lcb, \ucb]: |\hat{\delta}(\tau)| \leq 4\epsilon\}$. Now, we can use this property to prove two statements separately:

For the statement that $\tau^* \in [\lcb', \ucb']$, notice that $\delta(\tau^*) = 0$. According to \Cref{AccBoundPan}, $|\hat{\delta}(\tau^*)| \leq 4\epsilon$. Then, since $\tau^* \in [\lcb, \ucb]$ and $|\hat{\delta}(\tau^*)| \leq 4\epsilon$, there must be $\tau^* \in [\lcb', \ucb']$, because $[\lcb', \ucb'] = \{\tau \in [\lcb, \ucb]: |\hat{\delta}(\tau)| \leq 4\epsilon\}$. 

Next, we prove that $|\delta(\tau)| \leq 8\epsilon$ for all $\tau \in [\lcb, \ucb]$. This is true because $[\lcb', \ucb'] = \{\tau \in [\lcb, \ucb]: |\hat{\delta}(\tau)| \leq 4\epsilon\}$, and we have $|\hat{\delta}(\tau) - \delta(\tau)| \leq 4\epsilon$ from \Cref{AccBoundPan}. Therefore, $|\delta(\tau)| \leq |\hat{\delta}(\tau)| + |\hat{\delta}(\tau) - \delta(\tau)| \leq 8\epsilon$ for all $\tau \in [\lcb', \ucb']$.

Finally, the bound $R(\tau^*) - R(\tau) \leq 8\epsilon$ directly follows \Cref{LossBoundPan} and  that $|\delta(\tau)| \leq 8\epsilon$.
\end{proof}

\subsubsection{Proof of \Cref{PB2Regret}}
 To prove \Cref{PB2Regret}, we need to first give an initialization algorithm such that its output should satisfy the conditions listed in \Cref{mainthmpan}. Formally, we have the following lemma:

\begin{Lemma}
\label{PB2Init}
After running no more than $1000 \sqrt{T}\log T$ samples from $\D_1$ and $\D_2$, with probability $1 - T^{-10}$ we can output an \emph{initial interval} $[\lcb, \ucb]$ that satisfies $|g(\tau)| \leq T^{-1/4}$ and $ \tau^* \in [\lcb, \ucb]$.

\end{Lemma}

\begin{proof}
We first run $1000 \sqrt{T}\log T$ extra samples for $X_2$ and calculate an estimate $\hat F_2(x)$. We can show that $|\hat F_2(x) - F_2(x)| \leq \frac{1}{2}T^{-\frac{1}{4}}$ with probability $1 - T^{-10}$: After running $N = C\cdot \sqrt{T} \log T$ samples, the DKW inequality (\Cref{DKW}) shows that $\Pr{|\hat F_2(x) - F_2(x)| > \varepsilon = \frac{1}{2}T^{-\frac{1}{4}}} \leq 2\exp(-2N\varepsilon^2) = 2T^{-C/2}$. Then, with probability $1 - T^{-10}$, we have $|\hat F_2(x) - F_2(x)| \leq \frac{1}{2}T^{-\frac{1}{4}}$ simultaneously holds for all $x \in [0, 1]$ when $C>100$. In the following proof, we assume this accuracy bound always holds. 

Next, we calculate $\hat g(\tau)$ by replacing $F_2(x)$ with $\hat F_2(x)$ in \eqref{eq:gainPBTwo}. When $|\hat F_2(x) - F_2(x)| \leq \frac{1}{2}T^{-\frac{1}{4}}$ holds simultaneously for all $x \in [0, 1]$, we have $|\hat g(\tau) - g(\tau)| \leq \int_\tau^1 |\hat F_2(x) - F_2(x)| \leq \frac{1}{2}T^{-\frac{1}{4}}$. Then, we let $[\lcb, \ucb] := \{\tau: |\hat g(\tau)| \leq \frac{1}{2} T^{-\frac{1}{4}}\}$. Since $\hat g'(\tau) = \hat{F}_2(\tau) - 1 \leq 0$, function $\hat g(\tau)$ is a non-increasing. So, the set $\{\tau: |\hat g(\tau)| \leq \frac{1}{2} T^{-\frac{1}{4}}\}$ must form an interval. Besides, notice that $g(\tau^*) = 0$, which means $|\hat g(\tau^*)| \leq \frac{1}{2} T^{-\frac{1}{4}}$, so we must have $\tau^* \in [\lcb, \ucb]$. Furthermore, for every $\tau \in [\lcb, \ucb]$, $|g(\tau)| \leq |\hat g(\tau)| + |\hat g(\tau) - g(\tau)| \leq T^{-\frac{1}{4}}$, which finishes the proof.
\end{proof}

Now, we are ready to prove \Cref{PB2Regret}:

\begin{proof}[Proof of \Cref{PB2Regret}]

For the core part of the algorithm, we run \Cref{Doubling} and then use \Cref{DoublingBound} to bound the regret. To run \Cref{Doubling}, we let the constraints to mean that the threshold played in each round is inside the interval $[\lcb,\ucb]$ given by \Cref{ISAPan}. Besides, we require the conditions listed in \Cref{mainthmpan} hold (with high probability). We discuss them separately:

\begin{itemize}
    \item $|g(\tau)| \leq T^{-\frac{1}{4}}$ for all $\tau \in [\lcb, \ucb]$: This is guaranteed by \Cref{PB2Init}.
    \item $\tau^* \in [\lcb, \ucb]$: For the first phase, this is guaranteed by \Cref{PB2Init}. Starting from the second phase, this is from \Cref{mainthmpan} of the previous phase.
    \item $|\delta(\tau)| \leq 16\epsilon$: For the first phase, this is true because $\epsilon_1 = 1$. Starting from the second phase, this is from \Cref{mainthmpan} of the previous phase. Notice that the statement in \Cref{mainthmpan} is a little bit different: It guarantees that $|\delta(\tau)| \leq 8\epsilon$ with respect to the $[\lcb, \ucb]$ and $\epsilon$ from the previous phase. When switching to the new phase, notice that $F_2(\ucb') - F_2(\lcb') \leq F_2(\ucb) - F_2(\lcb)$, which means $|\delta(\tau)|$ drops when switching to the new phases. Besides, the parameter $\epsilon_{new}$ in the new phase is exactly $\frac{1}{2}\epsilon_{old}$. Combining these two differences shows that $|\delta(\tau)| \leq 16\epsilon$ holds in the new phase.
\end{itemize}

Therefore, \Cref{ISAPan} satisfies algorithm $\alg$ in \Cref{DoublingBound}. Applying \Cref{DoublingBound} gives the $O(\sqrt{T} \log T)$ regret bound.

Besides, we also run samples for initialization and constructing CDF estimates for \Cref{ISAPan}. These are not coverd by \Cref{DoublingBound}. For the initialization, \Cref{PB2Init} states that $\Theta(\sqrt{T} \log T)$ rounds are sufficient. So the regret from the initialization is $O(\sqrt{T} \log T)$. For constructing $\hat F_1(x)$ and $\hat F_2(x)$, assume we run $k$ phases, then the total number of samples is 
\begin{align*}
    \sum_{i = 1}^k \Theta(\frac{\log T}{\epsilon_i}) = O(\sqrt{T} \log T).
\end{align*}
Combining three parts finishes the proof.

\end{proof}

%
%

\section{Missing Proofs from Section 4}\label{sec:appendixSec4}

\subsection{Proof of \Cref{PBInit}}\label{sec:appendixSec4Lemma1}

\begin{proof}

We first prove the lemma for a single $i$. For $[\lcb_i, \ucb_i]$, we run $C \cdot \sqrt{T}\log T$ extra samples for $X_i$ with $C = 1000$, and calculate an estimate $\hat F_i(x)$. We can show that $|\hat F_i(x) - F_i(x)| \leq \frac{1}{2}T^{-\frac{1}{4}}$ with probability $1 - T^{-11}$: After running $N = C\cdot \sqrt{T} \log T$ samples, the DKW inequality (\Cref{DKW}) shows that $\Pr{|\hat F_i(x) - F_i(x)| > \varepsilon = \frac{1}{2}T^{-\frac{1}{4}}}\leq2\exp(-2N\varepsilon^2) = 2T^{-C/2}$. Then with probability $1 - T^{-11}$, we have $|\hat F_i(x) - F_i(x)| \leq \frac{1}{2}T^{-\frac{1}{4}}$ holds for every $x \in [0, 1]$ when $C > 100$. In the following proof, we assume this accuracy bound always holds.  By the union bound over all $i \in [n]$, the whole proof succeeds with probability $1 - T^{-10}$.

Next, we calculate $\hat g_i(\tau)$ by replacing $F_i(x)$ with $\hat F_i(x)$ in \eqref{eq:gainPBGen}. When $|\hat F_i(x) - F_i(x)| \leq T^{-\frac{1}{4}}$ holds for all $x \in [0, 1]$, we have $|\hat g_i(\tau) - g_i(\tau)| \leq \int_\tau^1 |\hat F_i(x) - F_i(x)| \leq \frac{1}{2}T^{-\frac{1}{4}}$. Then, we let $[\lcb_i, \ucb_i] := \{\tau: |\hat g_i(\tau)| \leq \frac{1}{2} T^{-\frac{1}{4}}\}$. Since $\hat g'(\tau) = \hat{F}_i(\tau) - 1 \leq 0$, which means $\hat g_i(\tau)$ is a decreasing function, then the set $\{\tau: |\hat g_i(\tau)| \leq \frac{1}{2} T^{-\frac{1}{4}}\}$ must form an interval. Besides, notice that $g_i(\tau^*) = 0$, which means $|\hat g_i(\tau^*)| \leq \frac{1}{2} T^{-\frac{1}{4}}$, so there must be $\tau^*_i = \sigma_i \in [\lcb_i, \ucb_i]$. Furthermore, for every $\tau \in [\lcb_i, \ucb_i]$, $|g_i(\tau)| \leq |\hat g_i(\tau)| + |\hat g_i(\tau) - g_i(\tau)| \leq T^{-\frac{1}{4}}$.

Finally, combining the statements for all $n$ intervals finishes the proof. 
\end{proof}

\subsection{Proof of \Cref{PBGenDisAcc}}\label{sec:appendixSec4Lemma2}
\begin{proof}
We first show that $|\hat{F}_i(x) - F_i(x)| \leq \frac{\sqrt{\epsilon}}{2n}$ with probability $1 - T^{-13}$ with $N = C\cdot \frac{n^2\log T}{\epsilon}$ samples, where $C$ is set to be $1000$. Using DKW inequality (\Cref{DKW}), we have $\Pr{|\hat{F}_i(x) - F_i(x)| >  \frac{\sqrt{\epsilon}}{2n}} \leq 2\exp(-2N\frac{\epsilon}{4n^2}) = 2T^{-C/4}$. So the bound holds with probability $1 - T^{-13}$ when $C = 1000$. By the union bound, with probability $1 - T^{-12}$ we have $|\hat{F}_i(x) - F_i(x)| \leq \frac{\sqrt{\epsilon}}{2n}$ holds for every $i \in [n]$. Then, for the accuracy of $\prod_{i \in S} F_i(x)$, we have $ \big((1 - \frac{\sqrt{\epsilon}}{2n})^n - 1\big) \leq \prod_{i \in S} \hat{F}_i(x) -\prod_{i \in S} F_i(x) \leq \big((1 + \frac{\sqrt{\epsilon}}{2n})^n - 1\big)$. For the lower bound, we have $(1 - \frac{\sqrt{\epsilon}}{2n})^n - 1 \geq 1 - \frac{\sqrt{\epsilon}}{2} - 1 > -\sqrt{\epsilon}$. For the upper bound, we have $(1 + \frac{\sqrt{\epsilon}}{2n})^n - 1 \leq \exp(\frac{\sqrt{\epsilon}}{2n} \cdot n) - 1 \leq 1 + 2\cdot \frac{\sqrt{\epsilon}}{2}  - 1 = \sqrt{\epsilon}$. Combining two bounds finishes the proof.
\end{proof}

\subsection{Proof of \Cref{AccBoundPBGen}}\label{sec:appendixSec4Lemma3}

\begin{proof}
 
 Since $\delta_i(\tau) = \Delta_i(\tau) - (R_\ucb - R_\lcb)$, there are three parts in $\delta_i(\tau)$. We show that the accuracy of each part is bounded by $\frac{\epsilon}{3}$ with probability $1 - T^{13}$, then taking a union bound over three accuracy bounds gives \Cref{AccBoundPBGen}.

 First, similar to the derivation in Equation~\eqref{eq:pandoratwoderivation} we bound the magnitude of the $\Delta_i$ function:
 \begin{align}
    |\Delta_i(\tau)| & ~\leq~  (F_{\pi,i}(\ucb) - F_{\pi,i}(\tau))\int_{\tau}^{\ucb}(1 - F_i(x))dx + 
    (F_{\pi,i}(\tau) - F_{\pi,i}(\lcb))\int_{\lcb}^\tau (1 - F_i(x))dx\notag\\
    & ~=~ (F_{\pi,i}(\ucb) - F_{\pi,i}(\tau))(g_i(\tau) - g_i(\ucb)) + (F_{\pi,i}(\tau) - F_{\pi,i}(\lcb))(g_i(\lcb) - g_i(\tau))\notag\\
    & ~\leq~ (F_{\pi,i}(\ucb) - F_{\pi,i}(\lcb))(g_i(\lcb) - g_i(\ucb)) \notag\\
    &~\leq~ |\delta_{\pi, \ucb, \lcb, i}(\lcb)| + |\delta_{\pi, \ucb, \lcb, i}(\ucb)| ~\leq~ 12\epsilon,\label{eq:newbound}
\end{align}
where we use the bound $|\delta_{\pi', \ucb', \lcb', i}(\tau)| \leq 6\epsilon$ in \Cref{MoveOP}.

Next, we hope to propose an estimator $\hat \Delta_i^{(k)}(\tau)$ for the $\Delta_i$ function which uses  $N = C \cdot \frac{\log T}{\epsilon}$ samples for $C = 10^5$. For $k \in [N]$, define
\begin{align*}
    \hat \Delta_i^{(k)}(\tau)= \int_{\tau}^u (\hat F^{(k)}_{\pi,i}(u) - \hat F^{(k)}_{\pi,i}(x))(\hat F^{(k)}_i(x) - 1)dx - 
    \int_{\ell}^\tau (\hat F^{(k)}_{\pi,i}(x) - \hat F^{(k)}_{\pi,i}(\ell))(\hat F^{(k)}_i(x) - 1)dx,
\end{align*}
where $\hat F^{(k)}_{\pi,i}(.)$ and $\hat F^{(k)}_i(.)$ are simple threshold functions at the $i$th initial sample, which are estimates for the densities $F_{\pi,i}$ and $F_i$, respectively. This definition  implies
$    \hat \Delta_i(\tau) = \frac{1}{N} \sum_{k\in [N]} \hat \Delta_i^{(k)}(\tau)$,
and Equation~\eqref{eq:newbound} implies
\begin{align}
    \Ex{\hat \Delta_i^{(k)}(\tau) }= \Delta_i(\tau) ~\leq~ 12\epsilon.\label{eq:expbound2}
\end{align}

Now it is easy to see that $\hat F^{(k)}_{\pi,i}(x) - \hat F^{(k)}_{\pi,i}(y)$ is a Bernoulli random variable which are one if and only if the maximum value obtained from $X_{\pi(1)},\dots,X_{\pi(\pi^{-1}(i)-1)}$ is in~$[\lcb_i, \ucb_i]$. In particular, this implies that $\hat \Delta_i^{(k)}(\tau)$ is bounded by $1$ since
\begin{align}
    |\hat \Delta_i^{(k)}(\tau)| & \leq \int_{\tau}^u \Big|(\hat F^{(k)}_{\pi,i}(u) - \hat F^{(k)}_{\pi,i}(x))(\hat F^{(k)}_i(x) - 1)\Big|dx - 
    \int_{\ell}^\tau \Big|(\hat F^{(k)}_{\pi,i}(x) - \hat F^{(k)}_{\pi,i}(\ell))(\hat F^{(k)}_i(x) - 1)\Big|dx \notag \\
    & \leq \int_{\tau}^u 1dx + 
    \int_{\ell}^\tau 1dx ~=~ \ucb - \lcb ~\leq~ 1.\label{eq:avali2}
\end{align}
Combining Equations~\eqref{eq:expbound2} and~\eqref{eq:avali2}, we have the variance bound:
\begin{align}
    \textsf{Var}[ \hat\Delta_i(\tau) ] ~\leq~ \Ex{ \hat\Delta_i(\tau)^2 } ~\leq~ \Ex{ \hat\Delta_i(\tau) }~\leq~ 12\epsilon.\label{eq:dovomi2}
\end{align}

Hence, using Bernstein inequality, we have
\begin{align}
\Pr{|\hat \Delta_i(\tau) - \Delta_i(\tau)| \geq  \frac{\epsilon}{12}} ~\leq~ 2\exp(-\frac{N\epsilon^2/144}{2\textsf{Var}[ \hat \Delta_i(\tau) ] + \frac{2}{3} \frac{\epsilon}{12}}) ~=~ 2 T^{-\frac{C}{3464}}. \label{eq:bernstein2}
\end{align}
 Therefore, $|\hat \Delta_i(\tau) - \Delta_i(\tau)| \leq \frac{\epsilon}{12}$ holds with probability $1 - T^{-14}$ when $C = 10^5$.

 The bound above is only for a single $\tau$. To give the bound for a whole interval, we discretize $[\lcb_i, \ucb_i]$ uniformly into a discrete set $\mathcal{C}$ and make sure that each pair of adjacent $\tau, \tau' \in \mathcal{C}$ follows $|\tau - \tau'| \leq \frac{1}{T}$. Then, there must be $|\mathcal{C}| \leq T$ and the union bound implies $|\hat \Delta_i(\tau) - \Delta_i(\tau)| \leq \frac{\epsilon}{12}$ holds with probability $1 - T^{-13}$ for all $\tau \in \mathcal{C}$.

 Next, we bound the Lipschitz constant of $\Delta_i$ (and similarly $\hat \Delta_i$):
\begin{align*}
    \Big|\Delta_i(\tau) - \Delta_i(\tau')\Big| & = \Big|\int_{\tau}^{\tau'} (F_{\pi,i}(u) - F_{\pi,i}(x))(F_i(x) - 1)dx\Big| + \Big|\int_{\tau}^{\tau'} (F_{\pi,i}(u) - F_{\pi,i}(x))(F_i(x) - 1)dx\Big| \\
    & \leq 2|\tau - \tau'|.
\end{align*}
 Finally, for every $\tau \in [\lcb_i, \ucb_i]$, let $\tau'$ be the closest value in $\mathcal{C}$. Then, we have:
 \begin{align*}
     \textstyle \Big|\hat \Delta_i(\tau) - \Delta_i(\tau)\Big| \leq \Big|\Delta_i(\tau) - \Delta_i(\tau')\Big| + \Big|\hat \Delta_i(\tau') - \Delta_i(\tau')\Big| + \Big|\hat \Delta_i(\tau') - \hat \Delta_i(\tau)\Big| \leq \frac{\epsilon}{12} + \frac{4}{T} \leq  \frac{\epsilon}{6},
 \end{align*}
 where the last inequality is true because $\epsilon > T^{-\frac{1}{2}} = \omega(\frac{1}{T})$. Therefore, with probability $1 - T^{-13}$, we have $\Big|\hat \Delta_i(\tau) - \Delta_i(\tau)\Big| \leq \frac{\epsilon}{6}$  for all $\tau \in [\lcb_i, \ucb_i]$ simultaneously.

Next, we use Hoeffding's Inequality (\Cref{Hoeffding}) to bound the accuracy of $|R_\lcb - \hat{R}_\lcb|$. In each round, the reward falls in [-0.5,0.5]. Then, after running $N = C\cdot \epsilon^{-2}\log T$~ samples, we have $\Pr{|R_\lcb - \hat{R}_\lcb| >  \frac{\epsilon}{6}} \leq 2\exp(-2N\epsilon^2/36) = 2 T^{-C/18} $. Therefore, $|R_\lcb - \hat{R}_\lcb| \leq \frac{\epsilon}{6}$ with probability $1 - T^{-13}$ when $C > 1000$. Besides, the proof for $|R_\ucb - \hat{R}_\ucb|$ is identical. Combining three parts with union bound finishes the proof.
\end{proof}

\subsection{Proof of \Cref{Lemma:Clever}}\label{sec:appendixSec4Lemma4}

In this section, we show that \Cref{ApproxClever}  finds an approximately clever threshold setting. We first  introduce the following lemma:

\begin{algorithm}[tbh]
\caption{Finding Approximately Clever Threshold}
\label{ApproxClever}
\KwIn{$(I, S)$, $m$, $i$, $\hat{F}_1(x),...,\hat{F}_n(x)$}
\For{$P \subseteq [n]$}
{
\lIf{$\exists k: (k, i) \in S \land k \notin P$ or $\exists k: (i, k) \in S \land k \in P$ or $\exists k,j: (k, j) \in S \land k \notin P \land j \in P$
}
{
Skip this $P$
}
For $k \in P$, let $\tau_k = \ucb_k$ \\
For $k \in [n] \setminus (T \cup \{i\})$, let $\tau_k = \lcb_k$ \\
Let $\ucb_T = \min\{\ucb_i, \tau_{k: k \in T}\}$, $\lcb_T = \max\{\lcb_i, \tau_{k: k \notin (T \cup \{i\})}\}$ \\
Set partial setting $\pi_T$ be: Let $\tau_i \in [\lcb_T, \ucb_T]$. $\pi_T$ sorts the thresholds in a decreasing order. Break the ties according to the constraints in $S$. \\
Let $\hat{F}_{\pi_T, i}(x) = \prod_{k \in T} \hat{F}_k(x)$. \\
Calculate $q_T := \hat{F}_{\pi_T, i}(\ucb_T) - \hat{F}_{\pi_T, i}(\lcb_T)$
}
Let $T^* = \arg \max q_T$. \\
\KwOut{$\pi_{T^*}, \lcb_{T^*}, \ucb_{T^*}, \hat{F}_{\pi_{T^*}, i}$.}
\end{algorithm}

\begin{Lemma}
\Cref{ApproxClever} calculates a \textit{clever} threshold setting, up to an $4\sqrt{\epsilon}$ additive error. The running time of \Cref{ApproxClever} is $O(n\cdot 2^n)$.

\end{Lemma}

\begin{proof}
The goal of a \textit{clever} threshold setting is to maximize $F_{\pi, i}(\ucb) - F_{\pi, i}(\lcb)$. Fix $i$. When the set $P$, which represents the distributions in front of $X_i$ is determined, the function $F_{\pi, i}(x)$ is fixed. Therefore, to maximize $F_{\pi, i}(\ucb) - F_{\pi, i}(\lcb)$, we should maximize $\ucb$ and minimize $\lcb$. This can be achieved by maximizing the thresholds in $P$ and minimizing the thresholds in $[n] \setminus (P \cup \{i\})$, which is exactly lines 8 and 9 in \Cref{PBISA}. Then, after enumerating all valid subsets $P$, we can find a setting that maximizes $F_{\pi, i}(\ucb) - F_{\pi, i}(\lcb)$.

There is one   missing detail: we only know the value of $\hat{F}_i(x)$. From \Cref{PBGenDisAcc}, we know $\hat{F}_i(x)$ is an estimate of $F_i(x)$ with accuracy $\sqrt{\epsilon}$. Therefore, $\max_{\pi} \hat{F}_{\pi, i}(\ucb) - \hat{F}_{\pi, i}(\lcb)$ is at most $2\sqrt{\epsilon}$  different from $\max_{\pi} F_{\pi, i}(\ucb) - F_{\pi, i}(\lcb)$. After getting $\pi' = \arg \max_{\pi} \hat{F}_{\pi, i}(\ucb) - \hat{F}_{\pi, i}(\lcb)$, the real value of $F_{\pi', i}(\ucb) - F_{\pi', i}(\lcb)$ is at most $2\sqrt{\epsilon}$ different from $\hat F_{\pi', i}(\ucb) - \hat F_{\pi', i}(\lcb)$. Combining two errors proves the $4\sqrt{\epsilon}$ error bound.

For the running time of \Cref{ApproxClever}, we need to enumerate a subset $S$, then calculate the corresponding $F_{\pi, i}(x)$ function. So the running time is $O(n\cdot 2^n)$.
\end{proof}

\end{document}